%% file: main.tex
\documentclass[aps,prxquantum,longbibliography,notitlepage,twocolumn,superscriptaddress,nofootinbib]{revtex4-2}
\usepackage{bbm}
\usepackage{mathrsfs}
\usepackage{xcolor}
\usepackage{graphicx}
\usepackage{amsfonts}
\usepackage{amsthm}
\usepackage[figuresright]{rotating}
\usepackage{amssymb}
\usepackage{amsmath}
\usepackage{dcolumn}
\usepackage{physics}
\usepackage{float}
\usepackage{bm}
\usepackage{verbatim}
\usepackage[normalem]{ulem}
\usepackage[ruled,vlined,linesnumbered]{algorithm2e}
\usepackage{setspace}
\usepackage[colorlinks,linkcolor=blue,anchorcolor=blue,citecolor=blue,urlcolor=blue]{hyperref}
\usepackage{tikz}
\usetikzlibrary{arrows.meta,calc,shapes.misc,backgrounds}
\newtheorem{theorem}{Theorem}
\newtheorem{lemma}[theorem]{Lemma}
\newtheorem*{lemma*}{}

\newcommand*{\horzbar}{\rule[.5ex]{2.5ex}{0.5pt}}

\begin{document}

\title{Faster Quantum Monte Carlo Simulation by Random Compilation}

\author{John M. Martyn} 
\affiliation{Physical and Computational Sciences, Pacific Northwest National Laboratory, Richland, WA 99354, USA}
\affiliation{Harvard Quantum Initiative, Harvard University, Cambridge, MA 02138, USA}
\affiliation{Center for Theoretical Physics, Massachusetts Institute of Technology, Cambridge, MA 02139, USA}

\author{Joshua Lin} 
\affiliation{Physics Division, Argonne National Laboratory, Lemont, IL 60439, USA}
\affiliation{Center for Theoretical Physics, Massachusetts Institute of Technology, Cambridge, MA 02139, USA}

\author{Neill C. Warrington} 
\affiliation{Center for Theoretical Physics, Massachusetts Institute of Technology, Cambridge, MA 02139, USA}

\author{Isaac L. Chuang}
\affiliation{Department of Physics, Massachusetts Institute of Technology, Cambridge, MA 02139, USA}
\affiliation{Department of Electrical Engineering and Computer Science,
Massachusetts Institute of Technology, Cambridge, MA 02139, USA}

\author{Andrew J. Daley}
\affiliation{Clarendon Laboratory, University of Oxford, Parks Road, Oxford OX1 3PU, UK}

\begin{abstract}    
    Quantum Monte Carlo (QMC) algorithms are among the most powerful classical methods for simulating quantum systems, yet their accuracy is often limited by the systematic errors in the approximations used, such as Trotterization. Here we introduce \textit{randomly compiled quantum Monte Carlo} (RC-QMC) as a general framework that suppresses these systematic errors by averaging over a family of approximations rather than relying on a single fixed one. This strategy is grounded in the concept of randomized compiling from quantum computing, which suppresses errors by sampling over quantum gates, at essentially no additional computational cost. Consequently, our framework achieves a computational advantage over standard QMC methods when estimating a target state to a desired level of accuracy. We illustrate this advantage on two key Monte Carlo algorithms: (1) path integral quantum Monte Carlo for estimating thermal states, and (2) the quantum trajectories method for simulating open system dynamics. In aggregate, these results represent a cross-fertilization of quantum and classical algorithms, are readily generalizable to other QMC methods, and suggest wider applications in classical simulation.

\end{abstract}

\maketitle

\section{Introduction}

Quantum Monte Carlo (QMC) algorithms are central computational methods for classically simulating quantum systems and have yielded valuable insights across fields ranging from condensed matter physics to quantum chemistry~\cite{becca2017quantum, zhang201315, rothe2012lattice, Daley_2014, Plenio_1998}. 
While these tools are broadly applicable, they can be computationally expensive, in part due to their reliance on costly approximation schemes, like high-order Trotter decompositions and Hubbard-Stratonovich transformations. This is in addition to the large number of samples required to achieve high precision, and the sign problem afflicting certain systems.

When viewed alongside quantum computing, QMC algorithms share many structural similarities with quantum algorithms for simulation. Both approaches incorporate approximation methods like Trotterization~\cite{Suzuki_1990, Suzuki_1991, Lloyd_1996}, and more fundamentally, both estimate expectation values by averaging over many outcomes. These connections suggest that further tools from quantum algorithms could be repurposed to enhance QMC simulations.

One such tool is \textit{randomized compiling},\footnote{Note that our use of ``randomized compiling'' refers to the concept from quantum computing, and should not be confused with randomized techniques in classical compiler design.} which randomly samples quantum gates to suppress errors~\cite{Campbell_2017, Wallman_2016}. Thanks to its generality, randomized compiling has been employed across various quantum applications, including mitigating gate errors in quantum circuits~\cite{Wallman_2016} and accelerating time evolution algorithms~\cite{Campbell_2019, Cho_2024, pocrnic2023composite}. However, its potential in classical simulation has not yet been fully explored. In this work, we fill this gap by introducing \textit{randomly compiled quantum Monte Carlo} (RC-QMC) as a general framework that uses randomized compiling to improve the efficiency and accuracy of QMC.

\input{figures/result_fig}

At their core, QMC algorithms estimate an observable $\hat{O}$ in a target state $\rho$ by
mapping the expectation value $\langle\hat{O}\rangle := \text{tr}(\hat{O}\rho)$ to a
classically sampled random variable. This is often done using an approximation method, which we label by an index $j$ (e.g., a specific Trotter decomposition). Standard QMC fixes a single such approximation and evaluates $\langle\hat{O}\rangle \approx \langle\hat{O}\rangle_j$ for a chosen $j$.
RC-QMC instead averages over a family of approximations, 
\begin{equation}\label{eq:avg_obs}
    \langle \hat{O} \rangle \approx \sum_j p_j \langle \hat{O} \rangle_j
      = \mathop{\mathbb{E}}_{j \sim p_j}\!\Big[ \langle \hat{O} \rangle_j \Big],
\end{equation}
where $p_j$ is a probability distribution over approximations.
This average is easily incorporated into a QMC algorithm by jointly sampling $j \sim p_j$ alongside the standard estimator. Drawing on theorems from randomized compiling, we prove that this achieves a smaller systematic error than any individual approximation, while incurring a cost that is the weighted average of the approximations---nearly the same as standard QMC. 
We note that this error reduction pertains to systematic error only; it does not mitigate statistical sampling errors nor resolve sign problems in systems where they arise.

As a simple example, it has been shown that randomizing over the ordering of terms in
a Trotter step reduces discretization error~\cite{Childs_2019}. When applied to a QMC algorithm like path integral Monte Carlo via Eq.~\eqref{eq:avg_obs}, the randomly compiled algorithm require fewer Trotter steps to achieve the same level of accuracy as the standard algorithm, thus reducing the computational cost.

Here we will adopt this strategy to develop applications of RC-QMC and use more advanced randomized compiling methods to achieve larger gains. We demonstrate RC-QMC in two settings: (1) \textit{path integral Monte Carlo} for estimating thermal states, and (2) the \textit{quantum trajectories method} for simulating open system dynamics.\footnote{Although the quantum trajectories method is not conventionally considered a QMC algorithm, it shares enough structural similarities to fit into the RC-QMC framework.} An overview of the RC-QMC framework and its main results is provided in Fig.~\ref{fig:rcqmc_overview}.

In these applications we specialize to two distinct randomized compiling methods, and it is important to distinguish them. In randomly compiled path integral Monte Carlo, we incorporate a randomized Trotterization algorithm known as \textit{QDrift}~\cite{Campbell_2019}, which improves error scaling by replacing the standard Trotter error $\mathcal{O}(L^3\beta^3/r^2)$ with $\mathcal{O}(\lambda^2\beta^2/r)$. Here $r$ is the number of Trotter steps, $h_j > 0$ are the Hamiltonian coefficients, and $\lambda = \sum_j h_j$ is their 1-norm. While the standard error grows with the number of Hamiltonian terms $L$, the QDrift error does not, which yields a significant advantage when $\lambda \ll L$. In the randomly compiled quantum trajectories method, we use \textit{randomly corrected Trotterization}~\cite{Cho_2024}, which doubles the effective Trotter order by suppressing the time-step error from $\mathcal{O}(\Delta t^{n+1})$ to $\mathcal{O}(\Delta t^{2n+2})$, where $n$ is the Trotter order. This substantially reduces the number of time steps required for a given accuracy. Both gains are underpinned by the same RC-QMC mechanism, but arise from different randomized compiling methods applied to different algorithmic structures. Beyond these applications, RC-QMC extends to other QMC algorithms and is compatible with other randomization techniques; we comment on these prospects in the conclusion section.

We present our results in a way intended to be accessible to readers from both QMC and quantum computing communities. Sec.~\ref{sec:QMC} reviews core QMC concepts, introduces a general framework for QMC, and reviews both path integral Monte Carlo and the quantum trajectories method. Sec.~\ref{sec:RC_QMC} introduces randomized compiling and develops RC-QMC as a general framework. Secs.~\ref{sec:Application_PIMC} and~\ref{sec:Application_QTM} then demonstrate applications of RC-QMC to path integral Monte Carlo and quantum trajectories, respectively, and can be read independently for practitioners in those areas. Finally, we conclude and discuss future prospects of this work in Sec.~\ref{sec:Conclusions}. In addition, to maintain clarity of notation throughout this work, we provide a reference list of symbols in Appendix~\ref{app:notation}.

\section{Quantum Monte Carlo Algorithms}\label{sec:QMC}

Quantum Monte Carlo is an umbrella term for a variety of classical statistical methods for simulating quantum systems. QMC algorithms work by expressing the expectation value of an observable as an average of a random variable that can be evaluated by classical sampling. Thanks to their generality, QMC algorithms have found broad applications across quantum science, from condensed matter physics to high-energy physics~\cite{becca2017quantum, zhang201315, rothe2012lattice, Evertz_1993, Krzakala_2008, sadoune2022efficient, morningstar2007monte, Carmen_Banuls_2020, Albergo_2019}.  As a representative example, \textit{path integral Monte Carlo}~\cite{Ceperley_1995} equates a thermal expectation value to an average over classical paths:
\begin{equation}\label{eq:PIMC_expression_0}
\begin{aligned}
    \text{tr} \bigg( \hat{O} \frac{e^{-\beta \hat{H}}}{Z} \bigg)
      = \mathop{\mathbb{E}}_{\mathbf{X} \sim e^{-S(\mathbf{X})}/Z}
        \Big[ \langle \mathbf{X}(0) | \hat{O} | \mathbf{X}(\beta) \rangle \Big],
\end{aligned}
\end{equation}
where $\mathbf{X}(\tau)$ is a classical path extending from Euclidean time $\tau=0$ to $\tau=\beta$, $S[\mathbf{X}(\tau)]$ is the Euclidean action, and $Z = \text{tr}(e^{-\beta \hat{H}})$ is the partition function.

This section provides a detailed overview of QMC, establishing the framework and key error scalings that RC-QMC will improve upon. Sec.~\ref{sec:QMC_General} begins by casting QMC within a unified framework tailored to support the development of RC-QMC. Secs.~\ref{subsec:PIMC} and~\ref{sec:QTM} then review path integral Monte Carlo and the quantum trajectories method within that framework, focusing on their computational costs and systematic errors.

\subsection{A General Framework for QMC}\label{sec:QMC_General}
While QMC algorithms vary in their construction and applications, a broad class of them can be unified under the following general framework. Let us consider a QMC algorithm that targets a state $\rho$, e.g. a thermal state or a ground state. In these algorithms, the expectation value of an observable $\hat{O}$ is expressed as a classical average of the following form:
\begin{equation}\label{eq:QMC}
    \langle \hat{O} \rangle = \tr(\hat{O}\rho) = \mathop{\mathbb{E}}_{X \sim \mathcal{P}} \big[ F(\hat{O}, X) \big].
\end{equation}
Here, $X$ is a random variable drawn from a known distribution $\mathcal{P}$, and $F(\hat{O}, X)$ is an \emph{estimator function} that depends on $\hat{O}$ and $X$. In practice, classically evaluating  $F(\hat{O}, X)$ is much less expensive than evaluating $\langle \hat{O} \rangle$ directly, such that $\langle \hat{O} \rangle$ can be estimated as an empirical average. This renders QMC more efficient than brute force computation in many settings.

As an instance of this general framework, in path integral Monte Carlo (Eq.~\eqref{eq:PIMC_expression_0}), the random variable is the configuration of classical paths $X = \mathbf{X}$, sampled from the distribution $\mathcal{P} = e^{-S(\mathbf{X})}/ Z$, and the estimator function is $F(\hat{O},X) = \langle \mathbf{X}(\beta) | \hat{O} | \mathbf{X}( 0 ) \rangle$. Other QMC algorithms also fit into this general framework. In the following subsections, we illustrate this in more detail for path integral Monte Carlo and the quantum trajectories method. 

While Eq.~\eqref{eq:QMC} is formally exact, it can only be approximately realized in practice. In general, the distribution $\mathcal{P}$ cannot be sampled exactly, and instead must be approximated. For instance, the exact path integral Monte Carlo expression of Eq.~\eqref{eq:PIMC_expression_0} holds only in the limit of a continuous path. In actual implementations of this algorithm, this expression is approximated by sampling over discrete paths, where the distribution $e^{-S(\mathbf{X})}/Z$ is approximated using a discretized action. 
Accordingly, in practice QMC algorithms replace $\mathcal{P}$ with an \emph{approximate distribution} $\widetilde{\mathcal{P}}$, yielding an expression
\begin{equation}\label{eq:QMC_approx}
    \langle \hat{O} \rangle = \tr(\hat{O}\rho) \approx \mathop{\mathbb{E}}_{X \sim \widetilde{\mathcal{P}} } \big[ F(\hat{O}, X) \big] =: \langle \hat{O} \rangle_{\widetilde{\mathcal{P}}} ,
\end{equation}
where $\langle \hat{O} \rangle_{\widetilde{\mathcal{P}}} $ denotes the observable estimated using $\widetilde{\mathcal{P}}$. This approximation introduces a systematic error
\begin{equation}\label{eq:QMC_sys_err}
     \big|\langle \hat{O} \rangle - \langle \hat{O} \rangle_{\widetilde{\mathcal{P}}} \big|. 
\end{equation}
Consequently, both the accuracy of the QMC algorithm and its computational cost depend on the chosen approximate distribution $\widetilde{\mathcal{P}}$. For example, the discretization used to realize path integral Monte Carlo affects both: a finer discretization achieves smaller systematic errors at the price of an increased cost.

It is important to distinguish this systematic error from the statistical error also suffered in QMC algorithms. Statistical error arises from estimating the average $\mathbb{E}_{X \sim \widetilde{\mathcal{P}}} [F(\hat{O},X)]$ as an empirical mean, and decays as $\mathcal{O}(1/\sqrt{N_{\text{s}}})$ when estimated with $N_{\text{s}}$ samples. In contrast, the systematic error arises from using an approximate distribution $\widetilde{\mathcal{P}}$ and persists even as the statistical error vanishes, thereby limiting the accuracy of QMC.

\subsection{Example: Path Integral Monte Carlo}\label{subsec:PIMC}
To illustrate the general QMC framework above, let us overview path integral Monte Carlo, with a focus on its systematic error and corresponding computational cost.
As mentioned previously, path integral Monte Carlo aims to estimate a thermal expectation value $\text{tr}\big( \hat{O} e^{-\beta \hat{H}} /Z\big)$ for an operator $\hat{O}$, Hamiltonian $\hat{H}$ and inverse temperature $\beta$. The standard approach is to divide $\beta$ into $r$ steps of size $ \Delta\beta := \beta/r$ and insert resolutions of the identity between each time step. Here we will consider a system of $N$ particles, and resolve the identity as $I = \sum_{\mathbf{x}} |\mathbf{x}\rangle \langle \mathbf{x} |$ for a basis $|\mathbf{x}\rangle$, where $\mathbf{x}$ is an $N$-dimensional row vector representing the positions of the particles. This yields
\begin{equation}\label{eq:PIMC_O_1}
\begin{aligned}
    &\frac{1}{Z}\tr(e^{-\beta \hat{H}} \hat{O} ) = \frac{1}{Z}\tr((e^{-\Delta \beta \hat{H}})^r \hat{O} ) \\
    &  \ \ = \frac{1}{Z} \sum_{\mathbf{x}_0, \mathbf{x}_1, ..., \mathbf{x}_r } \langle \mathbf{x}_0 | e^{-\Delta \beta \hat{H}} | \mathbf{x}_1 \rangle \langle \mathbf{x}_1 | e^{-\Delta \beta \hat{H}} | \mathbf{x}_2 \rangle ... \\
    & \qquad \qquad \qquad \qquad \qquad \times \langle \mathbf{x}_{r-1} | e^{-\Delta \beta \hat{H} } | \mathbf{x}_r \rangle \langle \mathbf{x}_r | \hat{O} | \mathbf{x}_0 \rangle , 
\end{aligned}
\end{equation}
which has re-expressed the expectation value as a sum of a product of matrix elements. For an operator $\hat{O}$ that is diagonal in position space, $\langle \mathbf{x}_r | \hat{O} | \mathbf{x}_0 \rangle = \delta_{\mathbf{x}_{r}=\mathbf{x}_{0}} \cdot \langle \mathbf{x}_0 | \hat{O} | \mathbf{x}_0 \rangle$ and the paths are constrained to be periodic $\mathbf{x}_{r}=\mathbf{x}_{0}$, which we will assume throughout. 

To evaluate these matrix elements, let us specialize to a Hamiltonian $\hat{H} = \sum_{j=1}^L \hat{H}_j$ composed of $L$ terms that decomposes into two mutually non-commuting subsets $\hat{H}=\hat{T}+\hat{V}$. Here, $\hat{T}$ represents the kinetic energy, and $\hat{V}$ represents the potential energy, which we take to be diagonal in position space: $\hat{V}|\mathbf{x}\rangle = V(\mathbf{x}) |\mathbf{x}\rangle $, for a potential energy function $V(\mathbf{x})$. Individual matrix elements are then approximated using (second-order) Trotterization as 
\begin{equation}\label{eq:PIMC_approxs}
\begin{aligned}
    &\langle \mathbf{x}_k | e^{- \Delta \beta \hat{H}} |\mathbf{x}_{k+1} \rangle \\
    & \ \ =
    \langle \mathbf{x}_k | \big[ e^{-\frac{\Delta \beta}{2} \hat{V} } e^{-\Delta \beta \hat{T} } e^{-\frac{\Delta \beta}{2} \hat{V}  } + \mathcal{O}(L^3 \Delta \beta^3) \big]| \mathbf{x}_{k+1} \rangle  \\ 
    & \ \ = e^{- \Delta \beta \big(\frac{1}{2}V(\mathbf{x}_k) + T(\mathbf{x}_k, \mathbf{x}_{k+1}) + \frac{1}{2}V(\mathbf{x}_{k+1}) \big) }  + \mathcal{O}(L^3\Delta \beta^3) .
\end{aligned}
\end{equation}
Here the error $\mathcal{O}(L^3\Delta \beta^3)$ follows from the second-order nested commutators of $\hat{T}$ and $\hat{V}$ being upper bounded by $\mathcal{O}(L^3)$ in general\footnote{Here, we use big $\mathcal{O}$ notation to refer to the limit of large $L$, large $\beta$, and large $r$ (equivalently, a small time step $\Delta \beta = \beta/r$). }~\cite{Childs_2021}. In addition, we have introduced a kinetic energy function $T(\mathbf{x}_k, \mathbf{x}_{k+1})$ defined by the matrix element $\langle \mathbf{x}_k | e^{-\Delta \beta \hat{T} } | \mathbf{x}_{k+1} \rangle =: e^{-\Delta \beta T(\mathbf{x}_k, \mathbf{x}_{k+1})}$. For example, when $\hat{T} = \hat{p}^2/2$ is the standard kinetic energy, this function corresponds to the discretized kinetic energy $T(\mathbf{x}_k, \mathbf{x}_{k+1}) = \frac{1}{2}(\mathbf{x}_k -  \mathbf{x}_{k+1})^2/\Delta\beta^2$.

This expression can be neatly packaged by defining an array of particle paths over time: 
\begin{equation}\label{eq:paths_array}
    \mathbf{X} = 
    \begin{bmatrix}
        \horzbar & \mathbf{x}_0 & \horzbar \\
        \horzbar & \mathbf{x}_1 & \horzbar \\
        \vdots & \vdots & \vdots \\
        \horzbar & \mathbf{x}_{r-1} & \horzbar
    \end{bmatrix} . 
\end{equation}
In this array, the rows $\mathbf{X}_{k,:} = \mathbf{x}_k$ represent the particles' configuration at time step $k$, while the columns $\mathbf{X}_{:,i}$ encode the path of particle $i$. Inserting this into Eq.~\eqref{eq:PIMC_O_1} yields a compact expression for path integral Monte Carlo:
\begin{equation}\label{eq:PIMC_O_2}
\begin{aligned}
    \langle \hat{O} \rangle &= \text{tr} \bigg( \hat{O} \frac{e^{-\beta \hat{H}}}{Z} \bigg)  \\
    &= \mathop{\mathbb{E}}_{ \mathbf{X} \sim e^{-S(\mathbf{X})}}  \Big[ \langle \mathbf{X}_{0,:} | \hat{O} | \mathbf{X}_{0,:} \rangle \Big] + \mathcal{O}\left(\frac{L^3 \beta^3}{r^2} \right) . 
\end{aligned}
\end{equation}
Here, the discrete action is 
\begin{equation}\label{eq:PIMC_Action}
\begin{aligned}
    S(\mathbf{X}) = &\Delta \beta  \sum_{k=0}^{r-1} \Big[  T(\mathbf{X}_{k,:}, \mathbf{X}_{k+1,:}) +  V(\mathbf{X}_{k,:}) \Big] ,  
\end{aligned}
\end{equation}
and the paths $\mathbf{X}$ are sampled from the normalized distribution $\propto e^{-S(\mathbf{X})}$. 
The overall error is $\mathcal{O}(L^3 \beta^3/r^2)$, which arises from Trotterizing over $r$ steps that each suffer error $\mathcal{O}(L^3 \Delta\beta^3) = \mathcal{O}(L^3 \beta^3/r^3)$.

Eq.~\eqref{eq:PIMC_O_2} is the central equation of path integral Monte Carlo: observables can be estimated by sampling many paths $\mathbf{X} \sim e^{-S(\mathbf{X})}$ and evaluating the corresponding empirical average of $\langle \mathbf{X}_{0,:} |\hat{O} |\mathbf{X}_{0,:} \rangle$. The accuracy improves with increasing number of Trotter steps $r$ (equivalently, increasing path length), though at the price of increased computational cost. 
Comparing this with the general QMC formulation of Eq.~\eqref{eq:QMC_approx}, we see that the approximate distribution is $\widetilde{\mathcal{P}} \propto e^{-S(\mathbf{X})} $, the estimator function is $F(\hat{O}, X) = \langle \mathbf{X}_{0,:} | \hat{O} | \mathbf{X}_{0,:} \rangle$, and the resulting systematic error is $\mathcal{O}(L^3 \beta^3/r^2)$. Thus, achieving a systematic error $\epsilon$ requires a path length $r=\mathcal{O} \big( L^{3/2} \beta^{3/2}/\epsilon^{1/2} \big)$, and so the computational cost grows with $\beta $ and $L$.


\subsection{Example: Quantum Trajectories Method}\label{sec:QTM}

Another relevant algorithm that fits into the general QMC framework is the \textit{quantum trajectories method} (also known as the `Monte Carlo wave function method', or the `quantum jump method')~\cite{Dalibard_1992, Dum_1992, Carmichael1993, Molmer_1993, Daley_2014}. This algorithm simulates an open quantum system evolving under the master equation in Lindblad form:
\begin{equation}\label{eq:master_eq_main}
    \frac{d}{dt} \rho = -i[\hat{H},\rho] + \sum_l  \big( {c}_l \rho {c}_l^\dag - \frac{1}{2} ({c}_l^\dag {c}_l \rho + \rho {c}_l^\dag {c}_l) \big) ,
\end{equation}
where ${c}_l$ are the \emph{jump operators} that describe the dissipative dynamics (with dissipative rates absorbed into $c_l$ here for simplicity). The master equation may be conveniently re-expressed as:
\begin{equation}\label{eq:master_eq_reexpression}
    \frac{d}{dt} \rho = -i (\hat{H}_{\text{eff}}\rho - \rho \hat{H}_{\text{eff}}^\dag) + \sum_l c_l \rho c_l^\dag,
\end{equation}
where $ \hat{H}_{\text{eff}} = \hat{H} - \tfrac{i}{2} \sum_l c_l^\dag c_l $ is an effective Hamiltonian, which is in general non-Hermitian.

Eq.~\eqref{eq:master_eq_reexpression} closely resembles the Heisenberg equation for closed system dynamics, except for the non-Hermiticity of $\hat{H}_{\text{eff}}$ and the additional action of the jump operators. Owing to this similarity, this equation possesses an analytic solution where the time-evolved density matrix $\rho(t)$ unravels into an average of pure states. Specifically, each pure state evolves under $\hat{H}_{\text{eff}}$ with jump operators randomly interspersed. We briefly describe this solution below and include a comprehensive review in Appendix~\ref{app:QTM}.

To illustrate this analytic solution, let us specialize to a single jump operator $c_1 = c$ for simplicity, and consider an initial pure state $\rho(0) = |\psi \rangle \langle \psi |$. Then, an observable at time $t$ can be exactly written as an average of expectation values computed in these pure states:
\begin{equation}\label{eq:master_eq_obs}
\begin{aligned}
    \langle \hat{O} \rangle (t) &= \tr(\hat{O} \rho(t)) \\
    &= \mathop{\mathbb{E}}_{|\phi\rangle \sim \langle \phi | \phi \rangle} \Bigg[ \frac{\langle \phi (t|t_{1:m}) | \hat{O} | \phi (t|t_{1:m}) \rangle}{\langle \phi (t|t_{1:m}) | \phi (t|t_{1:m}) \rangle}  \Bigg]. 
\end{aligned}
\end{equation}
Here, $| \phi (t|t_{1:m}) \rangle $ is an unnormalized state that begins in $|\psi\rangle$, evolves under $\hat{H}_{\text{eff}}$ up to time $t$, and experiences $m$ intermittent jumps at times $t_{1:m} := \{ t_1, t_2, ..., t_m \}$:  
\begin{equation}\label{eq:QTM_state}
\begin{aligned}
    & | \phi (t| t_{1:m}) \rangle = \\
    & \qquad  e^{-i\hat{H}_{\text{eff}}(t-t_{m})} c e^{-i\hat{H}_{\text{eff}}(t_m - t_{m-1})} c ... c e^{-i\hat{H}_{\text{eff}}t_1} |\psi \rangle .
\end{aligned}
\end{equation}
In Eq.~\eqref{eq:master_eq_obs}, the notation $|\phi\rangle \sim \langle \phi | \phi\rangle $ means that the states $| \phi(t| t_{1:m}) \rangle$ are sampled proportional to their norm $ \big\| | \phi (t| t_{1:m}) \rangle \big\|^2$. This amounts to averaging over both the number of jumps $m$ and their times $t_{1:m}$.

The quantum trajectories method uses this decomposition to estimate time-evolved observables by simulating many pure states $|\phi\rangle$, known as \textit{trajectories}, and averaging over their corresponding observables according to Eq.~\eqref{eq:master_eq_obs}. This fits directly into the general QMC framework of Eq.~\eqref{eq:QMC}: the random variable is the state $X = |\phi(t|t_{1:m})\rangle$ sampled from its norm $\mathcal{P} = \big\| |\phi(t|t_{1:m})\rangle \big\|^2 $, and the estimator function is the observable evaluated in the corresponding normalized state $F(\hat{O},X) = \frac{\langle \phi (t|t_{1:m}) | \hat{O} | \phi (t|t_{1:m}) \rangle}{\langle \phi (t|t_{1:m}) | \phi (t|t_{1:m}) \rangle} $. The advantage of the quantum trajectories method is that it requires simulating only pure states rather than the full density matrix (i.e., $D$-dimensional vectors rather than a $D\times D$ matrix), which drastically reduces memory consumption and computation time.

In practice, the time-evolved states can be approximated by a state $ |\widetilde{\phi}(t|t_{1:m})\rangle \approx |\phi(t|t_{1:m})\rangle $ using classical techniques, such as tensor networks. In 1D systems, it is common to represent each approximate trajectory $ |\widetilde{\phi}(t|t_{1:m})\rangle $ as a matrix product state (MPS)~\cite{Klumper_1993, Vidal_2003, PerezGarcia_2007} and evolve it using a time evolution algorithm, like time-evolving block decimation (TEBD)~\cite{Vidal_2003, Vidal_2004, White_2004, Daley_2004}. As a result, the approximate QMC distribution is $\widetilde{\mathcal{P}} = \big\| | \widetilde{\phi} (t|t_{1:m})\rangle \big\|^2 $, and its associated error is the error incurred by the MPS representation and time-evolution algorithm. For example, TEBD works by applying Trotterized evolution and truncating the MPS bond dimension. Using $r$ steps of order-$n$ Trotterization results in a systematic error $ \mathcal{O}(r (Lt/r)^{n+1}) + \epsilon_{\text{trunc}}$, where $\epsilon_{\text{trunc}}$ is the error associated with truncation.

Moreover, to sample time-evolved states according to Eq.~\eqref{eq:master_eq_obs}, one creates an ensemble of $N_{\text{traj}}$ trajectories, each initialized as $|\psi\rangle$, and evolves them under $\hat{H}_{\text{eff}}$ with jump operators applied stochastically. Specifically, a jump is applied to a state when its squared norm decreases below a value randomly sampled from the uniform distribution on $[0,1]$. Early works showed that this process is equivalent to sampling time-evolved states $|\phi(t|t_{1:m}) \rangle \sim \mathcal{P}$~\cite{Dum_1992, Dalibard_1992, Carmichael1993, Molmer_1993}, with a statistical error $\mathcal{O}(1/\sqrt{N_{\text{traj}}})$; see Appendix~\ref{app:QTM} for details. This approach of realizing quantum trajectories has become a powerful tool in studying open system dynamics in quantum optics and condensed matter~\cite{Molmer_1993, Dum_1992_2, Ciccarello_2022, Daley_2014, sander2025large}, and has since been generalized to non-Markovian systems~\cite{Strunz_1999,Suess_2014}.

\section{Randomly Compiled Quantum Monte Carlo}\label{sec:RC_QMC}

This section develops \textit{randomly compiled quantum Monte Carlo} (RC-QMC) as a
general framework that uses randomized compiling~\cite{hastings2016turning,
Campbell_2017} to suppress systematic error in QMC algorithms.
The key insight is that averaging over a family of approximations can reduce the systematic error at little additional computational cost.
First, Sec.~\ref{sec:randomized_compiling} reviews randomized compiling and the error suppression result on which RC-QMC is founded. Then, Sec.~\ref{sec:RC_QMC_formulation} formulates RC-QMC and analyzes its advantage, while Sec.~\ref{sec:RC_QMC_remarks} comments on the statistical error exhibited by RC-QMC. 

\subsection{Randomized Compiling}\label{sec:randomized_compiling}

To understand randomized compiling, let us consider the task of implementing a unitary transformation $\mathcal{U}(\rho) = U \rho U^\dag$, for some unitary $U$. Standard methods, which we refer to as \textit{fixed compilation}, directly approximate $U$ with a single unitary $W \approx U$. In contrast, randomized compiling achieves a better approximation by using a quantum channel (or quantum operation) that is a mixture of $M$ unitaries:
\begin{equation}
    \Lambda(\rho) = \sum_{j=1}^M p_j W_j \rho W_j^\dagger,
\end{equation} 
where each $W_j$ is a distinct approximation to $U$. If the average of the unitaries, $\sum_{j=1}^M p_j W_j$, more closely approximates $U$ than any individual $W_j$, then $\Lambda$ provides an even better approximation to $\mathcal{U}$.

This statement can be quantified by considering the spectral norm $\| \cdot \|$ of operators and the 1-norm $\| \cdot \|_1$ between quantum states. In particular, if the error suffered by each unitary is at most $\| W_j - U \| \leq \epsilon $ for all $j$, and the average error is $ \big\| \sum_{j=1}^M p_j W_j - U \big\| \leq \mathcal{O}(\epsilon^2) $, then the error suffered by $\Lambda$ is~\cite{hastings2016turning, Campbell_2017}
\begin{equation}\label{eq:RC_err_suppression}
    \| \Lambda(\rho) - \mathcal{U}(\rho) \|_1 \leq \mathcal{O}(\epsilon^2) ,
\end{equation}
for any input state $\rho$ (see App.~\ref{app:mixing_lemma} for a formal presentation of this statement). Because the 1-norm upper bounds the error in observables (see Appendix~\ref{app:concepts}), this result implies that any observable estimated using $\Lambda(\rho)$ deviates from its target value by $\mathcal{O}(\epsilon^2)$. This is in contrast to fixed compilation methods that use a single approximation $W$ and incur error $\mathcal{O}(\epsilon)$. Randomized compiling therefore achieves a quadratic suppression of error (i.e., $\epsilon \mapsto \mathcal{O}(\epsilon^2)$) relative to fixed compilation. Intuitively, this error suppression arises because averaging over the $W_j$'s cancels out their first-order errors, leaving only higher-order $\mathcal{O}(\epsilon^2)$ errors. 

Quantum mechanically, the mixed state $\Lambda(\rho) = \sum_{j} p_j W_j \rho W_j^\dag$ can be realized by simply sampling $j\sim p_j$ and executing the corresponding unitary $W_j$. This incurs a cost that is the average cost of the $W_j$'s, weighted by $p_j$. Consequently, realizing $\Lambda(\rho)$ is no more expensive than implementing any single $W_j$. Similarly, estimating an observable in the state $\Lambda(\rho)$ corresponds to averaging over observables in the states $W_j \rho W_j^\dag$. To wit, $\text{tr} (\hat{O}\Lambda(\rho)) = \sum_j p_j \text{tr} (\hat{O} W_j \rho W_j^\dag )$ is an average over observables $\text{tr} ( \hat{O} W_j \rho W_j^\dag )$.

In many cases of interest, an appropriate set of unitaries $\{ W_j \}_{j=1}^M$ and probability distribution $p_j$ can be found analytically. As a result, randomized compiling has been used to enhance a range of quantum algorithms, including time evolution~\cite{Campbell_2019, pocrnic2023composite}, phase estimation~\cite{Wan_2022, gunther2025phase}, and quantum signal processing~\cite{martyn2024halving}. Of particular relevance to the RC-QMC applications developed in this paper will be randomized Trotter formulas for time evolution~\cite{Campbell_2019, Cho_2024, pocrnic2023composite}. In this context, the quadratic error suppression provided by randomized compiling compounds over many Trotter steps, leading to a substantial cost reduction; see Appendix~\ref{app:RandomCompilingTrotter} for a review of these methods. Randomized compiling techniques have also been extended to non-unitary operations like imaginary time evolution~\cite{pocrnic2023composite, martyn2024halving}, which we will make use of in later sections.

\subsection{Formulation of Randomly Compiled QMC}\label{sec:RC_QMC_formulation}

While the aforementioned applications of randomized compiling largely focus on quantum algorithms, its potential extends beyond this scope. We now show how randomized compiling can be repurposed to improve QMC simulations. 

Our goal is to use randomized compiling to suppress the systematic error suffered by QMC, and thereby lower the cost to achieve a desired accuracy. As described in the general QMC framework of Eq.~\eqref{eq:QMC_sys_err}, the systematic error of a QMC algorithm is governed by the approximate probability distribution $\widetilde{\mathcal{P}}$ used to estimate the target distribution $\mathcal{P}$ in evaluating an observable $\langle \hat{O} \rangle  $. Because randomized compiling works by averaging over many approximations, it naturally suggests replacing $\widetilde{\mathcal{P}}$ with a collection of approximations. We propose to realize this by introducing \textit{randomly compiled QMC}, which averages over a collection of approximate distributions: 
\begin{equation}\label{eq:RC_QMC}
    \langle \hat{O} \rangle \approx \mathop{\mathbb{E}}_{j\sim p_j} \bigg[ \mathop{\mathbb{E}}_{X \sim \widetilde{\mathcal{P}}_j } \big[ F(\hat{O}, X) \big] \bigg] = \mathop{\mathbb{E}}_{j\sim p_j} \Big[ \langle \hat{O} \rangle_{ \widetilde{\mathcal{P}}_j } \Big] . 
\end{equation}
Here $\{ \widetilde{\mathcal{P}}_j\}$ is a collection of approximate distributions, $p_j$ is a probability distribution over the approximations, and $\langle \hat{O} \rangle_{ \widetilde{\mathcal{P}}_j } := \mathop{\mathbb{E}}_{X \sim \widetilde{\mathcal{P}}_j } \big[ F(\hat{O}, X) \big]$ denotes the estimate of the observable using $\widetilde{\mathcal{P}}_j$. 
Eq.~\eqref{eq:RC_QMC} provides the general framework of RC-QMC that is compatible with many QMC algorithms. We will develop specific applications in Secs.~\ref{sec:Application_PIMC} and~\ref{sec:Application_QTM}, which will illustrate how this form emerges naturally in using randomized compiling to approximate a target state in QMC.

The essential improvement of RC-QMC is that by averaging over many approximations, it reduces the systematic error relative to standard QMC:
\begin{equation}
\begin{aligned}
    &\bigg|\langle \hat{O} \rangle - \mathop{\mathbb{E}}_{j\sim p_j} \Big[ \langle \hat{O} \rangle_{ \widetilde{\mathcal{P}}_j } \Big] \bigg| \ll 
    \big|\langle \hat{O} \rangle - \langle \hat{O} \rangle_{ \widetilde{\mathcal{P}} } \big|. 
\end{aligned}
\end{equation}
In addition, implementing RC-QMC is no more expensive than standard QMC. The double expectation value of Eq.~\eqref{eq:RC_QMC} can be realized by repeatedly sampling an approximation $j \sim p_j$ and evaluating an estimator of $\langle \hat{O} \rangle_{\widetilde{\mathcal{P}}_j} = \mathop{\mathbb{E}}_{X \sim \widetilde{\mathcal{P}}_j } \big[ F(\hat{O}, X) \big] $, as in standard QMC. The overall cost is therefore the average of the costs of sampling $\widetilde{\mathcal{P}}_j$, weighted by $p_j$. Thus, RC-QMC achieves a smaller systematic error than any single approximation $\widetilde{\mathcal{P}}_j$, while requiring a cost no greater than any individual approximation. This reduces the computational cost required to achieve a desired level of error relative to standard QMC. 

Applying RC-QMC to a specific problem requires choosing appropriate $\{ \widetilde{\mathcal{P}}_j \}$ and $p_j$ that provide a substantial error reduction without increasing the computational cost. The choice of these depends on the underlying QMC algorithm and approximation scheme, but can be guided by results from randomized compiling. In general, the approximate distributions $\{ \widetilde{\mathcal{P}}_j \}$ should concentrate around the target distribution $\mathcal{P}$, and their average with respect to $p_j$ should provide an even better approximation thereof. As we show in the applications of Secs.~\ref{sec:Application_PIMC} and~\ref{sec:Application_QTM}, appropriate such choices can be derived from existing results in randomized compiling.

Lastly, in comparing RC-QMC (Eq.~\eqref{eq:RC_QMC}) to the error reduction provided by randomized compiling (Eq.~\eqref{eq:RC_err_suppression}), one might imagine that RC-QMC offers at most a quadratic suppression of systematic error. This, however, is too conservative: in many QMC algorithms, the target distribution is approximated as a product of many operations (e.g., Trotterization). In these cases, the quadratic suppression of error compounds across each term in the product, yielding an overall greater reduction in error.

\subsection{Remarks on Statistical Error}\label{sec:RC_QMC_remarks}

It is important to emphasize that RC-QMC suppresses only the systematic error suffered in estimating an observable $\langle \hat{O} \rangle$, not the statistical error, which is rather governed by the number of samples. Therefore, the advantage of RC-QMC is most pronounced when the statistical error is smaller than the systematic error, as otherwise the statistical error dominates. 

We also note that RC-QMC introduces an additional statistical error due to sampling over the collection of distributions $\{ \widetilde{\mathcal{P}}_j \}$. However, because each $\widetilde{\mathcal{P}}_j$ closely approximates the target distribution $\mathcal{P}$, the corresponding observables $\langle \hat{O} \rangle_{\widetilde{\mathcal{P}}_j}$ are tightly concentrated around the exact value~\cite{Chen_2021}. As we show in Appendix~\ref{app:Stat_Error}, this introduces only a negligible increase in statistical error, of the same order as the systematic error. Thus the additional cost needed to compensate for this is far outweighed by the overall cost reduction afforded by RC-QMC.

\section{Application to Path Integral Monte Carlo}\label{sec:Application_PIMC}

In this section we develop \textit{randomly compiled path integral Monte Carlo} to improve upon the standard algorithm. While standard path integral Monte Carlo incurs a systematic error $\mathcal{O}(L^3\beta^3/r^2)$ that grows with the number of Hamiltonian terms $L$, our randomly compiled version employs a random Trotter formula known as \textit{QDrift}~\cite{Campbell_2019} and reduces the error to $\mathcal{O}(\lambda^2\beta^2/r)$, where $\lambda = \sum_{j=1}^L h_j $ is the 1-norm of the Hamiltonian coefficients. This removes explicit $L$-dependence and yields a substantial advantage when $\lambda \ll L$.

In more detail, Sec.~\ref{sec:RC_PIMC_subsec} derives randomly compiled path integral Monte Carlo and bounds is systematic error. Then, Sec.~\ref{sec:RC_PIMC_Example} illustrates the algorithm on a simple Hamiltonian and shows how QDrift compresses the effective path lengths required for simulation. Sec.~\ref{sec:RC_PIMC_Experiments} demonstrates its advantage through numerical experiments on a spin chain. For a more thorough review of the QDrift algorithm underlying these results, see Appendix~\ref{app:RandomCompilingTrotter}. 

\subsection{Randomly Compiled Path Integral Monte Carlo}\label{sec:RC_PIMC_subsec}

Path integral Monte Carlo estimates the thermal state, $\rho = e^{-\beta \hat{H}}/Z$, as an average over classical evolution paths, achieved by approximating $e^{-\beta \hat{H}}$ with Trotterization. For a Hamiltonian composed of $L$ terms and Trotterized over $r$ time steps (equivalent to a path length $r$), standard path integral Monte Carlo suffers a systematic error $\mathcal{O}(L^3\beta^3/r^2)$ in estimating arbitrary observables. As a result, achieving error at most $\epsilon$ requires a path length that grows with $\beta$ and $L$ as $r=\mathcal{O} (L^{3/2}\beta^{3/2}/\epsilon^{1/2})  $.

We will now develop a randomly compiled path integral Monte Carlo algorithm whose path length does not explicitly depend on $L$, and can be considerably more efficient for long-range Hamiltonians. To achieve this, we will use a random Trotter formula known as \textit{QDrift}~\cite{Campbell_2019}. QDrift importance samples terms from the Hamiltonian to evolve under, which has been shown to remove explicit dependence on $L$. This enables faster simulation than standard Trotterization and is particularly advantageous for long-range Hamiltonians with many weak interactions. Although QDrift was originally introduced for real-time evolution, here we employ its extension to imaginary-time~\cite{pocrnic2023composite} to enable efficient thermal state estimation.

In more detail, let us again consider a Hamiltonian $\hat{H} = \sum_{j=1}^L \hat{H}_j$ composed of $L$ terms, which we now express as $\hat{H} = \sum_{j=1}^L h_j \hat{\mathcal{H}}_j$ for $h_j>0$ and $\| \hat{\mathcal{H}}_j \| = 1$ by rescaling.\footnote{These conditions can always be met by rescaling the coefficients $h_j$ and absorbing phases into $\hat{\mathcal{H}}_j$. Recall that $\| \cdot \|$ denotes the spectral norm.} Let us also define the associated probability distribution $p_j = h_j/\lambda$, where the index $j$ labels each Hamiltonian term, and $\lambda = \sum_{j=1}^L h_{j}$ is the 1-norm of the coefficients. 
The QDrift channel for an imaginary time step $\Delta \beta = \beta/r$ is then defined as~\cite{Campbell_2019, pocrnic2023composite}:
\begin{equation}\label{eq:QDrift_Channel}
    \Lambda_{\text{QDrift}}(\rho) = \sum_{j=1}^L p_j e^{-\lambda \Delta \beta \hat{\mathcal{H}}_j } \rho e^{- \lambda \Delta \beta \hat{\mathcal{H}}_j } . 
\end{equation}
This channel corresponds to importance sampling a term $j \sim p_j$ and evolving under $e^{-\hat{\mathcal{H}}_j \lambda \Delta \beta}$. Repeating this channel over $r$ steps, the final state $(\Lambda_{\text{QDrift}})^{\circ r}(\rho)$ deviates from the target state $e^{-\beta \hat{H}} \rho e^{-\beta \hat{H}}$ by an error $\mathcal{O}(\lambda^2 \beta^2 /r)$ in trace distance.

Thus, achieving an error $\epsilon$ requires $r=\mathcal{O}(\lambda^2 \beta^2/\epsilon)$ steps, which scales with the 1-norm $\lambda$, rather than explicitly with the number of terms $L$ as in standard Trotterization. This can provide a substantial advantage when $\lambda \ll L$. The advantage of QDrift becomes even more pronounced when considering the total number of operations needed to perform evolution. Because each QDrift step evolves under a single term, the total number of operations is simply $r$. In contrast, each step of 2nd-order Trotterization evolves under $\mathcal{O}(L)$ terms, resulting in a total number of operations $\mathcal{O}(rL)$. Thus, QDrift can provide an even greater savings in the total number of operations. 

To integrate QDrift into path integral Monte Carlo, we approximate the thermal state as an average of $r/2$ symmetrized QDrift steps, each individually normalized: 
\begin{equation}\label{eq:RC_thermalstate}
    \frac{e^{-\beta \hat{H}}}{Z} \approx \sum_{\mathbf{j}} p_{\mathbf{j}} \frac{ \prod_{k=1}^{r/2} e^{-\lambda \frac{\beta}{r} \hat{\mathcal{H}}_{j_k} } \times \prod_{k'=r/2}^1 e^{-\lambda \frac{\beta}{r} \hat{\mathcal{H}}_{j_{k'}} } }{Z_{\mathbf{j}}}, 
\end{equation}
where
\begin{equation}\label{eq:Z_j_equation}
    Z_{\mathbf{j}} = \tr( \prod_{k=1}^{r/2} e^{-\lambda \frac{\beta}{r} \hat{\mathcal{H}}_{j_k} } \times \prod_{k'=r/2}^1 e^{-\lambda \frac{\beta}{r} \hat{\mathcal{H}}_{j_{k'}} }) . 
\end{equation}
In this expression $\mathbf{j} = (j_1, j_2, ..., j_{r/2})$ is a multi-index that we refer to as the \textit{QDrift sequence}, where $j_k \in \{ 1, 2, ..., L\}$ denotes the term sampled at time step $k$. The corresponding probability is
\begin{equation}
    p_{\mathbf{j}} = \prod_{k=1}^{r/2} p_{j_k} = \prod_{k=1}^{r/2} \frac{h_{j_k}}{\lambda} . 
\end{equation}
Intuitively, Eq.~\eqref{eq:RC_thermalstate} is an average of symmetric products of $r$ imaginary time steps $e^{-\lambda \frac{\beta}{r} \hat{\mathcal{H}}_{j_k} }$, thus resembling repeated application of the QDrift channel. We adopt this particular construction because it is explicitly an average over normalized density matrices and will be directly compatible with the RC-QMC framework.
Although this state is not exactly equal to $r/2$ applications of the QDrift channel, we show in Appendix~\ref{app:Error_Bound_RC_PIMC} that it approximates the thermal state to the same level of error, namely a trace distance $\mathcal{O}(\lambda^2 \beta^2 / r)$.

Using this state to estimate a thermal expectation value yields the expression
\begin{equation}\label{eq:RC_thermalstate_obs}
\begin{aligned}
    &\tr( \hat{O} \frac{e^{-\beta \hat{H}}}{Z} ) \approx \sum_{\mathbf{j}} p_{\mathbf{j}} \langle \hat{O} \rangle_{\mathbf{j}} ,
\end{aligned}
\end{equation}
where
\begin{equation}
\begin{aligned}
    &  \langle \hat{O} \rangle_{\mathbf{j}} = \frac{ \tr( \hat{O} \prod_{k=1}^{r/2} e^{-\lambda \frac{\beta}{r} \hat{\mathcal{H}}_{j_k} } \times \prod_{k'=r/2}^1 e^{-\lambda \frac{\beta}{r} \hat{\mathcal{H}}_{j_{k'}} } ) }{ Z_{\mathbf{j}} }  . 
\end{aligned}
\end{equation}
This equates a thermal observable to an average over observables $\langle \hat{O} \rangle_{\mathbf{j}}$, each of which is indexed by $\mathbf{j}$ and can be computed via path integral Monte Carlo methods. With this formulation in place, we now parallel the development of path integral Monte Carlo to work this expression into an RC-QMC algorithm.

To do so, we first re-express $\langle \hat{O} \rangle_{\mathbf{j}}$ as a sum over classical paths by inserting resolutions of the identity between each term in the product in the numerator. Schematically, this would yield an expression $\langle \hat{O} \rangle_{\mathbf{j}} = \mathop{\mathbb{E}}_{ \mathbf{X} \sim e^{-S(\mathbf{X}, \mathbf{j})}}  \big[ \langle \mathbf{X}_{0,:} | \hat{O} | \mathbf{X}_{0,:} \rangle \big] $, where the action $S(\mathbf{X}, \mathbf{j})$ now depends on the QDrift sequence $\mathbf{j}$ through the sampled terms $\hat{\mathcal{H}}_{j_k}$.

However, this procedure can be drastically simplified to reduce computational cost. Because each step $\hat{\mathcal{H}}_{j_k}$ corresponds to evolution under a single term, many of which mutually commute, it is unnecessary to insert resolutions of the identity between each term in the product. Rather, we can partition the product into mutually commuting sets, and insert resolutions of the identity only between these sets. In addition, if each Hamiltonian term $\hat{\mathcal{H}}_{j}$ acts locally, these resolutions of the identity are required only on the local support of $\hat{\mathcal{H}}_j$. Together, these two considerations allow us to substantially compress the path lengths required to simulate the thermal state, in contrast to standard path integral Monte Carlo where paths are not compressible in general.

\begin{figure*}[htbp]
    \centering
    \includegraphics[width=1.0\linewidth]{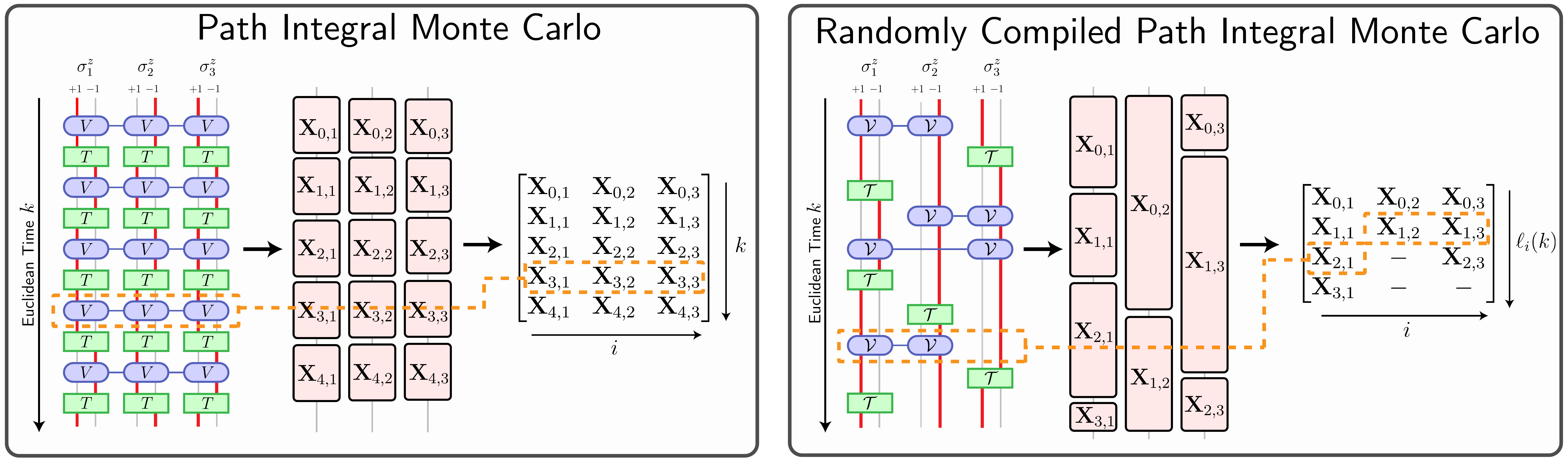}
    \caption{(\textbf{Left}) Standard path integral Monte Carlo on $3$ sites: Trotterization is used to decompose $e^{-\beta H}$ into a product of evolution under the potential ($V$) and kinetic ($T$) energies, and particle positions take values $\pm 1$ in the $\sigma^z$ basis at each time step $k$. The particle paths are then aggregated into an array $\mathbf{X}$, indexed by site $i$ and time step $k$. The dashed orange line depicts the mapping of a single time step to a particle configuration in $\mathbf{X}$. 
    (\textbf{Right}) Randomly compiled path integral Monte Carlo on $3$ sites: individual potential ($\mathcal{V}$) and kinetic ($\mathcal{T}$) energy terms are sampled according to QDrift, and a particle's position can change only when its corresponding kinetic term is sampled. The paths are then aggregated into a \textit{jagged} array $\mathbf{X}$, indexed by site $i$ and its corresponding index at time step $k$, given by $\ell_i(k)$. The dashed orange line again depicts the mapping of a single time step to a configuration in the array, showing that this mapping generally corresponds to different indices for different particles. Note that in our formulation described in Eq.~\eqref{eq:RC_thermalstate}, the pattern of sampled terms is symmetric, which is not depicted here for illustrative simplicity. 
    }
    \label{fig:PIMC}
\end{figure*}

As we will show below, this compression implies that a sequence $\mathbf{j}$ constrains the set of paths that make nonzero contributions to $\langle \hat{O}\rangle_{\mathbf{j}}$. In addition, for a fixed sequence $\mathbf{j}$, the paths associated with different particles may differ in length. To represent all such paths efficiently, we use an array $\mathbf{X}$, whose $i$th column $\mathbf{X}_{:,i}$ corresponds to the path of particle $i$, as introduced in Eq.~\eqref{eq:paths_array}. Because these lengths may vary, $\mathbf{X}$ now takes the form of a \textit{jagged array}, where each column may have different size. In effect, this representation reduces the number of variables that need to be kept track of, thus compressing the average path length to well below the number of time steps $r$. We will make this explicitly clear in Sec.~\ref{sec:RC_PIMC_Example}, where we exemplify this procedure for a simple Hamiltonian.

This analysis leads to the defining expression of randomly compiled path integral Monte Carlo:
\begin{equation}\label{eq:RC_PIMC}
\begin{aligned}
    \langle \hat{O} \rangle &= \text{tr} \bigg( \hat{O} \frac{e^{-\beta \hat{H}}}{Z} \bigg) \\
    &= \mathop{\mathbb{E}}_{\mathbf{j} \sim p_{\mathbf{j}}} \ \bigg[ \mathop{\mathbb{E}}_{\mathbf{X} \sim e^{-S(\mathbf{X}, \mathbf{j})} } \Big[ \langle \mathbf{X}_{0,:} | \hat{O} | \mathbf{X}_{0,:} \rangle \Big]\bigg]  + \mathcal{O}\Big(\frac{\lambda^2 \beta^2 }{r}\Big) ,  
\end{aligned}
\end{equation} 
where we have included the systematic error $\mathcal{O}(\lambda^2 \beta^2 /r)$. Here we again take $\hat{O}$ to be diagonal in position space, such that each particle's path is periodic. Algorithmically, the double expectation value in Eq.~\eqref{eq:RC_PIMC} can be evaluated by sampling configurations $\mathbf{j} \sim p_{\mathbf{j}}$, and then sampling paths $\mathbf{X} \sim e^{-S(\mathbf{X}, \mathbf{j})}$ conditioned on $\mathbf{j}$. Because $p_{\mathbf{j}}$ is known analytically, it can be sampled directly, whereas sampling paths from $e^{-S(\mathbf{X}, \mathbf{j})}$ can be performed using standard Markov chain Monte Carlo. 

Lastly, we note that alternative variants of QDrift can also be incorporated into path integral Monte Carlo. For example, one could use composite QDrift~\cite{Hagan_2023}, which randomizes over only a subset of terms in the Hamiltonian. Ref.~\cite{pocrnic2023composite} explored this for path integral Monte Carlo, but sampled only a single QDrift sequence per QMC simulation, whereas ours averages over all such sequences. This distinction leads to both enhanced error suppression and reduced computational cost in our formulation.

\subsection{Example for a Simple Hamiltonian}\label{sec:RC_PIMC_Example}

For concreteness, let us illustrate randomly compiled path integral Monte Carlo for an $N$-particle (or $N$-site) Hamiltonian $\hat{H} = \hat{T} + \hat{V}$, where $\hat{T}$ and $\hat{V}$ serve as analogs of kinetic and potential energy. Let us take $\hat{T}$ to decompose into a sum of single-particle operations, and $\hat{V}$ into a sum of two-body interactions between particles $i,i'$: 
\begin{equation}
    \hat{T} = \sum_{i=1}^N t_i \hat{\mathcal{T}}_i , \quad \hat{V} = \sum_{i \neq i'}^N v_{i,i'} \hat{\mathcal{V}}_{i,i'} , 
\end{equation}
where $t_i, v_{i,i'} > 0$, and each operator is individually normalized as $\|\hat{\mathcal{T}}_i\|=1$ and $\|\hat{\mathcal{V}}_{i,i'} \|=1$.\footnote{Strictly speaking, this normalization is not possible in continuous space, where the kinetic energy is unbounded. Here, however, we will consider models with a nonzero lattice spacing, such as spin chains, where this condition can be imposed.} 

To develop this into randomly compiled path integral Monte Carlo, we will use a basis $\{ |\mathbf{x}\rangle \}$ in which the potential energy is diagonal, satisfying 
\begin{equation}
    \hat{\mathcal{V}}_{i,i'} |\mathbf{x}\rangle = \mathcal{V}(\mathbf{x}(i),\mathbf{x}(i')) |\mathbf{x}\rangle , 
\end{equation}
for a local potential energy function $\mathcal{V}(x,x')$ (e.g., $\mathcal{V}(x,x') \propto 1/|x-x'|$ for a Coulomb interaction). On the other hand, the kinetic terms are not diagonal in this basis. Consequently, resolutions of the identity need only be inserted after each sampled kinetic term $e^{- \lambda \Delta \beta \hat{\mathcal{T}}_i }$ in a given QDrift sequence, as only these can couple non-equal basis states and modify the path configuration. In addition, these resolutions need only act nontrivially on the support of $\hat{\mathcal{T}}_i$, allowing us to locally resolve the identity on particle $i$ as 
\begin{equation}
    I = I_{1:i-1} \otimes \sum_{x} |x \rangle \langle x | \otimes I_{i+1:N} , 
\end{equation}
where $I_{i:i'}$ denotes the identity on site $i$ through $i'$.

This contrasts with standard path integral Monte Carlo, where each particle follows a classical path and can change its position at every time step. In the randomly compiled formulation, each particle similarly follows a classical path, but its position can only change at time steps where its kinetic term is sampled. Consequently, the QDrift sequence $\mathbf{j}$ constrains the allowed paths, both in their length and configuration. 


As an example, Fig.~\ref{fig:PIMC} depicts a QDrift sequence being mapped to a sum over paths that are represented as a jagged array $\mathbf{X}$, in comparison with standard path integral Monte Carlo. As we see, the path of particle $i=1$ is described by a vector of length $\text{len}(\mathbf{X}_{:,1}) = 4$,
while the paths of particles $i=2$ and $i=3$ are of length $2$ and $3$, respectively.

Because each particle can change position only when its kinetic term is sampled, the time steps at which these updates occur differ for each particle. To keep track of this, we define $\ell_i(k) \in \{ 0, 1,2,..., \text{len}(\mathbf{X}_{:,i}) -1 \}$ as the index in the path of particle $i$ at time step $k$. 
For example, as illustrated in right image of Fig.~\ref{fig:PIMC}, time step $k=0$ corresponds to index $\ell_i(0) = 0$ for each particle, whereas time step $k = 7$ (denoted by the dashed orange line) corresponds to indices $\ell_1(7) = 2,\ \ell_2(7) = 1,\ \ell_3(7) = 1$.

In aggregate, this analysis shows that the action decomposes into a sum over time steps as 
\begin{equation}\label{eq:RC_PIMC_Action_1}
    S(\mathbf{X}, \mathbf{j}) = \lambda \Delta \beta \sum_{k=0}^{r-1} S(\mathbf{X}, \mathbf{j}; k) ,
\end{equation}
where $S(\mathbf{X}, \mathbf{j}; k)$ is the contribution from time step $k$: 
\begin{equation}\label{eq:RC_PIMC_Action_2}
\begin{aligned}
    &S(\mathbf{X}, \mathbf{j}; k) = \\
    &\qquad 
    \begin{cases}
        \mathcal{T}\big( \mathbf{X}_{\ell_i(k),i} , \ \mathbf{X}_{\ell_i(k) + 1,i} \big) & \text{if } \hat{\mathcal{T}}_{i} \text{ is sampled} \\
        \mathcal{V}\big(\mathbf{X}_{\ell_i(k),i}, \ \mathbf{X}_{\ell_{i'}(k),i'} \big) & \text{if } \hat{\mathcal{V}}_{i,i'} \text{ is sampled.}
    \end{cases}
\end{aligned}
\end{equation}
Here, $\mathcal{T}(x_\ell, x_{\ell+1})$ is a local kinetic energy function defined by $\langle x_\ell | e^{-\lambda \Delta \beta \hat{\mathcal{T}}_i} | x_{\ell+1} \rangle = e^{-\lambda\Delta\beta \mathcal{T}(x_\ell , x_{\ell+1})}$ for positions $x_\ell, x_{\ell+1}$ at site $i$. 
This action is analogous to the action in standard path integral Monte Carlo (Eq.~\eqref{eq:PIMC_Action}), but with the index $\ell_i(k)$ replacing the time step $k$, and individually sampled kinetic/potential energies replacing the sum of all Hamiltonian terms. A key distinction, however, is that the action is now a random variable, as it depends on the QDrift sequence $\mathbf{j}$. The benefit of this randomization is that the action can now be composed of far fewer operations than in standard path integral Monte Carlo, leading to computational savings that we demonstrate in the experiments below.

In evaluating $S(\mathbf{X}, \mathbf{j})$ in practice, it is often simplest to compute the entire kinetic contribution as a sum over all indices $\ell_i$, namely $\sum_{i=1}^N \sum_{\ell_i=0}^{\text{len}(\mathbf{X}_{:,i})-1} \mathcal{T} \big( \mathbf{X}_{\ell_i, i}, \ \mathbf{X}_{\ell_i + 1, i} \big) $, 
and to record each potential term $\mathcal{V} \big(\mathbf{X}_{\ell_i,i}, \ \mathbf{X}_{\ell_{i'},i'} \big)$ as it is sampled. This avoids needing to track the indices $\ell_i(k)$ across every time step, resulting in a more memory-efficient implementation.

Finally, while this analysis specializes to a simple Hamiltonian $\hat{H}=\hat{T}+\hat{V}$, we emphasize that randomly compiled path integral Monte Carlo can be extended to more general Hamiltonians in the same manner. A key consideration in doing so is to minimize the computational cost by inserting local resolutions of the identity only between non-commuting sets of operators that appear in the QDrift sequence.

\subsection{Numerical Experiments}\label{sec:RC_PIMC_Experiments}
To demonstrate our algorithm, let us apply randomly compiled path integral Monte Carlo to the long-range Ising model~\cite{Scha_2013, Buyskikh_2016} on $N$ spins: 
\begin{equation}\label{eq:LR_Ising_Model}
    \hat{H} = -J \sum_{i,i'=1; \ i < i'}^N \frac{\sigma^z_i \sigma^z_{i'}}{|i-i'|^{2}} - h\sum_{i=1}^N \sigma^{x}_i  , 
\end{equation}
where $\sigma^z_i$ and $\sigma^x_i$ are the Pauli matrices on site $i$. 
Because its Hamiltonian features long-range interactions between all pairs of particles, the total number of terms is $L = \mathcal{O}(N^2)$. However, as these interactions decay with distance, the 1-norm of the coefficients is $\lambda = J \sum_{i < i'}^N \frac{1}{|i-i'|^2} + hN = \mathcal{O}(N)$, which is much less than $L$ for large systems. We therefore expect an advantage in deploying randomly compiled path integral Monte Carlo on this model.

\begin{figure*}[t]
    \centering
    \includegraphics[width=0.87\linewidth]{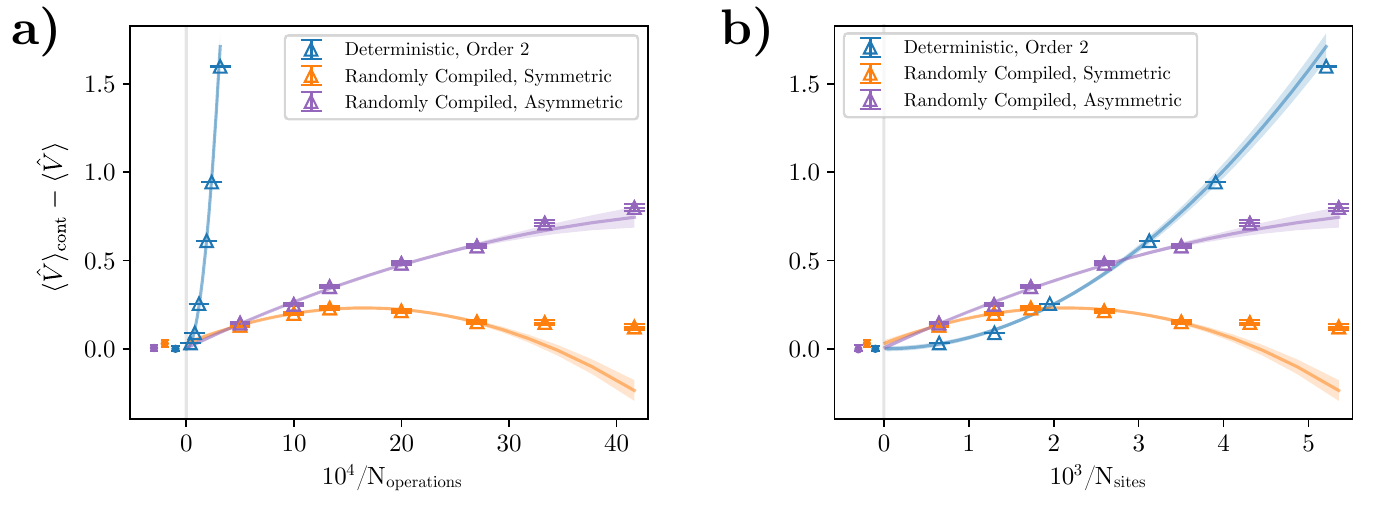}
    \caption{Comparison of error in the observable $\langle \hat{V} \rangle$ between standard path integral Monte Carlo (using 2nd-order Trotterization) and randomly compiled path integral Monte Carlo (using symmetric and asymmetric approaches), where the error has been defined relative to the continuum limit of the standard algorithm (denoted $\langle \hat{V} \rangle_\mathrm{cont}$). 
    Inner error bars for the randomly compiled methods denote the MCMC sampling errors within QDrift sequences, and outer error bars are the total error estimate including the variance due to averaging over different QDrift sequences (as described in Appendix~\ref{app:Stat_Error}). 
    Continuum limit extrapolated values are shown at slightly negative offset values for clarity. 
    \textbf{a)}: Here the comparison is shown in terms of the number of operations $N_\mathrm{operations}$ required to compute the action. The randomly compiled approaches reach smaller systematic errors at far fewer operations compared to the standard algorithm.  
    \textbf{b)}: The comparison is shown in terms of the second metric considered, $N_\mathrm{sites}$ counting the number of variables required to represent a path configuration in either scheme. 
    }
    \label{fig:QDrift-exp}
\end{figure*}

For this Hamiltonian, the kinetic terms are $\sigma_i^x $ and the potential terms are $\sigma_i^z \sigma_{i'}^z$. Following the prescription above, we will therefore work in the eigenbasis of $\sigma^z_i$ where $\mathbf{X}_{\ell_i, i} = \pm 1$, and insert resolutions of the identity after each $\sigma_i^x $. This yields an action of the form shown in Eqs.~\eqref{eq:RC_PIMC_Action_1} and~\eqref{eq:RC_PIMC_Action_2}, where:
\begin{align}
&\mathcal{T}\big( \mathbf{X}_{\ell_i(k),i} , \ \mathbf{X}_{\ell_i(k) + 1,i} \big) =\nonumber \\
&\qquad \  \frac{-1}{\lambda \Delta \beta}\ln\left( 1 + \mathbf{X}_{\ell_i(k),i} \mathbf{X}_{\ell_i(k) + 1,i} e^{-2\lambda \Delta \beta} \right) + A \\
&\mathcal{V}\big(\mathbf{X}_{\ell_i(k),i},
\ \mathbf{X}_{\ell_{i'}(k),i'} \big) = - \mathbf{X}_{\ell_{i}(k),i}\mathbf{X}_{\ell_{i'}(k),i'} , 
\end{align}
for a constant $A$ that is independent of $\mathbf{X}$ and thus dismissible. 
We sample this action using Metropolis-Hastings by proposing local flip updates to the individual particle configurations $\mathbf{X}_{\ell_i, i}$.

To compare with standard path integral Monte Carlo, we fix a set of parameters $N = 32, \  \beta = 8, \  J = 0.1, \ h = 1.0$, and choose our target observable to be the potential energy $\langle \hat{V} \rangle := \frac{1}{Z}\mathrm{tr}(e^{-\beta \hat{H}} \hat{V})$. 
We additionally consider a modified version of randomly compiled path integral Monte Carlo. In contrast to the construction presented in Eq.~\eqref{eq:RC_thermalstate}, where the sequences of imaginary time steps are symmetric, here we also study an alternative formulation where the Hamiltonian term at each step is drawn independently and the sequence is no longer symmetric. We refer to this modified version as the \textit{asymmetric} approach, and the original construction as the \textit{symmetric} approach. Note that the error bounds proven for the symmetric approach do not necessarily carry over to the asymmetric setting to this structural difference. However, as we discuss in Appendix~\ref{app:Error_Bound_RC_PIMC}, we expect that similar bounds still hold, and our experiments corroborate this. An advantage of the asymmetric approach is that it permits averaging an observable over all time steps, analogous to volume averaging in translationally invariant systems, thereby reducing statistical noise.

To compare the different path integral Monte Carlo approaches, Fig.~\ref{fig:QDrift-exp} plots continuum-limit ($r \to \infty$) extrapolations of $\langle \hat{V} \rangle$ against two metrics. 
The first metric we consider is the number of operations required to compute the action on the configuration, assuming each $\hat{\sigma}_z\hat{\sigma}_z$ and $\hat{\sigma}_x$ counts as a single operation. We denote this by $N_\mathrm{operations}$. This is a classical analog of the `gate count' cost of implementing algorithms on a quantum computer.
For standard path integral Monte Carlo, the number of operations is $r \cdot (N + \frac{N (N-1)}{2} )$ to account for every term in the Hamiltonian, whereas for its randomly compiled variant it is simply the number of sampled terms $r$.

The second metric is the number of degrees of freedom used in each simulation (known as `sites' in the path integral Monte Carlo literature), denoted by $N_\mathrm{sites}$. 
For standard path integral Monte Carlo, $N_\mathrm{sites} = N \cdot r$ because a complete basis is inserted at every time step. 
However, for the randomly compiled approaches, a basis is inserted only after a kinetic term is sampled, and thus the average number of sites is given by $\overline{ N_\mathrm{sites}} = N \cdot \mathbb{E}[\text{max}(1, \# \mathcal{T})]$ where $\# \hat{\mathcal{T}}$ is the number of times a kinetic term is sampled for a single spin. 
This expression evaluates to:
\begin{equation}
\begin{aligned}
    &\overline{ N_\mathrm{sites}} = \\
    & \qquad N \cdot \left( \frac{r h}{\lambda } + \left( 1 - \frac{h}{\lambda} \right)^{x}\right) , 
    \quad x = \begin{cases}
        \frac{r}{2} & \text{symmetric} \\
        r & \text{asymmetric,}  
    \end{cases} 
\end{aligned}
\end{equation}
where the first term in parentheses is the expected number of kinetic terms sampled for a single spin, and second term is the contribution when a kinetic term is never sampled.
A smaller $N_\mathrm{sites}$ for a given systematic error corresponds to requiring less storage, and possibly lower autocorrelation times due to Monte Carlo sampling over a lower-dimensional space.

Note that because randomly compiled path integral Monte Carlo is a first order Trotterization scheme, the continuum limit fit includes both linear and quadratic pieces, whereas that of standard path integral Monte Carlo (using 2nd-order Trotterization) starts at a quadratic order. 
As seen in Fig.~\ref{fig:QDrift-exp}, randomly compiled path integral Monte Carlo significantly outperforms the standard algorithm when convergence is measured in terms of the number of operations. Both the symmetric and asymmetric variants require substantially fewer operations to achieve the same level of systematic error. In practice, this translates to a reduced computational cost per Monte Carlo update, such as a fewer floating-point operations needed to compute the action in a step of Markov chain Monte Carlo. On the other hand, randomly compiled path integral Monte Carlo performs comparably to the standard algorithm when measured in terms of $N_{\text{sites}}$, implying that these algorithms have similar memory requirements.

Lastly, for completeness we note that because this model's local Hilbert space is finite-dimensional (each spin is $\pm 1$), it is also possible to formulate an alternate continuous-time Monte Carlo approach that obviates the effects of Trotter error~\cite{Rieger1999-fo}, albeit at the expense of a more elaborate representation of paths and their associated update procedure. Here, we use this Ising-like setting primarily as a convenient benchmark for illustrating randomly compiled path integral Monte Carlo. More broadly however, randomly compiled path integral Monte Carlo is a general algorithm applicable to a much wider class of systems, including those with infinite dimensional local Hilbert spaces (such as lattice gauge theories), for which no continuous-time Monte Carlo algorithm is known.

\section{Application to the Quantum Trajectories Method}\label{sec:Application_QTM}

In this section we develop the \textit{randomly compiled quantum trajectories method} as an improvement over the standard algorithm. While standard quantum trajectories evolves states using deterministic Trotterization, our randomly compiled version uses \textit{randomly corrected Trotterization}~\cite{Cho_2024}, which doubles the effective Trotter order and reduces the number of time steps required to achieve a given accuracy. This is an improvement in Trotter order rather than parameter scaling, and is therefore complementary to the randomly compiled path integral Monte Carlo advantage developed in the previous section.

Sec.~\ref{sec:RC_QTM_subsec} begins by deriving the randomly compiled quantum trajectories method and establishing a bound on its systematic error. Sec.~\ref{sec:RC_QTM_Experiments} demonstrates the advantage of randomly compiled quantum trajectories on a dissipative spin chain model, benchmarked against exact results for small system sizes and extended to larger systems using matrix product states. For a more thorough review of randomly corrected Trotterization algorithm that these results rest on, see Appendix~\ref{app:RandomCompilingTrotter}.

\subsection{Randomly Compiled Quantum Trajectories Method}\label{sec:RC_QTM_subsec}

As we presented in Sec.~\ref{sec:QMC}, the quantum trajectories method simulates open system dynamics by averaging over many pure state trajectories, each of which evolves under an effective Hamiltonian interspersed with stochastic jumps. Explicitly, the pure states simulated take the form (see Eq.~\eqref{eq:QTM_state}) 
\begin{equation}
\begin{aligned}
    & | \phi (t| t_{1:m}) \rangle \\
    & \quad = e^{-i\hat{H}_{\text{eff}}(t-t_{m})} c e^{-i\hat{H}_{\text{eff}}(t_m-t_{m-1})} c ... c e^{-i\hat{H}_{\text{eff}}t_1} |\psi \rangle ,
\end{aligned}
\end{equation}
where $|\psi\rangle$ is the initial state, $\hat{H}_{\text{eff}}$ is the effective Hamiltonian, and $m$ is the number of jumps experienced in this trajectory. Here we consider realizing this using tensor networks (i.e., MPSs), where evolution under $\hat{H}_{\text{eff}}$ may be approximated using a Trotter-based algorithm like TEBD~\cite{Vidal_2003, Vidal_2004, White_2004, Daley_2004}. Employing a standard order-$n$ Trotter formula introduces an error $\mathcal{O}((\Delta t)^{n+1})$ at each time step of size $\Delta t$.

\begin{figure*}[htbp]
    \centering
    \includegraphics[width=0.98\linewidth]{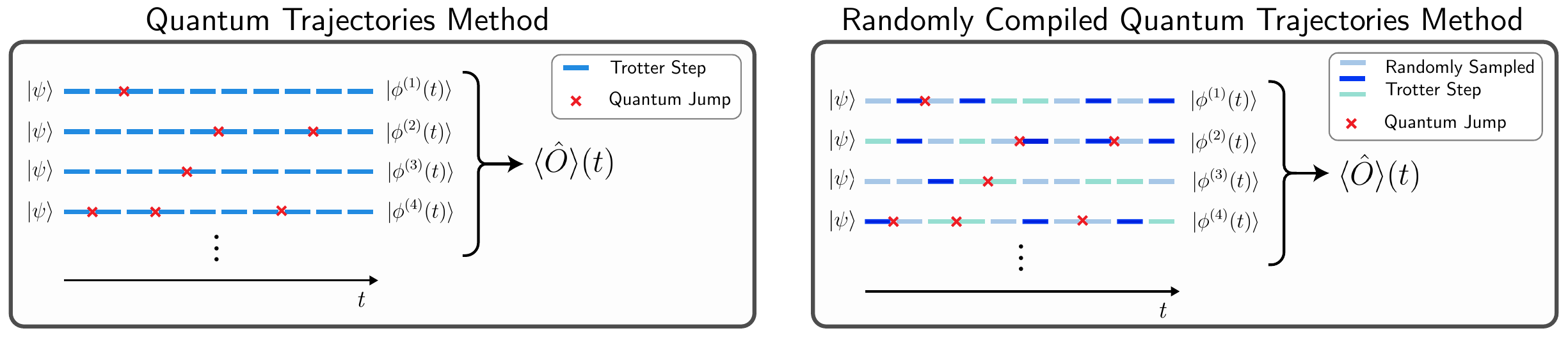}
    \caption{An illustration of the quantum trajectories method (\textbf{left}), and the randomly compiled quantum trajectories method (\textbf{right}). Here, pure state trajectories are simulated by time-evolving an initial state $|\psi\rangle$ using Trotterization (denoted by blue rectangles), with intermediate quantum jumps applied randomly (denoted by red crosses). In the standard algorithm, each Trotter step is taken to be the same approximation. In the randomly compiled algorithm, each Trotter step is randomly sampled from a set of approximations. 
    At the end of evolution, the trajectories are aggregated together to estimate an observable $\langle \hat{O} \rangle (t)$ according to Eqs.~\eqref{eq:master_eq_obs} and~\eqref{eq:RC_QTM}.
    }
    \label{fig:RC_QTM}
\end{figure*}

To develop a randomly compiled version of the quantum trajectories method, we replace this standard Trotter decomposition with a random channel that approximates evolution under $e^{-i \hat{H}_{\text{eff}} \Delta t}$:
\begin{equation}\label{eq:RC_QTM_Trotter}
    \Lambda_{\text{RandCorrTrotter}}(\rho) = \sum_j p_j W_j \rho W_j^\dagger , 
\end{equation}
for a choice of operators $W_j$ and probabilities $p_j$. In particular, here we use the method introduced in Ref.~\cite{Cho_2024}, which we refer to as \textit{randomly corrected Trotterization}, and have denoted the corresponding channel by $\Lambda_{\text{RandCorrTrotter}}$. This approach improves upon standard Trotterization by quadratically suppressing the time step error, effectively doubling the order of Trotterization.

In more detail, Ref.~\cite{Cho_2024} suppresses the error of an order-$n$ Trotter formula $S_n(\Delta t)$ (defined in Appendix~\ref{app:RandomCompilingTrotter}) by taking
\begin{equation}\label{eq:randomly_corrected_Trotter_ops}
    W_j = S_{n}(\tfrac{\Delta t}{2}) C_j(\Delta t ) S_{n} (\tfrac{\Delta t}{2}) ,
\end{equation}
where $C_j(\Delta t)$ is an order-$n$ \textit{correction operation} indexed by $j$. By carefully selecting these correction operations and their associated probabilities $p_j$ to cancel out the higher-order errors suffered by $S_n(\Delta t)$, the channel $\Lambda_{\text{RandCorrTrotter}}$ reproduces time evolution up to an error $\mathcal{O}(\Delta t^{2n+2})$, which is equivalent to an order-$(2n+1)$ Trotter formula. In practice, $C_j(\Delta t)$ and $p_j$ correspond to time evolution under nested commutators of the Hamiltonian terms; because their exact expressions are lengthy, we defer them to Appendix~\ref{app:RandomCompilingTrotter}. Owing to this error suppression, randomly corrected Trotterization requires $\mathcal{O}(tL (tL/\epsilon)^{1/(2n+1)})$ Trotter steps to approximate evolution over a duration $t$ to error $\epsilon$. 
This reduces computational cost relative to standard Trotterization by a factor of $\mathcal{O}\big((tL/\epsilon)^{\frac{n+1}{n(2n+1)}}\big)$, which approaches $\mathcal{O}((tL/\epsilon)^{1/2n})$ for large $n$. 

In applying this to the quantum trajectories method, we note that although evolution under $\hat{H}_{\text{eff}}$ is non-unitary, randomly corrected Trotterization extends to this setting in essentially the same way as it does to unitary evolution. The reason is that the mathematical statement that describes randomized compiling (i.e., Eq.~\eqref{eq:RC_err_suppression}) applies equally well to non-unitary maps; see Appendix~\ref{app:mixing_lemma} for details.

We now integrate randomly corrected Trotterization into the quantum trajectories method by replacing each time evolution step in the quantum trajectories method with the channel $\Lambda_{ \text{RandCorrTrotter}} $. This yields the following expression for an observable:
\begin{equation}\label{eq:RC_QTM}
\begin{aligned}
    \langle \hat{O} \rangle (t) &= \tr(\hat{O} \rho(t))\\
    & \approx\mathop{\mathbb{E}}_{\mathbf{j} \sim p_{\mathbf{j}}} \Bigg[ \mathop{\mathbb{E}}_{|\widetilde{\phi}_{\mathbf{j}}\rangle \sim \langle \widetilde{\phi}_{\mathbf{j}} | \widetilde{\phi}_{\mathbf{j}} \rangle} \Bigg[ \frac{\langle \widetilde{\phi}_{\mathbf{j}} (t|t_{1:m}) | \hat{O} | \widetilde{\phi}_{\mathbf{j}} (t|t_{1:m}) \rangle}{\langle \widetilde{\phi}_{\mathbf{j}} (t|t_{1:m}) | \widetilde{\phi}_{\mathbf{j}} (t|t_{1:m}) \rangle}  \Bigg] \Bigg] . 
\end{aligned}
\end{equation}
Here $\mathbf{j} = (j_1, j_2, ..., j_r)$ is a multi-index, and $p_{\mathbf{j}} = \prod_{k=1}^r p_{j_k}$ is the probability distribution of sampling $\textbf{j}$. The states $| \widetilde{\phi}_{\mathbf{j}} (t| t_{1:m}) \rangle$ approximate evolution under $\hat{H}_{\text{eff}}$ with a Trotter sequence indexed by $\mathbf{j}$:
\begin{widetext}
\begin{equation}\label{eq:RC_QTM_state}
\begin{aligned}
    | \widetilde{\phi}_{\mathbf{j}} (t| t_{1:m}) \rangle & =  \Big( \prod_{k_{m+1}=r_m+1}^{r} W_{j_{k_{m+1}}} \Big) \times \Big(c \prod_{k_m=r_{m-1}+1}^{r_m} W_{j_{k_m}} \Big) \times ... \times \Big(c \prod_{k_1=1}^{r_1} W_{j_{k_1}}\Big) |\psi \rangle \\
    & \approx \big( e^{-i\hat{H}_{\text{eff}}(t-t_{m})} \big) \times \big( c e^{-i\hat{H}_{\text{eff}}(t_m-t_{m-1})} \big) \times ... \times \big( c e^{-i\hat{H}_{\text{eff}}t_1} \big) |\psi \rangle . 
\end{aligned}
\end{equation}
\end{widetext}
In this expression, each term in parentheses denotes time evolution followed by a jump, and the $i$th jump occurs at Trotter step $r_i =  t_i/\Delta t = t_i r/t$. 
Eq.~\eqref{eq:RC_QTM} is thus the central expression of the randomly compiled quantum trajectories method, in which one averages over many trajectories, each generated with an independently sampled Trotter sequence $\mathbf{j}$. This construction is illustrated in Fig.~\ref{fig:RC_QTM} alongside standard quantum trajectories.

\subsection{Numerical Experiments}\label{sec:RC_QTM_Experiments}

\begin{figure*}
    \includegraphics[width=.99\linewidth]{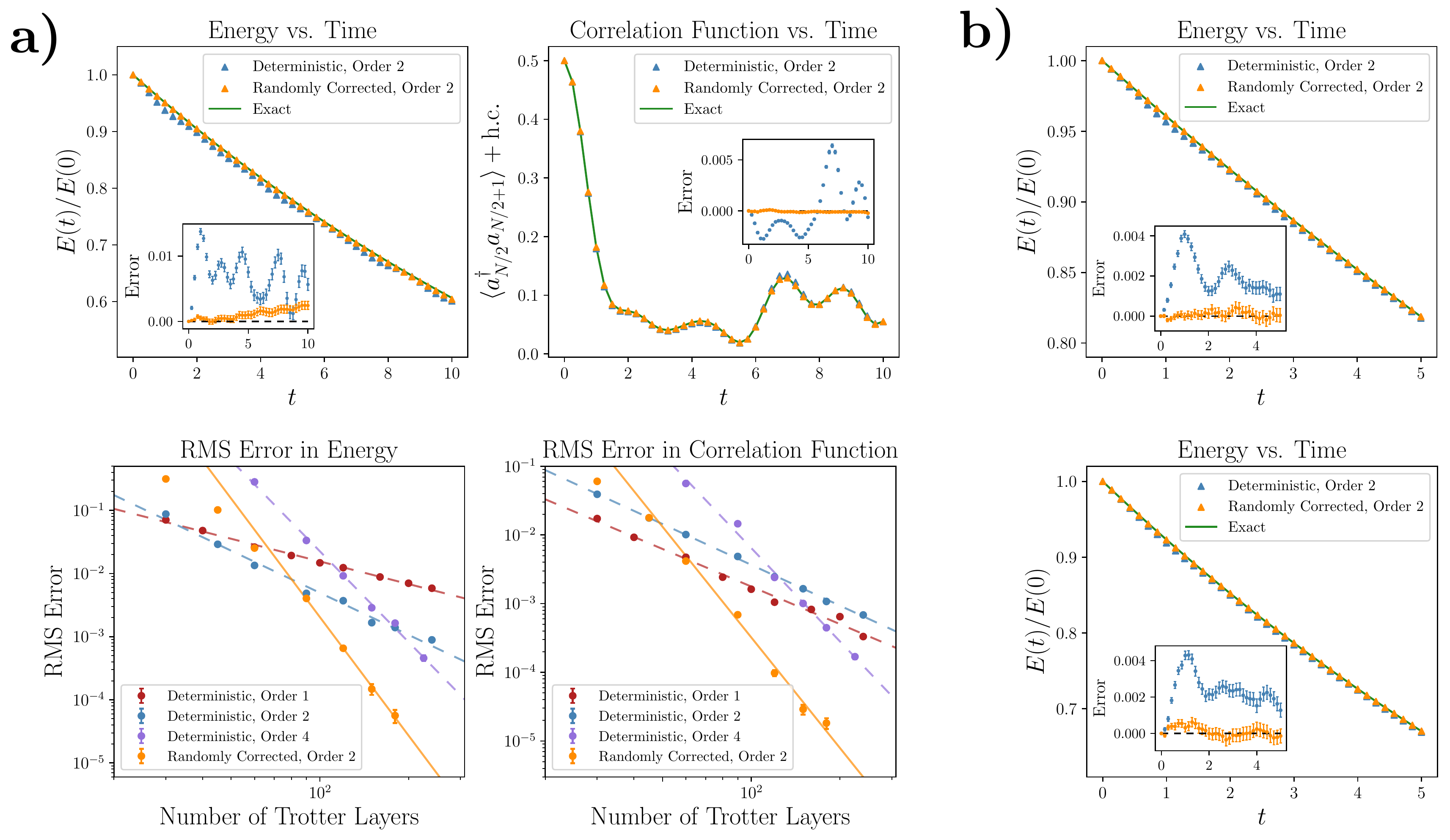}
    \caption{Comparison between the standard (deterministic) and randomly compiled quantum trajectories methods applied to the dissipative Bose-Hubbard model. 
    \textbf{a)}: Simulations conducted on $N=10$ spins with dissipation rate $\gamma = 0.05$, using Trotterization of various orders. We illustrate the time series for energy and the correlation function $\langle a_{N/2}^\dag a_{N/2+1} \rangle + \text{h.c.} $ at a fixed number of Trotter steps ($r=40$) ({top}), and the corresponding RMS error as a function of computational cost ({bottom}).
    \textbf{b)}: Simulations conducted on $N=30$ spins with dissipation rates $\gamma = 0.04$ ({top}) and $\gamma = 0.08$ ({bottom}), using 2nd-order Trotterization. Here we illustrate the time series for energy at a fixed number of Trotter steps ($r=35$). 
    }
    \label{fig:QTM}
\end{figure*}

The randomly compiled quantum trajectories method suppresses the Trotter error suffered in estimating observables. This enables the use of larger time steps while maintaining the same level of accuracy, thereby accelerating the overall algorithm. When combined with tensor network implementations of the quantum trajectories method, this reduces the number of tensor contractions and singular value decompositions needed to simulate time evolution.

To demonstrate the advantage of randomly compiled quantum trajectories, we consider a hardcore Bose-Hubbard model described by the Hamiltonian:
\begin{equation}
    \hat{H} = -J \sum_{i=1}^{N-1} \big( a_{i+1}^\dag a_i + a_i^\dag a_{i+1} \big), \quad (a_i^\dag)^2 = 0 ,  
\end{equation}
for a coupling $J>0$, and bosonic creation/annihilation operators that obey $[a_i, a_j^\dag] = \delta_{ij}$ for $i\neq j$ and $\{ a_i , a_j \}=1$. We further take this system to experience spontaneous emission at a rate $\gamma$, which localizes particles and is modeled by jump operators $c_i = \sqrt{\gamma} a_i^\dag a_i $ at each site $i$~\cite{Daley_2014}. This causes the system's energy to dissipate over time; analytically, it is known that the energy decays as $E(t) = \tr\big(\hat{H} \rho(t)\big) = E(0)e^{-\gamma t}$~\cite{Daley_2014}. Further, we fix coupling strength $J=1$ and choose an initial state with two excitations in the center of the chain: $|\psi(0)\rangle = |0\rangle^{\otimes (N/2-1)} |+\rangle |+\rangle |0\rangle^{\otimes (N/2-1)}$.

In our simulations, we use matrix product states of bond dimension $\chi=10$, which captures the time-evolved state with small truncation error $\lesssim 10^{-6}$. To simulate time evolution, we use TEBD with both deterministic Trotterization up to $4$th-order and randomly corrected Trotterization of $2$nd-order. To reduce computational cost, we use a modified version of randomly corrected Trotterization that we derived in Appendix~\ref{app:Efficient_randomly_corrected_Trotter}, which requires only a single call to the 2nd-order Trotter formula $S_2(\Delta t)$ per time step instead of the two calls required in Eq.~\eqref{eq:randomly_corrected_Trotter_ops}. As we explain in Appendix~\ref{app:Efficient_randomly_corrected_Trotter}, this modified version suppresses error to $\mathcal{O}(\Delta t^5)$ and thus achieves performance similar to $4$th-order Trotterization.

We first study a small system of $N=10$ spins where we can benchmark against the exact solution, and consider weak dissipation $\gamma = 0.05$ and a total time $t=10$. We show results in Fig.~\ref{fig:QTM}a). In the top plots, we display the time series for energy and a correlation function $\langle a_{N/2}^\dag a_{N/2 + 1}  \rangle + \text{h.c.}$, computed using deterministic and randomly corrected Trotterization of 2nd-order, both at $r=40$ Trotter steps. While both methods follow the exact solution, randomly corrected Trotterization achieves error an order of magnitude smaller in the energy than the deterministic approach, and several-fold smaller in the correlation function. To analyze this error more closely, in the bottom plots we showcase the root-mean-square (RMS) error in both time series as a function of the number of Trotter layers per trajectory (i.e., the number of layers of two-site operations applied to an MPS, which is proportional to the total number of operations). While this neglects application of the random correction operation, doing so is admissible because the correction operation is a single tensor applied per time step, whose cost does not scale extensively with system size and contributes only negligibly to the runtimes of our simulations. From these results, we see that randomly corrected Trotterization is superior at achieving small errors and attains the most favorable asymptotic scaling, with a convergence that is visibly steeper than that of deterministic 4th-order Trotterization. This improved performance is important when high-accuracy simulations are required.

We also explore systems beyond exact diagonalization by simulating $N=30$ spins at dissipation rates $\gamma = 0.04$ and $\gamma = 0.08$. While we cannot compute expectation values exactly on a system of this size, we can use the analytic result that the energy decays as $E(t) = E(0)e^{-\gamma t}$. We plot the decay of energy in Fig.~\ref{fig:QTM}b) using both standard and randomly corrected Trotterization at the same number of Trotter steps. Mirroring the previous plots, we again see that randomly compiled quantum trajectories achieves lower systematic error than the standard method.

Looking forward, the improvement provided by randomly compiled quantum trajectories could be extended to other time evolution algorithms beyond TEBD. For instance, extension to the time-dependent variational principle (TDVP)~\cite{Paeckel_2019} could entail randomizing over the differential equation solver used to integrate the Schr\"odinger equation. 
In addition, randomization could be incorporated into the truncation step present in both TEBD and TDVP to suppress truncation error, as suggested in Refs.~\cite{ferris2015unbiased, harrow2025randomized, wang2025faster}, but we leave this to future work. 

\section{Conclusions and Outlook}\label{sec:Conclusions}

The central lesson of this work is that randomness is a computational resource that can be traded against systematic error. In what follows, we first recapitulate the main results that establish RC-QMC as a general framework for harnessing this resource, then survey the breadth of fields and problems that stand to benefit from it, and finally outline the improvements and conceptual extensions that lie ahead. This culminates in a discussion of how entropy lies at the heart of our methods and serves as a bridge from classical simulation to quantum information.

\emph{Summary of results.}  
We have demonstrated that well-placed randomness can speed up computations in quantum Monte Carlo. By developing RC-QMC as a general framework for reducing systematic error in QMC, we enabled more accurate and efficient simulations of quantum systems. A key reason for this improvement is that QMC algorithms already repeatedly sample from a distribution, which allowed us to weave in randomized compiling to further suppress errors at negligible additional cost. 

In general, effective deployment of RC-QMC requires tailoring the choice of randomized compiling technique to the specific Hamiltonian or dominant source of error. As our experiments showed, methods like QDrift are advantageous for long-range Hamiltonians, whereas Hamiltonians highly susceptible to Trotter error benefit from randomly corrected Trotterization. Because classical algorithms like QMC currently provide important insights into many quantum phenomena, these gains directly improve our ability to study complex quantum systems and benchmark early quantum computers.

\emph{Applications to other algorithms and fields.} 
RC-QMC extends beyond the applications and methods presented here. As introduced in Sec.~\ref{sec:RC_QMC}, the framework is general and fully compatible with other QMC algorithms. For example, 
auxiliary-field QMC~\cite{zhang201315, Lee_2022} replaces inter-particle interactions with couplings to a random auxiliary field, and could perhaps admit an improvement by randomizing over that field. Similarly, machine-learning-based simulations using variational Monte Carlo~\cite{Carleo_2017} may benefit from randomizing over variational parameters, with initial steps in this direction recently taken in Ref.~\cite{martyn2025_HighAcc}. We also envision that randomized compiling techniques may prove beneficial to other classical simulation methods beyond QMC, such as Pauli propagation~\cite{Rall_2019, rudolph2025pauli} and density functional theory~\cite{parr1989density, koch2015chemist, teale2022dft}. Moreover, RC-QMC can incorporate randomized compiling methods beyond random Trotter formulas, such as those for rotations~\cite{hastings2016turning, Campbell_2017} and polynomial transformations~\cite{martyn2024halving}.

Beyond new algorithms, RC-QMC holds utility in further physical settings. A promising direction is using randomly compiled path integral Monte Carlo to simulate lattice field theories, such as those describing quantum chromodynamics~\cite{rothe2012lattice}, nuclear effective field theories~\cite{LEE2009117}, superconducting circuits~\cite{lin2025latticefieldtheorysuperconducting}, and other field theories in high energy physics~\cite{Ba_Omar_2022_QTRAJ,DeGrand:2015zxa,Laiho:2016nlp}. 
It may also prove worthwhile to investigate the structural similarities between RC-QMC and random lattice field theories, where averaging over instantiations of random spatial lattices restores rotational symmetry in the effective action~\cite{Christ:1982zq,Christ:1982ck}. 
Given the substantial computational demands of all these simulations, the cost reduction afforded by randomized compiling could provide a practical advantage in probing field theories.

\emph{Improvements.} 
Several practical improvements remain. Our current RC-QMC methods, and their application to the more advanced problems above, would likely benefit from additional code optimization, especially through parallel computing or GPU acceleration. Extending RC-QMC to new settings will also require the development of new randomized compiling schemes: while the present work leveraged existing constructions for real- and imaginary-time evolution, additional techniques are needed to approximate other transformations in QMC algorithms. These constructions can generally be derived analytically by following the framework outlined in Sec.~\ref{sec:RC_QMC} or satisfying the conditions established in Ref.~\cite{Campbell_2017}. Alternatively, Refs.~\cite{Cho_2024, martyn2024halving} propose randomized compiling constructions that are sufficiently generalizable to other transformations.

\emph{Outlook.}  
It is often remarked that quantum computing arose from the belief that simulations of Nature should be made quantum mechanical~\cite{feynman1982simulating}. While this is commonly taken to mean that simulating Nature requires a quantum computer, an alternative interpretation is that quantum mechanics should inform classical simulation methods. Indeed, the present work supports this perspective by leveraging randomized compiling, a concept that originated in quantum computing, to accelerate classical simulations of quantum systems.

At a deeper level, the unifying theme of this work is the understanding and harnessing of entropy (arising from randomness) to reduce computational error. We have focused on classical methodologies, in which the randomness is supplied by a classical source and the averaging is performed over many sequential, time-like runs. Yet this perspective carries natural implications for quantum information. The classical randomness that powers randomized compiling could instead be furnished by entanglement. Namely, through purification, a classical probability distribution becomes a coherent superposition over an ancillary register, from which the required randomness is recovered upon tracing it out. If the many sequential runs that we average over were replaced by many parallel, simultaneous runs, then this would invite an interpretation of our techniques (and of prior methods like QDrift) as a kind of error correction, since both ultimately result in a reduction of error. The analogy is imperfect but instructive: here the errors are biases, which are systematic and fixed rather than stochastic, so what randomization accomplishes is not to detect and reverse errors but to prevent these biases from accumulating coherently. Making this connection precise, and determining whether genuine quantum resources can sharpen it, is a compelling direction for future work.

\vspace{6pt}
\emph{Code Availability}---Code to run the experiments conducted in this paper is available at \url{https://github.com/jmmartyn/RandomlyCompiledQMC}.

\vspace{6pt}
\emph{Acknowledgments}---JMM thanks Nathan Wiebe, Aram Harrow, Brenda Rubenstein, Angus Lowe, and Megan Masterson for helpful discussions.

Portions of this work were conducted in MIT's \textit{Center for Theoretical Physics, a Leinweber Institute} and partially supported by the U.S.~Department of Energy under Contract No.~DE-SC0012567.

This research was supported by PNNL's Quantum Algorithms and Architecture for Domain Science (QuAADS) Laboratory Directed Research and Development (LDRD) Initiative. This material is based upon work supported by the U.S.~Department of Energy, Office of Science, National Quantum Information Science Research Centers, Quantum Science Center (QSC). The Pacific Northwest National Laboratory is operated by Battelle for the U.S.~Department of Energy under Contract DE-AC05-76RL01830.

While at MIT, JL's contributions were supported by the U.S.~Department of Energy under Contract No.~DE-SC0011090, by the SciDAC5 award DE-SC0023116, and additionally by the National Science Foundation under Cooperative Agreement PHY-2019786 (The NSF AI Institute for Artificial Intelligence and Fundamental Interactions, http://iaifi.org/). 
While at Argonne National Laboratory, JL's contributions were supported by the U.S.~Department of Energy, Office of Science, Office of Nuclear Physics through Contract No.~DE-AC02-06CH11357, and by the U.S.~Department of Energy, Office of Science, Office of Nuclear Physics, Early Career Award through Contract No.~DE-SCL0000017. 

NCW is supported by the U.S.~Department of Energy, Office of Science under grant Contract Numbers DE-SC0011090 and DE-SC0021006, the Simons Foundation grant 994314 (Simons Collaboration on Confinement and QCD Strings), and the U.S.~Department of Energy, Office of Science, National Quantum Information Science Research Centers, Co-Design Center for Quantum Advantage under Contract No. DE-SC0012704. 

IC was supported in part by the U.S.~Department of Energy, Office of Science, National Quantum Information Science Research Centers, Co-Design Center for Quantum Advantage (C$^2$QA) under Contract No.~DE-SC0012704. 

Work at the University of Oxford was supported by the EPSRC through the QQQS programme grant (EP/Y01510X/1) and by the QCi3 hub (EP/Z53318X/1).

\bibliography{References}

\section*{Appendices}

\appendix

The appendices collect supplementary material referenced throughout the main text.
Appendix~\ref{app:notation} lists notation used throughout the paper.
Appendix~\ref{app:concepts} reviews relevant concepts from quantum information, including the spectral norm, trace norm, and trace distance, together with the bound on observable differences that underlies the error analysis.
Appendices~\ref{app:PIMC} and~\ref{app:QTM} give detailed reviews of path integral Monte Carlo and the quantum trajectories method, respectively, expanding on the summaries in Sec.~\ref{sec:QMC}.
Appendix~\ref{app:VMC} illustrates the general RC-QMC framework on variational Monte Carlo as an additional example.
Appendix~\ref{app:mixing_lemma} formalizes randomized compiling through a mathematical statement known as the \textit{mixing lemma}, and presents its generalization to non-unitary maps. 
Appendix~\ref{app:RandomCompilingTrotter} reviews the randomized Trotter methods on which randomly compiled path integral Monte Carlo and randomly compiled quantum trajectories method are built: QDrift for imaginary-time evolution and randomly corrected Trotterization for real-time evolution.
Appendix~\ref{app:Error_Bound_RC_PIMC} provides the error bound for the thermal state approximation used in randomly compiled path integral Monte Carlo.
Appendix~\ref{app:Stat_Error} analyzes the additional statistical error introduced by averaging over approximations in RC-QMC, showing it to be negligible.

\section{Notation}\label{app:notation}
For reference, here we list the notation used in the main text. 
\begin{itemize}
    \itemsep1pt
    \item[-] $\rho$: target quantum state (e.g., thermal state, time-evolved state)
    \item[-] $\hat{O}$: Observable
    \item[-] $\hat{H}$: Hamiltonian
    \item[-] $X$: Random variable, used in the general framework of QMC
    \item[-] $\mathcal{P}$: Probability distribution from which $X$ is sampled
    \item[-] $F(\hat{O}, X) $: Estimator function of observable $\hat{O}$ and random variable $X$ used in the general framework of QMC
    \item[-] $\widetilde{\mathcal{P}}_j$: An approximate probability distribution, indexed by $j$
    \item[-] $p_j$: Probability distribution over approximations indexed by $j$
    \item[-] $\beta$: Inverse temperature of thermal state
    \item[-] $\tau$: Imaginary time
    \item[-] $\mathbf{X}$: Array of classical paths, whose $i$th column encodes the path of particle $i$
    \item[-] $S(\mathbf{X})$: Action of classical paths $\mathbf{X}$
    \item[-] $t$: Real time
    \item[-] $r$: Number of Trotter steps
    \item[-] $n$: Trotter order
    \item[-] $\Delta \beta = \frac{\beta}{r}$, $\Delta t = \frac{t}{r}$: time steps in imaginary and real time, respectively
    \item[-] $L$: number of terms in the Hamiltonian decomposition
    \item[-] $\hat{H} = \sum_{j=1}^L \hat{H}_j = \sum_{j=1}^L h_j \hat{\mathcal{H}}_j$: decomposition of the Hamiltonian into $L$ constituent terms, where $\hat{\mathcal{H}}_j$ is the normalized term obeying $\| \hat{\mathcal{H}}_j \|=1$, and $\hat{H}_j = h_j \hat{\mathcal{H}}_j$ for $h_j > 0$.  
    \item[-] $\lambda = \sum_j h_j$: 1-norm of Hamiltonian coefficients
    \item[-] $\mathbf{j} = (j_1, j_2, ..., j_r)$: multi-index composed of integers $j_k$ corresponding to each randomly sampled Trotter formula\footnote{Note that in randomly compiled path integral Monte Carlo, the length of $\mathbf{j}$ is actually $r/2$ (see Eq.~\eqref{eq:Z_j_equation}), whereas in randomly compiled quantum trajectories, the length is $r$ (see Eq.~\eqref{eq:RC_QTM}).}
    \item[-] $N$: System size (number of particles, or lattice sites)
    \item[-] $D$: Hilbert space dimension
    \item[-] $c_l$: Jump operator (or Lindblad operator)
    \item[-] $\hat{H}_{\text{eff}} = \hat{H} - \frac{i}{2} \sum_l c_l^\dag c_l$: Effective Hamiltonian used to model open system dynamics
    \item[-] $t_{1:m} = \{ t_1, t_2, ..., t_m \}$: Set of $m$ jump times 
    \item[-] $|\phi(t|t_{1:m})\rangle$: trajectory state at time $t$, given jumps experienced at times $t_{1:m}$
\end{itemize}

\section{Concepts in Quantum Information}\label{app:concepts}

While quantum mechanics is often expressed in terms of pure states, a general quantum state takes the form of a mixed state:
\begin{equation}
    \rho = \sum_j p_j | \psi_j \rangle \langle \psi_j | , \qquad \tr(\rho) = 1 , 
\end{equation}
for a normalized probability distribution $p_j$, and corresponding pure states $|\psi_j\rangle$. Quantum states evolve under quantum channels. 
In this context, it will be convenient to denote a unitary operator by a Latin character, say $U$, and its associated channel by the corresponding calligraphic character: $\mathcal{U}(\rho) = U \rho U^\dag$. A general channel takes the form of a probabilistic mixture of unitary transformations:
\begin{equation}
    \Lambda(\rho) = \sum_{j=1}^M p_j W_j \rho W_j^\dag , 
\end{equation}
where $p_j$ is again a normalized probability distribution, and $\{ W_j \}_{j=1}^M$ are the associated unitaries. This channel may be realized on a quantum computer by sampling the index $j \sim p_j$ and evolving under the corresponding unitary $W_j$.

In analyzing quantum states and channels, we consider the spectral norm:
\begin{equation}
    \|U \| = \sup_{|\psi \rangle, \ \langle \psi | \psi \rangle = 1} \left\|U |\psi \rangle \right\| ,
\end{equation}
and the trace norm:
\begin{equation}
    \| U \|_1 = \tr \big( \sqrt{U^\dag U} \big).
\end{equation} 
We also consider a metric for distinguishing two states, say $\rho$ and $\sigma$, known as the \textit{trace distance}:
\begin{equation}\label{eq:trace_distance}
    d_{\text{tr}} (\rho,\sigma) = \frac{1}{2} \| \rho - \sigma \|_1 . 
\end{equation}
The trace distance places a bound on the difference in expectation values of an observable, $\hat{O}$, as
\begin{equation}\label{eq:obs_bound}
    \big| \text{tr}\big(\rho \hat{O}\big) - \text{tr}\big(\sigma \hat{O}\big) \big| \leq \| \hat{O} \| \| \rho - \sigma \|_1  = 2 \| \hat{O} \| d_{\text{tr}} (\rho,\sigma) .
\end{equation} 
This property makes the trace distance a natural metric for distinguishing quantum states: a small trace distance guarantees that expectation values of arbitrary observables differ by at most a proportional amount.

\section{Path Integral Monte Carlo}\label{app:PIMC}

This appendix provides additional details on path integral Monte Carlo beyond the summary in Sec.~\ref{subsec:PIMC}.
Appendix~\ref{app:PIMC_formalism} gives a self-contained derivation of the path integral Monte Carlo formalism, arriving at the discrete path integral and its action.

\subsection{Review of formalism}\label{app:PIMC_formalism}

In this appendix we provide additional details of the path integral Monte Carlo method. The derivation presented can be found in standard texts~\cite{feynman1965path}. 

As previously stated, the aim is to estimate a thermal expectation value $\frac{1}{Z} \text{tr}\big( e^{-\beta \hat{H}} \hat{O} \big)$ for an operator $\hat{O}$. This is conventionally accomplished by rewriting $e^{-\beta \hat{H}}=(e^{-\Delta \beta \hat{H}})^r$ for time step size $\Delta\beta = \beta/r$, inserting resolutions of the identity between each time step, and then utilizing a chosen Trotter formula. The form of the resolution of the identity depends on the context. For an $N$-particle system, a position-space resolution of $I = \int d\mathbf{x}~ |\mathbf{x}\rangle \langle \mathbf{x} |$ is often taken, where $|\mathbf{x}\rangle$ is the position-space eigenstate of all $N$ particles. For a lattice field theory, it is common to choose $I = \int \prod_x d \phi_x ~e^{-\sum_x |\phi_x|^2}|\phi\rangle \langle \phi |$, where $x$ denotes sites on the lattice and $|\phi\rangle$ are either bosonic or fermionic coherent states satisfying $\hat{\phi}_x|\phi \rangle = \phi_x|\phi \rangle$. For simplicity and concreteness, here we will consider a single-particle system where the Hamiltonian has the specific form $\hat{H} =  \hat{\mathbf{p}}^2/2 + V(\hat{\mathbf{x}}) = \hat{T} + \hat{V} $, and $\mathbf{x},\mathbf{p}$ denote the particle position and momentum in a $d$-dimensional space, but analogous results hold for $N$-particle systems and lattice field theories.

Considering an operator $\hat{O}$ that is diagonal in position space, one finds
\begin{align}
    &\frac{1}{Z} \text{tr}\big( e^{-\beta \hat{H}} \hat{O} \big)   = \nonumber \\
    & \frac{\int Dx ~\langle \mathbf{x}_0 | e^{-\Delta \beta \hat{H}} | \mathbf{x}_1 \rangle \dots \langle \mathbf{x}_{r-1} | e^{-\Delta \beta \hat{H} } | \mathbf{x}_r \rangle \hat{O}(\mathbf{x}_0)}{\int Dx~\langle \mathbf{x}_0 | e^{-\Delta \beta \hat{H}} | \mathbf{x}_1 \rangle \dots \langle \mathbf{x}_{r} | e^{-\Delta \beta \hat{H} } | \mathbf{x}_0 \rangle }
\end{align}
where $Dx=\prod_{k=0}^{r-1} d \mathbf{x}_k$ is the path integral measure over discrete paths and we use the symbol $x=(\mathbf{x}_0,\dots,\mathbf{x}_{r-1})$ to denote a discrete spacetime trajectory\footnote{Operators involving momenta involve additional steps but introduce no essential changes to the formalism.}. The equation above is an exact rewriting; however, it is not yet useful because computing any of the matrix elements is as difficult as solving the whole problem. Progress is made by approximating $e^{-\Delta \beta \hat{H} }$ with a Trotter formula, say 
\begin{equation}
    e^{-\Delta \beta H}=e^{-\Delta \beta \hat{V}/2 } e^{-\Delta \beta \hat{T} } e^{-\Delta \beta \hat{V}/2 } + \mathcal{O}(\Delta \beta^3)~.
\end{equation}
By then inserting a resolution of the identity in momentum space between the exponentials, and subsequently integrating over the introduced momenta, one finds
\begin{align}
    &\langle \mathbf{x}_k | e^{-\Delta \beta \hat{V}/2 } e^{-\Delta \beta \hat{T} } e^{-\Delta \beta \hat{V}/2 } | \mathbf{x}_{k+1} \rangle = \\
    &e^{-\Delta \beta \Big[ \frac{1}{2}\frac{(\mathbf{x}_{k+1} - \mathbf{x}_{k})^2}{\Delta \beta^2} + \frac{V(\mathbf{x}_{k}) + V(\mathbf{x}_{k+1})}{2}\Big]}~.
\end{align}
Using this formula in the numerator and denominator of the desired thermal average, one finds:
\begin{equation}
    \frac{1}{Z} \text{tr}\big( e^{-\beta \hat{H}} \hat{O} \big) 
    = \frac{\int Dx ~e^{-S(x)}O(\mathbf{x}_0)}{\int Dx~e^{-S(x)} } + \mathcal{O}(r \Delta \beta^2)
\end{equation}
where
\begin{equation}
    S(x) = \sum_{k=0}^{r-1}\Delta \beta\Big[\frac{1}{2} \frac{(\mathbf{x}_{k+1} - \mathbf{x}_{k})^2}{\Delta \beta^2} + V(\mathbf{x}_{k})\Big]
\end{equation}
is the Euclidean action, and $\mathbf{x}_{r}:=\mathbf{x}_0$. The first term in the action is a discrete version of the particle's kinetic energy $\frac{1}{2} v^2$, while the second term represents the potential energy. The periodic boundary conditions and sign flip between the kinetic and potential energies (relative to the Minkowski action) result from considering a thermal physics problem.

Having written the desired thermal average as a path integral, a common workflow is to compute the necessary path integrals at several different values of $\Delta \beta$, then extrapolate the results to the $\Delta \beta = 0$ limit. As long as the statistical errors of each individual calculation are much smaller than the systematic errors, and provided the chosen values of $\Delta \beta$ are small enough, the continuum extrapolated result will be dominated by statistical error alone, which can be systematically reduced by running longer simulations. In this sense path integral Monte Carlo is an \textit{ab initio} method free uncontrolled systematic errors. 

To compute the required path integrals, one notes that, for a real action, 
\begin{equation}
    p(x) = \frac{e^{-S(x)}}{\int Dx~e^{-S(x)}}
\end{equation}
is a probability distribution. This probability distribution is sampled with Markov chain Monte Carlo methods, and the path integral is computed by simply averaging over the collected samples. There are many such methods, including Metropolis-Hastings and hybrid Monte Carlo methods~\cite{robert2016metropolishastingsalgorithm, DUANE1987216}, which guarantee exact convergence of the numerical result. 

We conclude by noting that not all actions are real. When this occurs, the path integral is said to have a \textit{sign problem}. In this case Monte Carlo methods can still be used, by sampling according to the real part of the action and absorbing  phases into the definition of observables; however, the encountered phases dramatically increase the number of samples required to reach a desired level of statistical error~\cite{Troyer_2005}.

\section{The Quantum Trajectories Method}\label{app:QTM}

This appendix provides a comprehensive treatment of the quantum trajectories method, expanding on the summary in Sec.~\ref{sec:QTM}.
Appendix~\ref{app:QTM_open_sys} reviews open quantum system dynamics and the Lindblad master equation.
Appendix~\ref{app:QTM_method} derives the quantum trajectories unraveling and the Monte Carlo algorithm that samples it.
Appendix~\ref{app:QTM_tensor_networks} describes its implementation on tensor networks, including TEBD and the sources of error in this realization.

\subsection{Open System Dynamics}\label{app:QTM_open_sys}

While the dynamics of a closed quantum system are governed by the Schr\"odinger equation, the dynamics of an open quantum system, such as a system interacting with an environment, is governed by the \emph{master equation}. In Gorini-Kossakowski-Sudarshan-Lindblad form, the master equation for a state $\rho$ is 
\begin{equation}
    \frac{d}{dt} \rho = -i[\hat{H},\rho] + \sum_l \gamma_l \big( {c}_l \rho {c}_l^\dag - \frac{1}{2} ({c}_l^\dag {c}_l \rho + \rho {c}_l^\dag {c}_l) \big) , 
\end{equation}
where ${c}_l$ are a set of \emph{jump operators} (or Lindblad operators) that describe the dissipative dynamics, and $\gamma_l \geq 0$ are the dissipation rates~\cite{Daley_2014}. In the limit $\gamma_l \rightarrow 0$, the usual time-dependent Schr\"odinger equation is recovered.

For simplicity, we absorb the dissipative rates into the jump operators as $c_l \mapsto \sqrt{\gamma_l} c_l$, as in Sec.~\ref{sec:QMC}. The master equation may then be conveniently re-expressed as
\begin{equation}\label{eq:master_eq_reexpression2}
\begin{aligned}
    &\dot{\rho} = -i (\hat{H}_{\text{eff}}\rho - \rho \hat{H}_{\text{eff}}^\dag) + \sum_l c_l \rho c_l^\dag, \\
    &\hat{H}_{\text{eff}} = \hat{H} - \tfrac{i}{2} \sum_l c_l^\dag c_l . 
\end{aligned}
\end{equation}
In this form, open system dynamics can be interpreted as a combination of evolution under the effective Hamiltonian $\hat{H}_{\text{eff}}$ and the action of the jump operators.

Satisfyingly, there exists an analytic solution to this form of the master equation~\cite{Dum_1992}. For the simplicity of its presentation, let us consider the scenario where there is only a single jump operator $c_1 =: c$. Here it will be convenient to define the channel corresponding to the action of the jump operator as
\begin{equation}
    \mathcal{C}(\rho) = c \rho c^\dag . 
\end{equation}
It will also be convenient to denote evolution under $\hat{H}_{\text{eff}}$ by $U^{[t,t']} = e^{-i\hat{H}_{\text{eff}}(t - t')}$, and define the corresponding channel as 
\begin{equation}
    \mathcal{U}^{[t,t']}( \rho ) = e^{-i \hat{H}_{\text{eff}} (t-t') } \rho e^{i\hat{H}_{\text{eff}} (t-t')} . 
\end{equation}
Then, for an initial state $\rho(0)$, the solution to the master equation can be written as
\begin{equation}\label{eq:master_eq_solution}
    \rho(t) = \sum_{m=0}^\infty \rho_m(t),
\end{equation}
where
\begin{equation}\label{eq:master_eq_solution_rho_n}
\begin{aligned}
    &\rho_m(t) = \int_0^t dt_m \int_0^{t_m} dt_{m-1} ... \int_0^{t_2} dt_1  \\ 
    &\qquad \qquad \qquad  \times \mathcal{U}^{[t,t_m]} \mathcal{C} \mathcal{U}^{[t_m, t_{m-1}]} ... \mathcal{C} \mathcal{U}^{[t_1, 0]} (\rho(0))
\end{aligned}
\end{equation}
is the contribution to the density matrix that experiences $m$ jumps (i.e., applications of $\mathcal{C}$) in the time interval $[0,t]$. Specifically, the state $\mathcal{U}^{[t,t_m]} \mathcal{C} \mathcal{U}^{[t_m, t_{m-1}]} ... \mathcal{C} \mathcal{U}^{[t_1, 0]} (\rho(0))$ experiences jumps at times $t_i \in [0, t]$ for $i=1,...,m$, and undergoes evolution under the effective Hamiltonian in between these jump times. The state $\rho_m(t)$ is an integral over all such jump times, and thus accounts for all possible combinations of $m$ jumps. For notational simplicity, let us denote the set of jump times as $t_{1:m} = \{ t_1, ..., t_m\}$.

When the initial state is a pure state, say $\rho(0) = |\psi \rangle \langle \psi|$, the time-evolved state may be conveniently expressed as 
\begin{equation}\label{eq:rho_reexpression_pureinput}
\begin{aligned}
    &\rho(t) = \sum_{m=0}^\infty \int dt_m ... dt_1 \ P_{ [0,t) }(t_1, ..., t_m) \\ 
    & \hspace{110pt} \times | \varphi(t|t_{1:m}) \rangle \langle \varphi (t|t_{1:m}) | ,
\end{aligned}
\end{equation}
where 
\begin{equation}\label{eq:traj_pure_states}
\begin{aligned}
    &| \varphi(t | t_{1:m}) \rangle =  \frac{| \phi(t | t_{1:m}) \rangle}{ \big\| | \phi(t | t_{1:m}) \rangle \big\| } , \text{ and } \\
    & | \phi(t | t_{1:m}) \rangle = U^{[t,t_m]} c U^{[t_m,t_{m-1}]} c \ ... \  c U^{[t_1,0]} |\psi \rangle , 
\end{aligned}
\end{equation}
are the normalized and unnormalized states at time $t$ that experience $m$ jumps at times $t_{1:m}$ and otherwise evolve under $\hat{H}_{\text{eff}}$, and 
\begin{equation}\label{eq:unnorm_state_distribution}
    P_{[0,t)}(t_{1:m}) = \langle \phi(t|t_{1:m})  | \phi(t| t_{1:m}) \rangle 
\end{equation}
is the probability density of $m$ jumps occurring at times $t_{1:m}$. As a result, we may express the expectation value of an observable $\hat{O}$ as an average over both the number of jumps $m$, and their corresponding times $t_{1:m}$ drawn from $P_{[0,t)}$:
\begin{equation}\label{eq:traj_exp_val}
\begin{aligned}
    \langle \hat{O}(t) \rangle &= \tr(\hat{O} \rho(t)) \\
    &= \mathop{\mathbb{E}}_{m, \ t_{1:m} \ \sim \ P_{[0,t)} } \Big[ \langle \varphi(t|t_{1:m}) | \hat{O} | \varphi(t|t_{1:m}) \rangle \Big]. 
\end{aligned}
\end{equation}
In other words, an expectation value of the time-evolved mixed state $\rho(t)$ is an average of expectation values over pure states $| \varphi(t|t_{1:m}) \rangle $. In the main text Eq.~\eqref{eq:master_eq_obs}, we denoted this expectation value as $|\phi\rangle \sim \langle \phi | \phi \rangle$, meaning that the states $|\phi (t|t_{1:m}) \rangle$ are sampled proportional to their norm $ \langle \phi(t|t_{1:m})  | \phi(t|t_{1:m}) \rangle $. This was done for notational simplicity, but conveys the same idea as Eqs.~\eqref{eq:unnorm_state_distribution} and~\eqref{eq:traj_exp_val} that the jumps and their times are drawn from the norm of the time-evolved state.

\subsection{The Quantum Trajectories Method}\label{app:QTM_method}
This solution to the master equation furnishes the following Monte Carlo algorithm for simulating open system dynamics. While the exact solution is explicitly provided by Eqs.~\eqref{eq:master_eq_solution} and \eqref{eq:master_eq_solution_rho_n}, evaluating the time-evolved state $\rho(t)$ is notably nontrivial. Although there exist methods for classically approximating a density matrix (e.g., tensor networks), computing the integral of Eq.~\eqref{eq:master_eq_solution_rho_n} is in general challenging.

To circumvent this difficulty and reduce computational costs, the quantum trajectories method~\cite{Dum_1992, Dalibard_1992, Carmichael1993, Molmer_1993, Daley_2014} proposes to simulate $\rho(t)$ by evolving an ensemble of pure states $|\phi (t| t_{1:m}) \rangle$, whose average equates to $\rho(t)$. As shown in Eq.~\eqref{eq:master_eq_reexpression2}, these states evolve under the effective Hamiltonian $\hat{H}_{\text{eff}}$ with the jump operators $c$ applied randomly throughout the evolution. The advantage of the quantum trajectories method is that one only needs to simulate pure quantum states rather than a density matrix, which drastically reduces memory consumption.

To implement the quantum trajectories method, one creates an ensemble of $N_{\text{traj}}$ states, or equivalently trajectories, each beginning in the initial state $|\psi\rangle$. Each trajectory is then evolved according to the interpretation of $|\phi (t|t_{1:m})\rangle $: they experience stochastic jumps at times sampled from $P_{[0,t)}$, and otherwise evolve under the effective Hamiltonian. Denoting the normalized final states as $|\varphi^{(s)} (t) \rangle$ for $s=1,...,N_{\text{traj}}$, their average approximates the time-evolved density matrix:
\begin{equation}\label{eq:empirical_QTM}
    \rho(t) \approx \frac{1}{N_{\text{traj}}} \sum_{s=1}^{N_{\text{traj}}} | \varphi^{(s)}(t) \rangle \langle \varphi^{(s)} (t) |,
\end{equation}
with statistical error that decreases as $\mathcal{O}(1/\sqrt{N_{\text{traj}}})$. Accordingly, one can estimate an expectation value as an empirical average over the trajectories: $\text{tr} \big(\hat{O} \rho(t) \big) \approx \frac{1}{N_{\text{traj}}} \sum_{s=1}^{N_{\text{traj}}} \langle \varphi^{(s)}(t) | \hat{O} | \varphi^{(s)} (t) \rangle$.

The nontrivial component of the quantum trajectories method is sampling jumps from the distribution $P_{[0,t)}(t_{1:m})$, which necessitates sampling both the number of jumps $m$ and their times $t_{1:m}$. This can be achieved auto-regressively by sampling a successive jump time, given the times of previous jumps. That is, given a trajectory that has already experienced jumps at times $ t_{1:i} $, the probability of a subsequent jump occurring at $t_{i+1} \in [t_i,t)$ is quantified by its cumulative distribution function, which is given by an integral of the conditional probability~\cite{Dum_1992}:
\begin{equation}\label{eq:p_jump}
\begin{aligned}
    \text{Prob}(t_{i+1}) &= \int_{t_i}^{t_{i+1}} dt' \frac{P_{[0,t)}(\overbrace{t_{1:i},t'}^{t_1,...,t_i,t'})}{P_{[0,t_i)}(t_{1:i})} \\
    &= 1- \frac{ \big\| | \phi(t_{i+1}|t_{1:i}) \rangle \big\|^2 }{ \big\| | \phi(t_i^+|t_{1:i}) \rangle \big\|^2}, \quad \ t_{i+1}\geq t_i ,  
\end{aligned}
\end{equation}
where $t_i^+$ denotes the time just after the $i$th jump operator is applied. Intuitively, this probability increases as the norm of the wave function decays due to dissipation experienced by evolution under $\hat{H}_{\text{eff}}$. In practice, one samples a successive jump time $t_{i+1}$ from $\text{Prob}(t_{i+1})$ using inverse transform sampling~\cite{Devroye_2006}: draw a uniformly distributed variable $y\sim \mathcal{U}(0,1)$, and then numerically solve $y = \text{Prob}(t_{i+1})$ for $t_{i+1}$. One then evolves to time $t_{i+1}$ and applies the corresponding jump operator. Afterwards, the trajectory state is re-normalized at time $t_{i+1}^+$, so as to update the cumulative distribution and determine the next jump time. This process is continued until reaching the final simulation time $t$.

Finally, although this presentation of the quantum trajectories method has specialized to a single jump operator and an initial state that is pure, this analysis naturally extends to multiple jump operators and initial states that are mixed. Given a set of multiple jump operators $\{c_l \}$, the algorithm is modified such that at a jump time $t_i$, a single jump operator is chosen at random to be applied. Specifically, a jump operator $c_l$ is selected with probability proportional to the norm of the state after the action of the jump: 
\begin{equation}
    \text{Prob}(c_l) \propto \langle \phi(t_i|t_{1:i-1}) | c_l^\dag c_l | \phi (t_i|t_{1:i-1}) \rangle .  
\end{equation}
Similarly, to simulate an initial mixed state $\rho(0) = \sum_j p_j |\psi_j \rangle \langle \psi_j | $, one instead takes the input state of each trajectory to be $|\psi_j \rangle$ with probability $p_j$.

This analysis of the quantum trajectories method provides an intuitive physical interpretation of open system dynamics and an efficient algorithm to simulate these dynamics. Due to these favorable properties, it has been cemented as a vital tool for the study of open systems. Noteworthy applications include its use in simulating laser cooling~\cite{Marte_1993, Castin_1995_Monte}, noisy random unitary circuits~\cite{Chen_2024_Optimized, Cheng_2023_Efficient}, and the quark-gluon plasma in quantum chromodynamics~\cite{Ba_Omar_2022_QTRAJ}.

\subsection{Implementation on Tensor Networks}\label{app:QTM_tensor_networks}
The quantum trajectories method requires a method for classically simulating pure states. A natural tool for this job is tensor networks, which provide an efficient classical representation of quantum states on a lattice~\cite{Klumper_1993, Vidal_2003, PerezGarcia_2007, Verstraete_2004_valence, verstraete2004renormalization, Vidal_2007_Entanglement, Vidal_2008_Class, Evenbly_2009_Alg}. Here we will focus on 1D lattices, where a tensor network known as the matrix product state (MPS)~\cite{Klumper_1993, Vidal_2003} is the standard tool. MPSs represent a quantum state in an exponentially large Hilbert space by factorizing it into a sequence of tensors associated with individual lattice sites. Adjacent tensors are contracted with each other across virtual bonds, whose dimension is restricted to be at most an integer $\chi$, known as the bond dimension. Equivalently, the Schmidt rank across any bipartition of the lattice is restricted to be at most $\chi$. Owing to this compressed representation, MPSs are widely used to model ground states of local Hamiltonians, simulate time evolution, and estimate thermal states.

In using MPSs to implement the quantum trajectories method, each trajectory state is modeled as an independent MPS. Then, evolution under $\hat{H}_{\text{eff}}$ can be straightforwardly simulated using a time evolution algorithm. A simple such algorithm is time-evolving block decimation (TEBD)~\cite{Vidal_2003, Vidal_2004, White_2004, Daley_2004}, which approximates the time evolution operation $e^{-i\hat{H}_{\text{eff}}t}$ with a Trotter formula. Formally, each Trotter step is applied to the MPS, after which the bond dimensions are truncated back down to $\chi$ by Schmidt decomposition. In this formulation, we can efficiently determine jump times according to Eq.~\eqref{eq:p_jump} by evaluating the norm of the MPS. In addition, the jump operators are often local and can be directly applied to the MPS, but still require truncating the bond dimension back down to $\chi$. 

There are three sources of error in this realization of the quantum trajectories method. First, over a time interval $t$, an order-$n$ Trotter formula implemented with $r$ time steps of size $\Delta t=t/r$ suffers error $\mathcal{O}(r (L \Delta t)^{n+1})$. Second, TEBD and application of the jump operators require truncation of the bond dimension. While the precise scaling of the resulting truncation error depends on the underlying state, this error decreases with increasing bond dimension, generally decaying exponentially with $\chi$ for weakly entangled states~\cite{Vidal_2004, Bridgeman_2017, Orus_2014}. Consequently, approaches to reduce the cost of the quantum trajectories method have studied ways to unravel the time-evolved state into weakly entangled trajectories, which can be more efficiently represented at low bond dimension~\cite{Chen_2024_Optimized, sander2026computational}. Finally, implementing the quantum trajectories method with a finite number of trajectories $N_{\text{traj}}$ incurs a statistical error $\mathcal{O}(1/\sqrt{N_{\text{traj}}})$.

Separately, while here we use Trotterization to perform time evolution, there exist other time evolution methods for tensor networks, like the time-dependent variational principle~\cite{Paeckel_2019}, which can be advantageous in certain settings. Nonetheless, Trotterization is simple and widely applicable to various systems, so we use it here rather than other time evolution algorithms.

\section{Variational Monte Carlo}\label{app:VMC}
In Sec.~\ref{sec:QMC_General}, we showed how the general framework for QMC encompasses path integral Monte Carlo and the quantum trajectories method. To further emphasize the generality of this framework, here we also show how it encompasses variational Monte Carlo.

Variational Monte Carlo is a variational method for estimating ground states of quantum systems. In variational Monte Carlo, a quantum state $|\psi_\theta \rangle$ is represented classically by a variational ansatz and parameterized by $\theta$. For simplicity of exposition, we assume $|\psi_\theta\rangle$ is normalized (if not, each expectation value below is implicitly divided by $\langle \psi_\theta | \psi_\theta \rangle$). The ground state is then approximated by minimizing the energy $E_\theta = \langle \psi_\theta |H|\psi_\theta \rangle$, typically using gradient-based optimization methods, like gradient descent~\cite{Stokes_2020} or stochastic reconfiguration~\cite{Sorella_1998}.

To compute an expectation value of an operator $\hat{O}$, Markov chain Monte Carlo methods are employed. An expectation value may be expressed in terms of the wave function's degrees of freedom $x$ (e.g., particle positions or spin configurations) as
\begin{equation}\label{eq:VMC_eq}
\begin{aligned}
    \langle \psi_\theta | \hat{O} | \psi_\theta \rangle  &=
    \int dx \ \psi^*_{\theta}(x) \hat{O}(x) \psi_\theta (x) \\
    &= 
    \mathop{\mathbb{E}}_{x \sim |\psi_{\theta}(x)|^2} \big[ O_{\text{loc}}(x) \big] ,
\end{aligned}
\end{equation}
where $O_{\text{loc}}(x) = \langle x| \hat{O}|\psi_\theta\rangle / \langle x | \psi_\theta \rangle$ is a scalar function known as the \textit{local observable}. Variational Monte Carlo uses this expression to estimate expectation values as averages of a local observable, sampled according to the square of the variational wave function $|\psi_\theta (x)|^2$.

Systematic error arises due to the variational wave function deviating from the true ground state. In the ideal case where the variational state is exactly the ground state, $|\psi_\theta \rangle = | E_0\rangle$, a ground state observable can be evaluated as
\begin{equation}
\begin{aligned}
    \langle E_0 | \hat{O} | E_0 \rangle  =
    \mathop{\mathbb{E}}_{x \sim |E_0(x)|^2} \big[ O_{\text{loc}}(x) \big] . 
\end{aligned}
\end{equation}
In comparison with the general QMC framework of Eq.~\eqref{eq:QMC}, we see that the random variable is the configuration $X=x$, sampled from the distribution $\mathcal{P} = |E_0(x)|^2$, and the estimator function is the local observable $F(\hat{O}, X) = O_{\text{loc}}(x) = \langle x| \hat{O}| \psi_\theta \rangle / \langle x | \psi_\theta \rangle$.

In practice, however, the ground state can only be approximated by the variational wave function $|\psi_\theta\rangle \approx |E_0\rangle$, which is achieved by minimizing the energy. Therefore, ground state observables are estimated using Eq.~\eqref{eq:VMC_eq} after minimizing the energy: 
\begin{equation}
\begin{aligned}
    \langle E_0 | \hat{O} | E_0 \rangle \approx
    \mathop{\mathbb{E}}_{x \sim |\psi_{\theta}(x)|^2} \big[ O_{\text{loc}}(x) \big] . 
\end{aligned}
\end{equation}
In comparison with the general framework of QMC, this suffers systematic error due to sampling the approximate distribution $\widetilde{\mathcal{P}} = |\psi_{\theta}(x)|^2$.

Given this error, recent work has shown that randomized compiling can improve the accuracy of ground state observables in variational Monte Carlo~\cite{martyn2025_HighAcc}. In short, the approach follows the RC-QMC framework by averaging observables over an ensemble of variational wave functions that each approximates the ground state (obtained by time-evolving $|\psi_\theta\rangle$ for a randomly sampled duration). This ensemble averaging yields estimates of ground state observables that are more accurate than any individual variational wave function.

\section{The Mixing Lemma and Its Generalization}\label{app:mixing_lemma}
In Sec.~\ref{sec:randomized_compiling}, we introduced randomized compiling for approximating unitary transformations. Here we present it as a formal statement, known as the Hastings-Campbell \textit{mixing lemma}~\cite{hastings2016turning, Campbell_2017}. We also present a generalization of the mixing lemma to non-unitary maps, reproduced from Ref.~\cite{martyn2024halving}. 
This generalization is important as it underlies the randomly-compiled quantum trajectories method developed in Sec.~\ref{sec:RC_QTM_subsec}.

\subsection{The Mixing Lemma}

First we present the mixing lemma:
\begin{lemma}[Hastings-Campbell Mixing Lemma~\cite{Campbell_2017, hastings2016turning}]\label{lemma:Mixing}
    Let $U$ be a target unitary operator, and $\mathcal{U}(\rho) = U \rho U^\dag $ the corresponding channel. Suppose there exist $M$ unitaries $\{ W_j \}_{j=1}^M$ and an associated probability distribution $p_j$ that approximate $V$ as
    \begin{equation}
    \begin{aligned}
        & \| W_j - U \| \leq a \text{ for all } j,\\
        & \Big\| \sum_{j=1}^M p_j W_j - U \Big\| \leq b ,
    \end{aligned}
    \end{equation}
    for some $a,b > 0$. Then, the corresponding channel $\Lambda(\rho) = \sum_{j=1}^M p_j W_j \rho W_j^\dag$ approximates $\mathcal{U}$ as
    \begin{equation}
        \| \Lambda (\rho) - \mathcal{U} (\rho) \|_1 \leq a^2 + 2b  
    \end{equation}
    for any input state $\rho$. 
\end{lemma}
See Refs.~\cite{Campbell_2017, hastings2016turning} for the proofs of this lemma.

\subsection{The Generalized Mixing Lemma}

Next, we consider generalizing the mixing lemma to arbitrary operators, including non-unitary operators, following Ref.~\cite{martyn2024halving}. Suppose we we wish to implement a channel $\mathcal{A}(\rho) = A \rho A^\dag$, where $A$ is not necessarily unitary, yet we only have access to operators $B_j$ that approximate $A$. Then the following statement is true:
\begin{lemma}[Generalized Mixing Lemma]\label{thm:NonUnitaryMixing_app}
    Let $A$ be a target operator, possibly non-unitary, and $\mathcal{A}(\rho) = A \rho A^\dag $ the corresponding map. Suppose there exist $M$ operators $B_j$ and a probability distribution $p_j$ that approximate $A$ as
    \begin{equation}
    \begin{aligned}
        & \| B_j - A\| \leq a \text{ for all } j,\\
        & \Big\| \sum_{j=1}^M p_j B_j - A \Big\| \leq b ,
    \end{aligned}
    \end{equation}
    for some $a,b > 0$. Then, the corresponding map $\Gamma(\rho) = \sum_{j=1}^M p_j B_j \rho B_j^\dag$ approximates $\mathcal{A}$ as
    \begin{equation}\label{eq:nonunitary_mixing_lemma_conclusion}
        \| \Gamma(\rho) - \mathcal{A}(\rho) \|_1 \leq a^2 + 2b \|A \|. 
    \end{equation}
    for any input state $\rho$. 
\end{lemma}

See Ref.~\cite{martyn2024halving} for the proof. Evidently, this is nearly equivalent to the mixing lemma, up to a rescaling by the spectral norm $\|A \|$, which captures the deviation of $A$ from non-unitarity.

In applying this to the randomly-compiled quantum trajectories method, our target operator is evolution under the effective Hamiltonian $\hat{H}_{\text{eff}} = \hat{H} - \frac{i}{2} \sum_l c_l^\dag c_l $: $A = e^{-i\hat{H}_{\text{eff}} \Delta t}$~\cite{Daley_2014}. Because the effective Hamiltonian generates non-unitary evolution and decreases the norm of the state it acts on, we have $\| A \| \leq 1$ (or more precisely $\| A \| \leq 1 - \mathcal{O}(\Delta t)$ where the prefactor of $\Delta t$ depends on the magnitude of $\sum_l c_l^\dag c_l$). Consequently, the conclusion of the generalized mixing lemma yields 
\begin{equation}
    \| \Gamma(\rho) - \mathcal{A} (\rho) \|_1 \leq a^2 + 2b \|A\| \leq a^2 + 2b , 
\end{equation}
just as in the unitary setting. Therefore, we can apply randomized compiling to the quantum trajectories method, analogous to how it is applied to problems involving unitary evolution. This justifies the development of the randomly-compiled quantum trajectories method in Sec.~\ref{sec:RC_QTM_subsec}.

\section{Randomized Compiling Methods for Trotterization}\label{app:RandomCompilingTrotter}

This appendix reviews the randomized Trotter methods used in the randomly compiled path integral Monte Carlo and the randomly compiled quantum trajectories method.
Appendix~\ref{app:RC_trotter_standard} reviews standard (deterministic) Trotterization and establishes the error scalings that the random methods improve upon.
Appendix~\ref{app:RC_trotter_qdrift} reviews QDrift, which removes explicit $L$-dependence from the cost by importance sampling individual Hamiltonian terms.
Appendix~\ref{app:RC_trotter_corrected} reviews randomly corrected Trotterization, which doubles the effective Trotter order by averaging over correction operations; Secs.~\ref{app:RandomCorrectionExample} and~\ref{app:Efficient_randomly_corrected_Trotter} work out the explicit 2nd-order construction and its more efficient single-call variant.

\subsection{Standard Trotterization}\label{app:RC_trotter_standard}

Trotterization is a method for approximating the exponential of a Hamiltonian $\hat{H}$ as a product of many local operations~\cite{Childs_2021}, and works in real time (i.e., $e^{-i\hat{H}t}$) and imaginary time (i.e., $e^{-\beta \hat{H}}$). In this section, we will focus on real-time evolution for simplicity, but these methods analogously extend to imaginary time.

Trotterization considers a Hamiltonian composed of $L$ local constituent terms: $\hat{H} = \sum_{j=1}^L \hat{H}_j$, where the constituent terms $\hat{H}_j$ may not commute, $[\hat{H}_j, \hat{H}_{j'}] \neq 0$. The idea of Trotterization is to divide the simulation time $t$ into $r$ steps of size $\Delta t := t/r$, and approximate each time step $e^{-i\hat{H} \Delta t }$ as a product of evolutions under the individual local terms $\hat{H}_j$. The first-order Trotter formula takes the form
\begin{equation}\label{eq:first_order_Trotter}
    e^{-i\hat{H}\Delta t} \approx \prod_{j=1}^L e^{-i \hat{H}_j \Delta t} =: S_1(\Delta t) ,
\end{equation}
and suffers error $\| e^{-i\hat{H}\Delta t}  - S_1(\Delta t) \| = \mathcal{O}(L^2 \Delta t^2)$. At the next level, second-order Trotterization is given by
\begin{equation}
    \begin{aligned}
        & S_2(\Delta t) = \prod_{j=1}^L e^{\frac{-i\Delta t}{2} \hat{H}_j } \prod_{j=L}^1 e^{\frac{-i\Delta t}{2} \hat{H}_j } .
    \end{aligned}
\end{equation}
At order-$n$ for even $n\geq 2$, the symmetric Suzuki-Trotter formulas are defined recursively as 
\begin{equation}\label{eq:order_2k_Trotter}
    \begin{aligned}
        & S_{n}(\Delta t) = \\
        & \qquad \quad S_{n-2}(\alpha_n \Delta t)^2 S_{n-2}((1-4\alpha_n) \Delta t) S_{n-2}(\alpha_n \Delta t)^2 ,
    \end{aligned}
\end{equation}
where $\alpha_n = 1/(4 - 4^{1/(n-1)})$ is a constant. These formulas incur error $\|  e^{-i\hat{H}\Delta t} - S_{n}(\Delta t) \| = \mathcal{O}((L\Delta t)^{n+1})$.

When compounding a Trotter formula over $r$ time steps to approximate $e^{-i\hat{H}t}$, order-$n$ Trotterization incurs error $\mathcal{O}( r (L \Delta t)^{n+1} ) = \mathcal{O}( r (Lt/r)^{n+1} )$. Equivalently, achieving error $\epsilon$ requires a number of time steps $r = \mathcal{O}\big(Lt (tL/\epsilon)^{1/{n}} \big)$. Moreover, by solving the recursion relation of Eq.~\eqref{eq:order_2k_Trotter}, $S_{n}(\Delta t)$ can be shown to consist of $ \mathcal{O}(5^{n/2} L)$ exponential operations~\cite{Berry_2006}, such that in general the total number of local operations required to realize order-$n$ Trotterization is
\begin{equation}\label{eq:Trotter_time_complexity}
    \mathcal{O}\left( 5^{n/2} L^2 t (tL/\epsilon)^{1/n} \right).
\end{equation}

In the context of quantum algorithms, Eq.~\eqref{eq:Trotter_time_complexity} corresponds to the number of gates required to approximate time evolution. In classical simulations, Eq.~\eqref{eq:Trotter_time_complexity} represents the number of matrix multiplications or operations required to evolve a state. In both cases, increasing the Trotter order reduces the cost through the term $(tL/\epsilon)^{1/n}$, but also increases the pre-factor $5^{n/2}$. In practice, balancing these tradeoffs makes second- and fourth-order formulas the most useful.

\subsection{QDrift}\label{app:RC_trotter_qdrift}

Recent works have used randomized compiling to improve the performance of Trotterization~\cite{Campbell_2019, Childs_2019, Ouyang_2020, nakaji2023qswift, Hagan_2023, Cho_2024, peetz2025hamiltonian}. The central idea of these methods is to randomize over Trotter formula to suppress the error suffered over a single time step.

A notable such method is \textit{QDrift}, a randomized first-order Trotter formula introduced in Ref.~\cite{Campbell_2019}. QDrift decomposes the Hamiltonian as $\hat{H} = \sum_{j=1}^L h_j \hat{\mathcal{H}}_j$ for $h_j>0$ and $\| \hat{\mathcal{H}}_j \| = 1$, and defines the 1-norm of the coefficients as $\lambda = \sum_{j=1}^L h_j$. It then proposes to simulate evolution for a time step $\Delta t$ via the channel
\begin{equation}\label{eq:QDrift_real_time}
    \rho \rightarrow \sum_{j=1}^L p_j e^{-i \hat{\mathcal{H}}_j \lambda \Delta t} \rho e^{i \hat{\mathcal{H}}_j \lambda \Delta t} , 
\end{equation}
where $p_j = h_j/\lambda$ is a probability distribution. This corresponds to importance sampling a single term $\hat{\mathcal{H}}_j$ from the Hamiltonian, and then evolving under it.

Ref.~\cite{Campbell_2019} proves that this channel reproduces evolution under $e^{-i \hat{H} \Delta t}$ up to error $\mathcal{O}(\lambda^2 \Delta t^2)$. In addition, the QDrift channel can be realized by sampling $j \sim p_j$ and applying the corresponding evolution $e^{-i \hat{\mathcal{H}}_j \lambda \Delta t}$. The upshot is that QDrift requires only execution of a single operation $e^{-i \hat{\mathcal{H}}_j \lambda \Delta t}$ per time step, in contrast to first-order Trotterization (Eq.~\eqref{eq:first_order_Trotter}) which requires $L$ operations per time step.

Iterating this channel over $r$ time steps, the final state incurs error $\mathcal{O}(\lambda^2 t^2 /r)$. Therefore, achieving error $\epsilon$ requires $r=\mathcal{O}(\lambda^2 t^2/\epsilon)$ steps. Because each QDrift step requires only a single operation, this is also the computational cost (i.e., the number of gates in a quantum algorithm, or number of matrix multiplications in classical simulation). Importantly, this cost depends only on the 1-norm $\lambda$ rather than the number of Hamiltonian terms $L$, in contrast to standard Trotterization whose cost scales explicitly with $L$. As a result, QDrift offers an advantage for Hamiltonians with $\lambda \ll L$, such as systems with many small interactions. These include Hamiltonians with long-range interactions like the long-range Ising model and quantum chemical Hamiltonians~\cite{Campbell_2019}.

Lastly, we note that while Eq.~\eqref{eq:QDrift_real_time} defines QDrift for real-time evolution, Ref.~\cite{pocrnic2023composite} proves that QDrift also extends to imaginary time evolution by an appropriate modification of the channel:
\begin{equation}
    \Lambda_{\text{QDrift}}(\rho) = \sum_{j=1}^L p_j e^{-\lambda \Delta \beta \hat{\mathcal{H}}_j } \rho e^{- \lambda \Delta \beta \hat{\mathcal{H}}_j } . 
\end{equation}
An analogous error scaling is attained, namely a computational cost $\mathcal{O}(\lambda^2 \beta^2/\epsilon )$ independent of $L$. Hence, QDrift may also be applied to thermal state estimation, as we did in developing randomly compiled path integral Monte Carlo in Sec.~\ref{sec:Application_PIMC}.

\subsection{Randomly Corrected Trotterization}\label{app:RC_trotter_corrected}

Ref.~\cite{Cho_2024} introduced a randomized Trotter formula that doubles the order of Trotterization: it reduces error from $\mathcal{O}((L\Delta t)^{n+1})$ to $\mathcal{O}((L\Delta t)^{2n+2})$, while retaining the cost of order-$n$ Trotterization. We refer to this as \textit{randomly corrected Trotterization}. 

This method is achieved by noting that the error in an order-$n$ Trotter formula can be expanded as a power series in $\Delta t$, starting at order $ \Delta t ^{n+1}$. In this expansion, the coefficient of the $m$th term is a linear combination of order-$(m-1)$ nested commutators of the constituent terms of the Hamiltonian~\cite{Childs_2021}. For instance, if $\hat{H} = \hat{T} + \hat{V}$, then the dominant error in an order-1 formula is
\begin{equation}\label{eq:1st_order_Trotter_error}
    e^{-i \hat{H} \Delta t} - e^{-i\hat{T}\Delta t} e^{-i\hat{V}\Delta t} = \frac{1}{2}[\hat{T}, \hat{V}] \Delta t^2 + \mathcal{O}(\Delta t^3) . 
\end{equation}
Ref.~\cite{Cho_2024} proposes to cancel out these higher-order errors by randomly applying correction operations after each Trotter step. On average, this cancels each error term up through order-$(2n+1)$, thus suppressing the overall Trotter error to $\mathcal{O}(\Delta t^{2n+2} )$.

Specifically, for order-$n$ Trotterization and a time step size $\Delta t$, they propose the channel
\begin{equation}\label{eq:RC_Trotter}
    \Lambda_{\text{RandCorrTrotter}}(\rho) = \sum_{j} p_j W_j (\Delta t) \rho {W_j (\Delta t)}^\dag , 
\end{equation}
where the unitaries $W_j (\Delta t)$ are each a symmetric product of order-$n$ Trotter formulas and an order-$n$ \textit{correction operation} $C_j(\Delta t )$: 
\begin{equation}\label{eq:random_corrections_ops}
    W_j (\Delta t) = S_{n}(\tfrac{\Delta t}{2}) C_j(\Delta t ) S_{n} (\tfrac{\Delta t}{2}) . 
\end{equation}
The correction operations $C_j(\Delta t )$, as well as the corresponding probabilities $p_j$, are systematically constructed by analyzing the errors incurred by an order-$n$ Trotter formula. A detailed procedure for constructing these is presented in Ref.~\cite{Cho_2024}. There, it is shown that the correction operations correspond to evolving under higher-order nested commutators of the constituent terms of the Hamiltonian. Similarly, the probabilities $p_j$ are taken to cancel out the appearance of these commutators in the error incurred by the Trotter formula. To make this procedure concrete, we work through an example of constructing $C_j(\Delta t )$ and $p_j$ for second-order Trotterization below in Appendix~\ref{app:RandomCorrectionExample}. The structure of $W_j$ is illustrated in Fig.~\ref{fig:trotter_rand}.

\input{figures/trotter_rand_fig}

By canceling out errors, the channel $\Lambda_{\text{RandCorrTrotter}}$ reproduces time evolution with error $\mathcal{O}((L\Delta t)^{2n+2})$, thus achieving a quadratic improvement over standard order-$n$ Trotterization. As a result, the cost to achieve a total simulation error $\epsilon$ is $\mathcal{O}(5^{n/2}L^2t (Lt/\epsilon)^{1/(2n+1)})$. This reduces the cost relative to standard Trotterization (i.e., Eq.~\eqref{eq:Trotter_time_complexity}) by a factor of $\mathcal{O}((Lt/\epsilon)^{(n+1)/(n(2n+1))})$, or approximately $ \mathcal{O}((Lt/\epsilon)^{1/2n})$ for large $n$. The caveat to be noted is that realizing the correction operations requires care. Although they can be computed analytically, their support generally grows linearly with the Trotter order, which incurs an associated cost in both quantum and classical simulation.


\subsubsection{Randomly Corrected 2nd-Order Trotterization}\label{app:RandomCorrectionExample}
To make randomly corrected Trotterization more clear, let us demonstrate how it promotes a 2nd-order Trotter formula to a random formula whose error scales equivalently to a 4th-order formula. We note that this does not fully saturate the error bounds established in Ref.~\cite{Cho_2024}, according to which a $5$th order formula (i.e., $\mathcal{O}(\Delta t^6)$ error) is in principle possible. Nevertheless, our construction is intentionally simplified to cancel errors only through fourth order, because 4th-order Trotterization is often the highest order used in practice and will suffice for our purposes.

For simplicity, 
let us consider a Hamiltonian composed of two non-commuting terms:
\begin{equation}
    \hat{H} = \hat{A} + \hat{B}, \qquad [\hat{A}, \hat{B}] \neq 0. 
\end{equation}
This includes Hamiltonians that decompose into a kinetic and potential energy as $\hat{H} = \hat{T} + \hat{V}$, as well as nearest-neighbor Hamiltonians that decompose into operations on even and odd sites $\hat{H} = \hat{H}_{\text{even}} + \hat{H}_{\text{odd}}$. 
The corresponding 2nd-order Trotter formula ($n=2$) is $S_2(\Delta t) = e^{-i\hat{A}\Delta t/2 } e^{-i \hat{B} \Delta t } e^{-i \hat{A} \Delta t/2 }$, which suffers error $\mathcal{O}(\Delta t^3)$. Following Ref.~\cite{Cho_2024}, we will use randomized compiling to promote this to a channel that achieves error $\mathcal{O}(\Delta t^5)$, equivalent to a 4th-order Trotter formula.

As in Eq.~\eqref{eq:random_corrections_ops}, we consider randomizing over the operators
\begin{equation}\label{eq:RC_Trotter_order2}
    W_j(\Delta t) = S_2( \tfrac{\Delta t}{2} ) C_j(\Delta t) S_2( \tfrac{\Delta t}{2} ) ,
\end{equation}
(illustrated in Fig.~\ref{fig:rand_2nd_order}(a)) with probability $p_j$. In accordance with the mixing lemma (Eq.~\eqref{eq:RC_err_suppression} and its preceding discussion), our desired suppression of error requires that the average of these operations suffers error $\mathcal{O}(\Delta t^5)$:
\begin{equation}
    \sum_j p_j W_j(\Delta t) = e^{-i\hat{H}\Delta t} + \mathcal{O}(\Delta t^5) . 
\end{equation}
Because $S_2(\tfrac{\Delta t}{2})^{-1} = S_2(\tfrac{-\Delta t}{2})$, this can be equivalently expressed as 
\begin{equation}\label{eq:SatisfyCorrOps}
    \sum_j p_j C_j(\Delta t ) = S_2(\tfrac{-\Delta t}{2}) e^{-i(\hat{A}+\hat{B})\Delta t} S_2(\tfrac{-\Delta t}{2}) + \mathcal{O}(\Delta t^5) .  
\end{equation}
Let us denote this as evolution under an operator $\hat{K}$ for a judiciously chosen time step $(\Delta t/2)^3$:
\begin{equation}\label{eq:SatisfyCorrOps_2}
    \sum_j p_j C_j(\Delta t ) = e^{-i \hat{K} (\Delta t/2)^3} + \mathcal{O}(\Delta t^5) ,  
\end{equation}
where $\hat{K}$ is defined as 
\begin{equation}
    \hat{K} = i  (\tfrac{2}{\Delta t})^3 \ln(S_2(\tfrac{-\Delta t}{2}) e^{-i(\hat{A} + \hat{B})\Delta t} S_2(\tfrac{-\Delta t}{2}) ) + \mathcal{O}(\Delta t^2). 
\end{equation}
(This defines $\hat{K}$ as the lowest-order terms in this expression, neglecting all higher-order $\mathcal{O}(\Delta t^2)$ contributions.)

To satisfy this relationship, we begin by using the symmetric Baker-Campbell-Hausdorff (BCH) formula to write~\cite{Casas_2009, Berry_Email}
\begin{equation}\label{eq:2nd_order_Trotter}
    \begin{aligned}
        &S_2(\tfrac{\Delta t}{2}) = e^{-i\hat{A}\Delta t/4 } e^{-i \hat{B} \Delta t/2 } e^{-i \hat{A} \Delta t/4 } = e^{-i \hat{Z} \Delta t/2} , 
    \end{aligned}
\end{equation}
for an operator 
\begin{equation}
    \begin{aligned}
        &\hat{Z} = \hat{A} + \hat{B} -\frac{(i \Delta t/2)^2}{24} [ \hat{A} + 2\hat{B}, [\hat{A}, \hat{B}]] + \mathcal{O}(\Delta t^4) . 
    \end{aligned}
\end{equation}
Applying the symmetric BCH formula once more, we have 
\begin{equation}
\begin{aligned}
    &\ln( S_2(\tfrac{-\Delta t}{2}) e^{-i(\hat{A}+\hat{B})\Delta t} S_2(\tfrac{-\Delta t}{2}) )  \\ 
    =& \ln( e^{i\hat{Z}\Delta t/2} e^{-i(\hat{A}+\hat{B})\Delta t} e^{i\hat{Z}\Delta t/2} ) \\
    =&-i(\hat{A}+\hat{B})\Delta t + i\hat{Z}\Delta t \\
    & - \frac{(i\Delta t)^3}{24} \big[ \hat{Z} - 2(\hat{A}+\hat{B}) , \  [\hat{Z}, -(\hat{A}+\hat{B})] \big] + \mathcal{O}(\Delta t^5) \\
    =&-\frac{(i\Delta t/2)^3}{12} \big[ \hat{A} + 2\hat{B}, \ [\hat{A}, \hat{B}] \big] + \mathcal{O}(\Delta t^5) \\
    =&  \frac{i(\Delta t/ 2)^3}{12} \big[ \hat{A} + 2\hat{B}, \ [\hat{A}, \hat{B}] \big] + \mathcal{O}(\Delta t^5) . 
\end{aligned}
\end{equation}
Thus we have 
\begin{equation}\label{eq:K_op}
    \hat{K} = \frac{-1}{12} \big[ \hat{A} + 2\hat{B}, \ [\hat{A}, \hat{B}] \big] . 
\end{equation}
Inserting this into Eq.~\eqref{eq:SatisfyCorrOps_2}, we would like the average of the correction operators to reproduce evolution under $\hat{K}$:
\begin{equation}
\begin{aligned}
    \sum_j p_j C_j(\Delta t) &= e^{-i \hat{K} (\Delta t/2)^3} + \mathcal{O}(\Delta t^5) \\  
    &= e^{\frac{i}{12} (\frac{\Delta t}{2})^3 [ \hat{A} + 2\hat{B}, [\hat{A}, \hat{B}] ]} + \mathcal{O}(\Delta t^5) .  
\end{aligned}
\end{equation}

Ref.~\cite{Cho_2024} realizes this evolution by importance sampled Trotterization, analogous to QDrift. They decompose $\hat{K}$ into a sum of normalized non-commuting terms $\hat{K} = \sum_j k_j \hat{K}_j$, for $k_j>0$ and $\| \hat{K}_j\| =1$. As in QDrift, they also define the 1-norm $\alpha = \sum_j k_j $. Then, each correction operation $C_j(\Delta t)$ is chosen to be evolution under a single term:
\begin{equation}\label{eq:Corr_op_Trotter}
    C_j(\Delta t) = e^{-i\hat{K}_j \alpha (\Delta t/2)^3} , 
\end{equation}
and is importance sampled with probability $p_j = k_j / \alpha$. To verify this construction, we can expand $C_j(\Delta t)$ to lowest order in $\Delta t$ to find that 
\begin{equation}
\begin{aligned}
    \sum_j p_j C_j(\Delta t) &= 
    \sum_j \frac{k_j}{\alpha} \Big( 1-i\hat{K}_j\alpha (\tfrac{\Delta t}{2})^3 + \mathcal{O}(\Delta t^6) \Big) \\
    &= 1 -i\hat{K} (\tfrac{\Delta t}{2})^3 + \mathcal{O}(\Delta t^6) \\
    &=
    e^{-i \hat{K} (\Delta t/ 2)^3 } + \mathcal{O}(\Delta t^6) \\
    &= S_2(\tfrac{-\Delta t}{2}) e^{-i\hat{H}\Delta t} S_2(\tfrac{-\Delta t}{2}) + \mathcal{O}(\Delta t^5) .  
\end{aligned}
\end{equation}
(Technically, this ignores the error's dependence on $\alpha$, in favor of its scaling with $\Delta t$, which is the relevant quantity when focusing on doubling the order of Trotterization.) This satisfies the desired relationship of Eq.~\eqref{eq:SatisfyCorrOps}.

Let us illustrate this construction for a nearest-neighbor Hamiltonian that decomposes into even and odd terms as 
\begin{equation}\label{eq:Local_Ham}
\begin{aligned}
    & \hat{H} = \sum_j \hat{H}_{j,j+1} = \hat{H}_{\text{even}} + \hat{H}_{\text{odd}} \\
    & \hat{H}_{\text{even}} = \sum_{j \text{ even}} \hat{H}_{j,j+1} = \hat{A} , \\
    & \hat{H}_{\text{odd}} = \sum_{j \text{ odd}} \hat{H}_{j,j+1} =\hat{B} . 
\end{aligned}
\end{equation}
We can then write
\begin{equation}
    [\hat{A}, \hat{B}] = \sum_{j \text{ even}} \sum_{k \text{ odd}} [\hat{H}_{j,j+1}, \hat{H}_{k,k+1}].
\end{equation}
The commutator in the summation will only be nonzero for $k=j\pm 1$, so this reduces to
\begin{equation}\label{eq:asym_com_1}
\begin{aligned}
    &\sum_{j \text{ even}} [\hat{H}_{j,j+1}, \hat{H}_{j+1,j+2}] + [\hat{H}_{j,j+1}, \hat{H}_{j-1,j}] \\
    =&  \sum_{j \text{ even}} [\hat{H}_{j,j+1}, \hat{H}_{j+1,j+2}] - [\hat{H}_{j-1,j}, \hat{H}_{j,j+1}] \\
    =& \sum_{j} (-1)^j [\hat{H}_{j,j+1}, \hat{H}_{j+1,j+2}] \\
    =: &  \sum_{j} \hat{H}_{j:j+2} , 
\end{aligned}
\end{equation}
where $ \hat{H}_{j:j+2} = (-1)^j [\hat{H}_{j,j+1}, \hat{H}_{j+1,j+2}]$ is defined as the commutator in the second-to-last line, and acts on three sites ($j, \ j+1$, and $j+2$). Let us also write
\begin{equation}
    \hat{A} + 2\hat{B} = \hat{H}_\text{even} + 2\hat{H}_\text{odd} = \sum_j c_j \hat{H}_{j,j+1} ,
\end{equation}
where
\begin{equation}
    c_j =  
    \begin{cases}
        1 & j \text{ even} \\ 
        2 & j \text{ odd} . 
    \end{cases}
\end{equation}

Putting this all together, the commutator comprising $\hat{K}$ is
\begin{equation}
\begin{aligned}
    & [\hat{A} + 2\hat{B}, [\hat{A}, \hat{B}]] = \\
    &\sum_{j,k} c_j (-1)^k [ \hat{H}_{j,j+1} , [\hat{H}_{k,k+1}, \hat{H}_{k+1,k+2}] ].
\end{aligned}
\end{equation}
The terms in this sum are only nonzero for $k=j-2$, $k=j-1$, $k=j$, and $k=j+1$. Together, these terms equate to
\begin{equation}\label{eq:K_comm_decomp}
\begin{aligned}
    & \quad \sum_{j}  (-1)^j c_j \Big( [ \hat{H}_{j,j+1} , [\hat{H}_{j-2,j-1}, \hat{H}_{j-1,j}] ] \\
    & \qquad \qquad \qquad \ - [ \hat{H}_{j,j+1} , [\hat{H}_{j-1,j}, \hat{H}_{j,j+1}] ] \\
    & \qquad \qquad \qquad \ + [ \hat{H}_{j,j+1} , [\hat{H}_{j,j+1}, \hat{H}_{j+1,j+2}] ] \\
    & \qquad \qquad \qquad \ - [ \hat{H}_{j,j+1} , [\hat{H}_{j+1,j+2}, \hat{H}_{j+2,j+3}] ] \Big)\\
    &=\sum_j   (-1)^j c_j  \Big( [ \hat{H}_{j,j+1} + \hat{H}_{j+2,j+3} , [\hat{H}_{j,j+1}, \hat{H}_{j+1,j+2}]  ]  \\
    & \qquad \qquad - [ \hat{H}_{j,j+1} + \hat{H}_{j+2,j+3} , [\hat{H}_{j+1,j+2}, \hat{H}_{j+2,j+3}] ] \Big)  \\
    & =: \sum_j (-1)^j c_j \hat{H}_{j:j+3}, 
\end{aligned}
\end{equation}
where $\hat{H}_{j:j+3}$ is the 4-site term defined by the second-to-last line and acts on sites $(j, j+1, j+2, j+3)$. Therefore, by Eqs.~\eqref{eq:K_op} and~\eqref{eq:K_comm_decomp}, $\hat{K}$ decomposes into a sum of local 4-site terms as
\begin{equation}
\begin{aligned}
    \hat{K} = \frac{-1}{12} \sum_j  (-1)^j c_j \hat{H}_{j:j+3} . 
\end{aligned}
\end{equation}
We can then apply the importance-sampled Trotterization of Eq.~\eqref{eq:Corr_op_Trotter} to realize randomly corrected Trotterization, and achieve an error equivalent to a 4th-order Trotter formula.

\input{figures/rand_2nd_order_fig}

\subsubsection{A More Efficient Implementation of 2nd-Order Randomly Corrected Trotterization}\label{app:Efficient_randomly_corrected_Trotter}

While the above example promotes a 2nd-order Trotter formula to a 4th-order formula, it can be made more efficient as follows. First, note that as per the randomly corrected Trotterization construction of Ref.~\cite{Cho_2024}, the operators we randomize over at each time step consist of two Trotter formulas and a correction operator. For instance, above we took $ W_j(\Delta t) = S_2( \tfrac{\Delta t}{2} ) C_j(\Delta t) S_2( \tfrac{\Delta t}{2} )$, which requires two calls to $S_2( \tfrac{\Delta t}{2} )$. This symmetric product was chosen because it satisfies $W_j(\Delta t)^{-1} =  W_j(-\Delta t) $, and thus errors of even order in $\Delta t$ vanish. Above, this naturally canceled errors $\mathcal{O}(\Delta t^4)$ and allowed us to more easily promote to an equivalently-4th-order Trotter formula.

However, it would be less costly, and perhaps more natural, if each operator consisted of only a single Trotter step and a correction operation, i.e., a construction like $W_j(\Delta t) = S_2(\Delta t) C_j(\Delta t) $. However, in order to do so, we will need to contend with the $\mathcal{O}(\Delta t^4)$ errors, which do not necessarily cancel in such a product. Fortunately, constructing a symmetric product is not the only way to ensure this cancellation. In fact, we can cancel out these errors by randomizing over the ordering of terms in the Trotter formula.



We realize this by randomizing over the order of the $S_2$ and the correction operation in $W_j$. Specializing to a Hamiltonian $\hat{H} = \hat{A} + \hat{B}$, we can express the 2nd-order Trotter formula as
\begin{equation}\label{eq:Z_H_K_relation}
\begin{aligned}
    &S_2(\Delta t) = e^{-i\hat{A} \Delta t/2 } e^{- i\hat{B} \Delta t } e^{- i\hat{A} \Delta t/2 } = e^{- i\hat{Z} \Delta t} ,\\ 
    &\hat{Z} = \hat{A} + \hat{B} + \frac{(\Delta t)^2}{24}[\hat{A} + 2\hat{B}, [\hat{A}, \hat{B}]] + \mathcal{O}(\Delta t^4) \\
    &\ \ \  = \hat{H} - \frac{(\Delta t)^2}{2} \hat{K} + \mathcal{O}(\Delta t^4) . 
\end{aligned}
\end{equation}
As above, we decompose $\hat{K} = \sum_j k_j \hat{K}_j$ with $\alpha = \sum_j k_j$. With this relationship, we propose a randomly corrected Trotter formula with the following operators and probabilities:
\begin{equation}\label{eq:efficient_2nd_order_corrected_Trotter}
\begin{aligned}
    &W_{j,+ }(\Delta t) = S_2(\Delta t) e^{-i\hat{K}_j \frac{\alpha}{2} \Delta t^3}, \quad  p_{j,+} = \frac{k_j}{2\alpha} \\
    &W_{j,- }(\Delta t) = e^{-i\hat{K}_j \frac{\alpha}{2} \Delta t^3}S_2(\Delta t), \quad p_{j,-} = \frac{k_j}{2\alpha} .  \\
\end{aligned}
\end{equation}
Note that this is analogous to the earlier construction, but we now also randomize over the order of the correction operation and Trotter step, as illustrated in Fig.~\ref{fig:rand_2nd_order}(b). As we will see, this effectively cancels out $\mathcal{O}(\Delta t^4)$ errors.

To verify this construction, we again look to the mixing lemma. The individual error suffered by $W_{j,s}(\Delta t)$ is $\| W_{j,s}(\Delta t) - e^{-i\hat{H}\Delta t} \| \leq \mathcal{O}(\Delta t^3)$ for either $s=\pm$. On the other hand, the average error of these is seen to be $\mathcal{O}(\Delta t ^5) $ as:
\begin{equation}
    \begin{aligned}
        & \sum_j \sum_{s=\pm} p_{j,s} W_{j,s}(\Delta t) = \\
        &\sum_j \frac{k_j}{2\alpha} \Big( S_2(\Delta t) e^{-i\hat{K}_j \frac{\alpha}{2} \Delta t^3} + e^{-i\hat{K}_j \frac{\alpha}{2} \Delta t^3} S_2(\Delta t) \Big) = \\
        & \frac{1}{2} \Big( S_2(\Delta t) \big( e^{-i\hat{K} \frac{1}{2} \Delta t^3} + \mathcal{O}(\Delta t^6) \big)  \\
        & \ \quad + \big( e^{-i\hat{K} \frac{1}{2} \Delta t^3} + \mathcal{O}(\Delta t^6) \big) S_2(\Delta t) \Big) = \\
        & \frac{1}{2} \Big( e^{-i\hat{Z} \Delta t} e^{-i\hat{K} \frac{1}{2} \Delta t^3} + e^{-i\hat{K} \frac{1}{2} \Delta t^3} e^{-i\hat{Z} \Delta t} \Big) + \mathcal{O}(\Delta t^6)= \\
        &\frac{1}{2} \Big( e^{-i(\hat{Z} + \frac{\Delta t^2}{2} \hat{K})\Delta t} - \frac{1}{4}[\hat{Z}, \hat{K}]\Delta t^4 + \mathcal{O}(\Delta t^5) \\
        &\ \quad + e^{-i(\hat{Z} + \frac{\Delta t^2}{2} \hat{K})\Delta t} + \frac{1}{4}[\hat{Z}, \hat{K}]\Delta t^4 + \mathcal{O}(\Delta t^5) \Big) \\
        &\ \ + \mathcal{O}(\Delta t^6) = \\
        & e^{-i(\hat{Z} + \frac{\Delta t^2}{2} \hat{K})\Delta t} +\mathcal{O}(\Delta t^5) = \\
        & e^{-i\hat{H} \Delta t} +\mathcal{O}(\Delta t^5) . 
    \end{aligned}
\end{equation}
Here we have used the identity $e^{\hat{C} \Delta t } e^{\hat{D} \Delta t } = e^{(\hat{C} + \hat{D})\Delta t} + \frac{1}{2} [\hat{C}, \hat{D} ] \Delta t^2 + \mathcal{O}(\Delta t^3)$, and the relation $\hat{Z} + \frac{\Delta t^2}{2} \hat{K} = \hat{H} + \mathcal{O}(\Delta t^4)$ from Eq.~\eqref{eq:Z_H_K_relation}. 
Overall, this implies that the construction of Eq.~\eqref{eq:efficient_2nd_order_corrected_Trotter} indeed suppresses errors to $\mathcal{O}(\Delta t^5)$, thus achieving the performance of 4th-order Trotterization.

The upshot of this construction is that each $W_{j,\pm}(\Delta t)$ requires only one call to $S_2(\Delta t)$, in contrast to the earlier construction of Eq.~\eqref{eq:RC_Trotter_order2} which requires two calls to $S_2(\Delta t/2)$. Indeed, this is the construction we employed in the experiments of Sec.~\ref{sec:Application_QTM} to demonstrate the randomly compiled quantum trajectories method on the hardcore Bose-Hubbard model. In that context, $\hat{H} = -J \sum_j (a_{j+1}^\dag a_j + a_j^\dag a_{j+1} ) $ is a sum over terms acting on even and odd sites with $\hat{H}_{j,j+1} = -J ( a_{j+1}^\dag a_j + a_j^\dag a_{j+1} )$.

\section{Error Bound on the State Used in Randomly Compiled Path Integral Monte Carlo}\label{app:Error_Bound_RC_PIMC}

In this appendix, we show that the state in Eq.~\eqref{eq:RC_thermalstate} provides a good approximation to the thermal state. This is important because we used this state to evaluate thermal observables in randomly compiled path integral Monte Carlo. Recall that in this setting, we consider a Hamiltonian $\hat{H} = \sum_{j=1}^L \hat{H}_j$ composed of $L$ terms, and re-express it as $\hat{H} = \sum_{j=1}^L h_j \hat{\mathcal{H}}_j$, where $h_j > 0$ and $\|\hat{\mathcal{H}}_j\| = 1$. We also define the associated probability distribution $p_j = h_j/\lambda$ for $\lambda = \sum_j h_j $.

With these definitions, our goal is to approximate the thermal state $e^{-\beta H}/Z$. We do this by using the QDrift channel in imaginary time, which for a time step $\Delta \beta$ is 
\begin{equation}
    \Lambda_{\text{QDrift}} (\rho) = \sum_j p_j e^{- \hat{\mathcal{H}}_j \lambda \Delta \beta} \rho e^{- \hat{\mathcal{H}}_j\lambda \Delta \beta} . 
\end{equation}
This emulates the imaginary time evolution channel $\rho \rightarrow e^{-\Delta \beta H} \rho e^{- \Delta \beta H}$ to error $\mathcal{O}(\lambda^2 \Delta\beta^2)$ in the induced 1-norm~\cite{Campbell_2019, pocrnic2023composite}. Therefore, multiple compounded instances of this channel approximate the thermal state. In particular, let $\Delta \beta = \beta/r$ for an even $r$, and define the state that is the normalized state of $r/2$ compounded QDrift channels:
\begin{equation}
\begin{aligned}
    \sigma &= \frac{\Lambda_{\text{QDrift}}^{\circ r/2}(I)}{\tr(\Lambda_{\text{QDrift}}^{\circ r/2}(I))} \\
    &= \frac{\sum_{\mathbf{j}} p_{\mathbf{j}}\Big[ \prod_{k=1}^{r/2} e^{-\lambda \frac{\beta}{r} \hat{\mathcal{H}}_{j_k} } \times \prod_{k'=r/2}^1 e^{-\lambda \frac{\beta}{r} \hat{\mathcal{H}}_{j_{k'}} } \Big]  }{\sum_{\mathbf{j}'} p_{\mathbf{j}'} \tr \Big[ \prod_{\kappa=1}^{r/2} e^{-\lambda \frac{\beta}{r} \hat{\mathcal{H}}_{j_\kappa} } \times \prod_{\kappa'=r/2}^1 e^{-\lambda \frac{\beta}{r} \hat{\mathcal{H}}_{j_{\kappa'}} } \Big] } , 
\end{aligned}
\end{equation}
where $\mathbf{j} = (j_1, j_2, ... j_{r/2})$ is a multi-index that we refer to as the QDrift sequence, and $p_{\mathbf{j}} = \prod_{k=1}^{r/2} p_{j_k}$ is the associated probability distribution. To wit, the numerator is the sum over all possible symmetric products of sampled imaginary time steps $e^{-\lambda \frac{\beta}{r} \hat{\mathcal{H}}_{j_k} }$. Because the QDrift channel approximates imaginary time evolution, this state approximates the thermal state as
\begin{equation}
    \begin{aligned}
        \Big\| \sigma - \frac{e^{-\beta \hat{H}}}{Z} \Big\|_1 \leq \mathcal{O}\Big(\frac{r}{2} \lambda^2\Delta \beta^2 \Big) = \mathcal{O}(\lambda^2 \beta^2/r) . 
    \end{aligned}
\end{equation}
(See Ref.~\cite{pocrnic2023composite} for a detailed proof of this statement.)

While the state $\sigma$ is indeed a direct application of QDrift, it does not fit directly into the framework of RC-QMC. This is because an observable computed in $\sigma$ does not obviously decompose into an average of approximations as $\langle \hat{O}\rangle = \sum_{\mathbf{j}} p_{\mathbf{j}} \langle \hat{O}\rangle_{\mathbf{j}}$, which we need to employ randomly compiled QMC according to Eq.~\eqref{eq:RC_QMC}. This happens because $\sigma$ features a global normalization factor, which obstructs this desired decomposition; ultimately, this arises from the QDrift channel being non-unitary and requiring renormalization after application.

Fortunately, this issue is surmountable. Rather than integrate $\sigma$ into RC-QMC, we can replace it with an appropriately modified state $\sigma'$, whose corresponding observables decompose into an average over QDrift sequences $\mathbf{j}$. We posit the following modified state:
\begin{equation}\label{eq:sigma_prime}
    \sigma' = \sum_{\mathbf{j}} p_{\mathbf{j}} \frac{ \prod_{k=1}^{r/2} e^{-\lambda \frac{\beta}{r} \hat{\mathcal{H}}_{j_k} } \times \prod_{k'=r/2}^1 e^{-\lambda \frac{\beta}{r} \hat{\mathcal{H}}_{j_{k'}} } }{ \tr( \prod_{\kappa=1}^{r/2} e^{-\lambda \frac{\beta}{r} \hat{\mathcal{H}}_{j_\kappa} } \times \prod_{\kappa'=r/2}^1 e^{-\lambda \frac{\beta}{r} \hat{\mathcal{H}}_{j_{\kappa'}} }) }. 
\end{equation}
The upshot of $\sigma'$ is that it is manifestly an average of individually normalized states. Thus, an observable computed in $\sigma'$ decomposes into an average over QDrift sequences $\mathbf{j}$, which fits into the RC-QMC framework. Indeed, we used this state in Sec.~\ref{sec:Application_PIMC} to develop randomly compiled path integral Monte Carlo. In addition, $\sigma'$ approximates the thermal state to the same level of accuracy as $\sigma$, namely an error $\mathcal{O}(\lambda^2\beta^2/r)$. We prove this as follows: 
\begin{theorem}[Error bound on $\sigma'$]\label{thm:ErrorBound_sigma_prime}
    Let $\sigma'$ be the state of Eq.~\eqref{eq:sigma_prime}. Then this state differs from the thermal state as:
    \begin{equation}
        \Big\| \sigma' - \frac{e^{-\beta \hat{H}}}{Z} \Big\|_1 \leq \mathcal{O}(\lambda^2 \beta^2 / r) . 
    \end{equation}
\end{theorem}
\begin{proof}
    To begin proving this, let us simplify notation by introducing the non-unitary map:
    \begin{equation}
        \mathcal{A}_{j} (\rho) = e^{-\hat{\mathcal{H}}_j \frac{\lambda \beta}{r}} \rho e^{-\hat{\mathcal{H}}_j \frac{\lambda \beta}{r}} =: A_j \rho A_j, 
    \end{equation}
    where $A_j = e^{-\hat{\mathcal{H}}_j \frac{\lambda \beta}{r}} = A_j^\dag$. We can then write $\sigma'$ as 
    \begin{equation}\label{eq:NormalizedQDriftState}
        \sigma '  = \sum_{j_{1:r/2}} p_{j_{1:r/2}} \frac{\mathcal{A}_{j_{1:r/2}} (I)}{\tr( \mathcal{A}_{j_{1:r/2}} (I) ) } ,  
    \end{equation}
    where we use the abridged notation $j_{1:r/2} = (j_1, j_2, ..., j_{r/2})$, $p_{j_{1:{r/2}}} = p_{j_1} p_{j_2} ... p_{j_{r/2}}$, and $\mathcal{A}_{j_{1:r/2}} (I) = \mathcal{A}_{j_{r/2}} ( ... \mathcal{A}_{j_1} (I))$. This notation will ease the proof of this theorem.

    Next, note that $\sigma'$ is not a composition of quantum channels. This is due to the presence of a normalization factor for each individual term $\mathcal{A}_{j_{1:r/2}} (I)$ in the sum over $j_{1:r/2}$. This distinction from a composed quantum channel means that we cannot use the usual tricks from quantum computing to bound the 1-norm distance between two quantum states.

    Instead, we propose the following sequence of states, which resembles the action of a composed quantum channel:
    \begin{equation}
        \sigma'_n = \sum_{ j_{1:r/2-n}} p_{j_{1:r/2-n}} \frac{e^{\frac{-\beta n \hat{H}}{r}} \mathcal{A}_{j_{1:r/2-n}} (I) e^{\frac{-\beta n\hat{H}}{r}} }{ \tr( e^{\frac{-\beta n\hat{H}}{r}} \mathcal{A}_{j_{1:r/2-n}} (I) e^{\frac{-\beta n\hat{H}}{r}} ) } , 
    \end{equation}
    for $n= 0, 1 ..., r/2$. Observe that $\sigma'_0 = \sigma'$ is our proposed state, and $\sigma'_{r/2} = e^{-\beta \hat{H}}/Z $ is the thermal state. Therefore, with increasing $n$, the states $\sigma'_n$ converge to the thermal state. Our proof strategy, illustrated in Fig.~\ref{fig:rand_pimc_error}, is to show that adjacent states in this sequence are close to each other: $\| \sigma'_{n+1} - \sigma'_n \|_1 = \mathcal{O}(\lambda^2 \beta^2 /r^2)$, similar to the error suffered by application of the QDrift channel. By then constructing a telescoping series of these states, from $n=0$ up to $n=r/2$, we can show that $\sigma'$ and $e^{-\beta {\hat{H}}}/Z $ differ by at most $r/2 \cdot \mathcal{O}(\lambda^2 \beta^2 /r^2) = \mathcal{O}(\lambda^2 \beta^2 /r)$.

    \input{figures/rand_pimc_error_fig}

    Let us first show the stated claim: $\| \sigma'_{n+1} - \sigma'_n \|_1 = \mathcal{O}(\lambda^2 \beta^2 /r^2)$ for $n\in \{0, 1, ..., r/2-1\}$. For ease of notation, let $\chi := \mathcal{A}_{j_{1:r/2-(n+1)}} (I)$. Because $\chi$ is a symmetric product of Hermitian, positive semi-definite matrices (i.e., each $e^{-\hat{\mathcal{H}}_j \frac{\lambda \beta}{r}}$), $\chi$ is also Hermitian and positive semi-definite: $\chi = \chi^\dag$ and $\chi \succeq 0$. These two properties will be useful in proving our desired inequality. We can then write $\sigma_n'$ as
    \begin{equation}\label{eq:sigma_prime_n}
    \begin{aligned}
        &\sigma'_n = \sum_{j_{1:r/2-n}} p_{j_{1:r/2-(n+1)}} p_{j_{r/2-n}} \\
        &\qquad \qquad \qquad \times \frac{e^{\frac{-\beta n\hat{H}}{r}} A_{j_{r/2-n}} \chi A_{j_{r/2-n}} e^{\frac{-\beta n\hat{H}}{r}} }{ \tr( e^{\frac{-\beta n\hat{H}}{r}} A_{j_{r/2-n}} \chi A_{j_{r/2-n}} e^{\frac{-\beta n\hat{H}}{r}} ) } .  
    \end{aligned}
    \end{equation}

    We will now massage this expression into something that resembles $\sigma'_{n+1}$. For brevity, let us
    replace the subscript $j_{r/2-n}$ with $j$ in the math below. We can then write the term appearing in the numerator as:
    \begin{equation}\label{eq:sigma_numerator}
    \begin{aligned}
        A_j \chi A_j &= e^{\frac{-\lambda \beta}{r} \hat{\mathcal{H}}_j } \chi e^{\frac{-\lambda \beta}{r} \hat{\mathcal{H}}_j }  \\ 
        &= \chi - \frac{\lambda \beta}{r} (\hat{\mathcal{H}}_j \chi + \chi \hat{\mathcal{H}}_j) + \mathcal{O}(\| \chi\|_1 \lambda^2 \beta^2 / r^2) , 
    \end{aligned}
    \end{equation}
    where the higher-order terms in the $\mathcal{O}(\cdot)$ have 1-norm $\mathcal{O}(\| \chi\|_1 \lambda^2 \beta^2/r^2)$. This follows from application of H\"older's inequality and $\chi \succeq 0$ to the higher-order terms, e.g.,
    \begin{equation}
    \begin{aligned}
        \tr(\hat{\mathcal{H}}_j \chi \hat{\mathcal{H}}_j) &\leq \tr(|\hat{\mathcal{H}}_j \chi \hat{\mathcal{H}}_j|) = \| \hat{\mathcal{H}}_j \chi \hat{\mathcal{H}}_j \|_1 \\
        &\leq \|\hat{\mathcal{H}}_j\|^2 \|\chi \|_1 = \| \chi \|_1 = \tr(\chi) . 
    \end{aligned}
    \end{equation}
    We can similarly use this approach to express the term in the denominator as
    \begin{equation}\label{eq:sigma_denominator}
    \begin{aligned}
        & \tr( e^{\frac{-\beta n\hat{H}}{r}} A_{j} \chi A_{j} e^{\frac{-\beta n\hat{H}}{r}} ) = \\
        & \tr( e^{\frac{-\beta n\hat{H}}{r}} \chi e^{\frac{-\beta n\hat{H}}{r}} ) \\
        & \qquad- \frac{\lambda \beta}{r} \tr( e^{\frac{-\beta n\hat{H}}{r}} (\hat{\mathcal{H}}_j \chi + \chi \hat{\mathcal{H}}_j) e^{\frac{-\beta n\hat{H}}{r}} ) \\
        & \qquad + \mathcal{O}\Big( \lambda^2\beta^2/r^2 \cdot \tr( e^{\frac{-\beta n\hat{H}}{r}} \chi e^{\frac{-\beta n\hat{H}}{r}} ) \Big). 
    \end{aligned}
    \end{equation}

    Next, let us return to our expression for $\sigma'_n$ in Eq.~\eqref{eq:sigma_prime_n}, which we wish to simplify. By inserting Eq.~\eqref{eq:sigma_numerator} into the numerator, Eq.~\eqref{eq:sigma_denominator} into the denominator, and then using the Taylor series $1/(1-x) = 1+x+\mathcal{O}(x^2)$ to simplify the denominator, we can write
    \begin{widetext}
    \begin{equation}\label{eq:sigma_prime_n_simplify_1}
    \begin{aligned}
        & \frac{e^{\frac{-\beta n\hat{H}}{r}} A_{j} \chi A_{j} e^{\frac{-\beta n\hat{H}}{r}} }{ \tr( e^{\frac{-\beta n\hat{H}}{r}} A_{j} \chi A_{j} e^{\frac{-\beta n\hat{H}}{r}} ) } 
        \\
        & \qquad \qquad= \frac{e^{\frac{-\beta n\hat{H}}{r}} (\chi - \frac{\lambda \beta}{r} (\hat{\mathcal{H}}_j \chi + \chi \hat{\mathcal{H}}_j))  e^{\frac{-\beta n\hat{H}}{r}} }{ \tr( e^{\frac{-\beta n\hat{H}}{r}} \chi e^{\frac{-\beta n\hat{H}}{r}} ) } 
        \Bigg( 1 + \frac{\frac{\lambda \beta}{r} \tr(e^{\frac{-\beta n\hat{H}}{r}} (\hat{\mathcal{H}}_j \chi + \chi \hat{\mathcal{H}}_j) e^{\frac{-\beta n\hat{H}}{r}}) }{\tr( e^{\frac{-\beta n\hat{H}}{r}} \chi e^{\frac{-\beta n\hat{H}}{r}} )}  \Bigg) 
         + \mathcal{O} \Big(\frac{\lambda^2 \beta^2}{r^2} \Big) .
    \end{aligned}
    \end{equation}
    \end{widetext}
    Next, we insert this into Eq.~\eqref{eq:sigma_prime_n} and sum over $p_j$, making use of the facts that $\lambda \sum_j p_j \hat{\mathcal{H}}_j = \hat{H}$ and $\|\hat{H}\| \leq \lambda$. Because we retain terms only to linear order in $\lambda \beta / r$, the two separate appearances of $\hat{\mathcal{H}}_j \chi + \chi \hat{\mathcal{H}}_j$ in Eq.~\eqref{eq:sigma_prime_n_simplify_1} can be summed independently, yielding
    \begin{widetext}
    \begin{equation}
    \begin{aligned}
        & \sum_j p_j \frac{e^{\frac{-\beta n\hat{H}}{r}} A_{j} \chi A_{j} e^{\frac{-\beta n\hat{H}}{r}} }{ \tr( e^{\frac{-\beta n\hat{H}}{r}} A_{j} \chi A_{j} e^{\frac{-\beta n\hat{H}}{r}} ) } \\
        =& \frac{e^{\frac{-\beta n\hat{H}}{r}} (\chi - \frac{\lambda \beta}{r} \sum_j p_j (\hat{\mathcal{H}}_j \chi + \chi \hat{\mathcal{H}}_j))  e^{\frac{-\beta n\hat{H}}{r}} }{ \tr( e^{\frac{-\beta n\hat{H}}{r}} \chi e^{\frac{-\beta n\hat{H}}{r}} ) } 
        \Bigg( 1 + \frac{\frac{\lambda \beta}{r} \tr(e^{\frac{-\beta n\hat{H}}{r}} \sum_j p_j (\hat{\mathcal{H}}_j \chi + \chi \hat{\mathcal{H}}_j) e^{\frac{-\beta n\hat{H}}{r}}) }{\tr( e^{\frac{-\beta n\hat{H}}{r}} \chi e^{\frac{-\beta n\hat{H}}{r}} )}  \Bigg)
        + \mathcal{O} \Big(\frac{\lambda^2 \beta^2}{r^2} \Big) \\
        =&\frac{e^{\frac{-\beta n\hat{H}}{r}} (\chi - \frac{ \beta}{r} (\hat{H}\chi + \chi\hat{H}))  e^{\frac{-\beta n\hat{H}}{r}} }{ \tr( e^{\frac{-\beta n\hat{H}}{r}} \chi e^{\frac{-\beta n\hat{H}}{r}} ) } 
        \Bigg( 1 + \frac{\frac{ \beta}{r} \tr(e^{\frac{-\beta n\hat{H}}{r}} (\hat{H}\chi + \chi\hat{H}) e^{\frac{-\beta n\hat{H}}{r}}) }{\tr( e^{\frac{-\beta n\hat{H}}{r}} \chi e^{\frac{-\beta n\hat{H}}{r}} )}  \Bigg) 
        + \mathcal{O} \Big(\frac{\lambda^2 \beta^2}{r^2} \Big) \\
        =& \frac{e^{\frac{-\beta n\hat{H}}{r}} (\chi - \frac{ \beta}{r} (\hat{H}\chi + \chi\hat{H}))  e^{\frac{-\beta n\hat{H}}{r}} }{ \tr( e^{\frac{-\beta n\hat{H}}{r}} ( \chi  - \frac{\beta}{r} (H\chi + \chi\hat{H}) ) e^{\frac{-\beta n\hat{H}}{r}} ) } 
        + \mathcal{O} \Big(\frac{\lambda^2 \beta^2}{r^2} \Big) = \frac{e^{\frac{-\beta (n+1)\hat{H}}{r}} \chi  e^{\frac{-\beta (n+1)\hat{H}}{r}} }{ \tr( e^{\frac{-\beta (n+1)\hat{H}}{r}} \chi  e^{\frac{-\beta (n+1)\hat{H}}{r}}  ) } + \mathcal{O} \Big(\frac{\lambda^2 \beta^2}{r^2} \Big) .  
    \end{aligned}
    \end{equation}
    \end{widetext}

    Putting this back into Eq.~\eqref{eq:sigma_prime_n}, and re-inserting the notation $j = j_{r/2-n}$ and $\chi = \mathcal{A}_{j_{1:r/2-(n+1)}} (I)$, we obtain 
    \begin{widetext}
    \begin{equation}
        \begin{aligned}
            \sigma_n' &= \sum_{j_{1:r/2-n}} p_{j_{1:r/2-(n+1)}} p_{j_{r/2-n}} \frac{e^{\frac{-\beta n\hat{H}}{r}} A_{j_{r/2-n}} \chi A_{j_{r/2-n}} e^{\frac{-\beta n\hat{H}}{r}} }{ \tr( e^{\frac{-\beta n\hat{H}}{r}} A_{j_{r/2-n}} \chi A_{j_{r/2-n}} e^{\frac{-\beta n\hat{H}}{r}} ) } \\
            &= \sum_{j_{1:r/2-n}} p_{j_{1:r/2-(n+1)}} \Bigg[ \frac{e^{\frac{-\beta (n+1)\hat{H}}{r}} \chi  e^{\frac{-\beta (n+1)\hat{H}}{r}} }{ \tr( e^{\frac{-\beta (n+1)\hat{H}}{r}} \chi  e^{\frac{-\beta (n+1)\hat{H}}{r}}  ) } +  \mathcal{O} \Big(\frac{\lambda^2 \beta^2}{r^2} \Big) \Bigg] \\ 
            &= \sum_{j_{1:r/2-n}} p_{j_{1:r/2-(n+1)}} \frac{e^{\frac{-\beta (n+1)\hat{H}}{r}} \mathcal{A}_{j_{1:r/2-(n+1)}}(I)  e^{\frac{-\beta (n+1)\hat{H}}{r}} }{ \tr( e^{\frac{-\beta (n+1)\hat{H}}{r}} \mathcal{A}_{j_{1:r/2-(n+1)}}(I)  e^{\frac{-\beta (n+1)\hat{H}}{r}}  ) } +  \mathcal{O} \Big(\frac{\lambda^2 \beta^2}{r^2} \Big) \\
            &= \sigma'_{n+1} +  \mathcal{O} \Big(\frac{\lambda^2 \beta^2}{r^2} \Big) . 
        \end{aligned}
    \end{equation}
    \end{widetext}
    Because the $\mathcal{O}(\cdot)$ terms are bounded in 1-norm, this implies $\|\sigma'_n - \sigma'_{n+1} \|_1 \leq \mathcal{O}(\lambda^2\beta^2/r^2)$, as we set out to show. 

    Finally, we can invoke a telescoping series across the $\sigma'_n$ and employ the triangle inequality to upper bound $\| \sigma' - e^{-\beta \hat{H}}/Z\|_1$ as:
    \begin{equation}
        \begin{aligned}
             \| \sigma' - e^{-\beta \hat{H}}/Z\|_1 &= \| \sigma'_0 -\sigma'_{r/2} \|_1 \\
            & = \Bigg\|  \sum_{n=0}^{r/2-1} (\sigma'_n -\sigma'_{n+1}) \Bigg\|_1 \\
            & \leq \sum_{n=0}^{r/2-1} \big\|\sigma'_n - \sigma'_{n+1} \big\|_1 \\
            &\leq \frac{r}{2} \cdot \mathcal{O} \Big(\frac{\lambda^2 \beta^2}{r^2} \Big) = \mathcal{O} \Big(\frac{\lambda^2 \beta^2}{ r} \Big) . 
        \end{aligned}
    \end{equation}
    This proves the stated bound. 
    
    For completeness, we note that one could alternatively prove this result by using the approaches taken in Refs.~\cite{Campbell_2019, pocrnic2023composite} to prove similar bounds, namely expressing all relevant operations in terms of the Liouvillian representation of quantum channels. 
\end{proof}

With this bound proven, we can use $\sigma'$ to approximate thermal state observables and experience asymptotically the same error scaling as directly using QDrift. As emphasized above, the structure of $\sigma'$ allows observables to decompose into an average over QDrift sequences $\mathbf{j}$, and thus fits into the RC-QMC framework. This is indeed how we developed randomly compiled path integral Monte Carlo in Sec.~\ref{sec:Application_PIMC}. In addition, because the 1-norm upper bounds the error suffered in an arbitrary observable as per Eq.~\eqref{eq:obs_bound}, this error bound also proves that an observable estimated using $\sigma'$ deviates from its exact thermal value by $\mathcal{O}(\lambda^2\beta^2/r)$.

Furthermore, let us again mention that our definition of $\sigma'$ in Eq.~\eqref{eq:sigma_prime} is an average of symmetric products, i.e. $\prod_{k=1}^{r/2} e^{-\lambda \frac{\beta}{r} \hat{\mathcal{H}}_{j_k} } \times \prod_{k'=r/2}^1 e^{-\lambda \frac{\beta}{r} \hat{\mathcal{H}}_{j_{k'}} }$ for a sampled QDrift sequence $(j_1, ..., j_{r/2})$. While this symmetric structure was chosen to resemble a compounded QDrift channel, it also simplified the proof of the error bound in Theorem~\ref{thm:ErrorBound_sigma_prime}. The reason for this is that a symmetric product of Hermitian, positive definite matrices (which each $e^{-\hat{\mathcal{H}}_j \frac{\lambda \beta}{r}}$ is) is also Hermitian and positive semi-definite. This means that each numerator in the average in Eq.~\eqref{eq:sigma_prime} is also Hermitian and positive semi-definite. This allowed us to treat the numerator like an unnormalized density matrix and employ similar tools for proving error bounds.

As an alternative, one could consider the state that averages over products of imaginary time steps, rather than symmetric products:
\begin{equation}\label{eq:sigma_prime_prime}
    \sigma '' = \sum_{\mathbf{j}} p_{\mathbf{j}} \frac{ \prod_{k=1}^r e^{-\lambda \frac{\beta}{r} \hat{\mathcal{H}}_{j_k} } }{ \tr( \prod_{\kappa=1}^r e^{-\lambda \frac{\beta}{r} \hat{\mathcal{H}}_{j_\kappa} }) } .
\end{equation}
We expect that this state also closely approximates the thermal state to the same level of error as $\sigma'$, namely $\mathcal{O}(\lambda^2\beta^2/r)$. This can be intuited from concentration bounds for QDrift and products of random matrices~\cite{Chen_2021, huang2020matrix}. However, proving this is difficult because each numerator in Eq.~\eqref{eq:sigma_prime_prime} is neither necessarily Hermitian nor positive semi-definite, which renders many of the techniques used in the proof of Theorem~\ref{thm:ErrorBound_sigma_prime} not directly applicable. We imagine that establishing such a rigorous proof would require a detailed analysis and take us too far afield into concentration bounds and random matrix theory. 

Nonetheless, we also considered the state $\sigma''$ in the experiments on randomly compiled path integral Monte Carlo in Sec.~\ref{sec:RC_PIMC_Experiments}. There, we referred to use of $\sigma''$ as the ``asymmetric approach". Our results showed that $\sigma''$ also provides accurate estimates of thermal observables.

\section{Statistical Error of RC-QMC}\label{app:Stat_Error}

As we remarked in Sec.~\ref{sec:RC_QMC}, RC-QMC incurs an additional statistical error due to sampling over many approximations. Fortunately, this increased error is generally small because the approximations used in RC-QMC concentrate tightly around their mean.

To see this, let us first consider standard QMC, where we estimate an observable as
\begin{equation}
    \langle \hat{O} \rangle \approx \mathop{\mathbb{E}}_{X \sim \widetilde{\mathcal{P}} } \big[ F(\hat{O}, X) \big] . 
\end{equation}
The variance in this estimator is 
\begin{equation}
    \sigma^2_{\widetilde{\mathcal{P}}} := \mathop{\text{Var}}_{X \sim \widetilde{\mathcal{P}} } \big[ F(\hat{O}, X) \big].
\end{equation}
As such, by the central limit theorem, estimating $\langle \hat{O} \rangle$ to additive statistical error $\varepsilon$ requires averaging over $\mathcal{O}(\sigma^2_{\widetilde{\mathcal{P}}}/\varepsilon^2)$ samples drawn from $\widetilde{\mathcal{P}}$.

On the other hand, in randomly compiled QMC, we estimate an observable as
\begin{equation}
    \langle \hat{O} \rangle \approx \mathop{\mathbb{E}}_{j\sim p_j} \bigg[ \mathop{\mathbb{E}}_{X \sim \widetilde{\mathcal{P}}_j } \big[ F(\hat{O}, X) \big] \bigg] = \mathop{\mathbb{E}}_{j\sim p_j} \Big[ \langle \hat{O} \rangle_{ \widetilde{\mathcal{P}}_j } \Big] ,
\end{equation}
where $\{ \widetilde{\mathcal{P}}_j\}$ is a collection of approximate distributions, $p_j$ is a probability distribution over these approximations, and $\langle \hat{O} \rangle_{ \widetilde{\mathcal{P}}_j } := \mathop{\mathbb{E}}_{X \sim \widetilde{\mathcal{P}}_j } \big[ F(\hat{O}, X) \big]$ is an estimate of the observable using $\widetilde{\mathcal{P}}_j$.
By the law of total variance, the variance of this estimator is 
\begin{equation}
    \begin{aligned}
        &\mathop{\mathbb{E}}_{j \sim p_j} \bigg[ \mathop{\text{Var}}_{X \sim \widetilde{\mathcal{P}}_j } \big[ F(\hat{O}, X) \big] \bigg] + \mathop{\text{Var}}_{j \sim p_j} \bigg[ \mathop{\mathbb{E}}_{X \sim \widetilde{\mathcal{P}}_j } \big[ F(\hat{O}, X) \big] \bigg] \\
        =&\mathop{\mathbb{E}}_{j \sim p_j} \big[ \sigma_{\widetilde{\mathcal{P}}_j}^2 \big] + \mathop{\text{Var}}_{j \sim p_j} \big[ \langle \hat{O}\rangle_{\widetilde{\mathcal{P}}_j} \big] . 
    \end{aligned}
\end{equation}

The first term is the average variance of each estimator, and in practice scales like the variance $\sigma^2_{\widetilde{\mathcal{P}}}$ experienced in standard QMC. On the other hand, the second contribution is new: it is the variance across each estimated observable. While this term increases the total variance, its contribution is typically quite small. In particular, randomized compiling ensures that the observables $\langle \hat{O} \rangle_{\widetilde{\mathcal{P}}_j}$ are tightly concentrated around their mean, implying that their variance is small. For example, the conditions outlined in Sec.~\ref{sec:randomized_compiling} imply that these observables exhibit a variance of order $\mathcal{O}(\epsilon^2)$, matching the scaling of the error. When such a channel is compounded over many iterations, as in random Trotterization, the variance grows linearly with the number of iterations and therefore continues to scale with the total error, which is small in practice. In either case, this additional variance is a small quantity.

Thus, the increase in variance is not particularly large, nor is the number of samples needed to compensate for this increased variance. Moreover, this slight increase in sampling cost is more than offset by the improvements provided by RC-QMC. For further discussion and bounds on the statistical error associated with randomly compiled channels, see Refs.~\cite{Chen_2021, Kiss_2023}.

\end{document}

%% file: figures/result_fig.tex
%

\begin{figure*}[t]
\centering
\begin{tikzpicture}[
  >=Stealth,
  every node/.style={font=\small},
  goalstyle/.style={
    draw=black!55, very thick, rounded corners=5pt,
    fill=black!3, inner sep=10pt, align=center, text width=12.4cm
  },
  stdstyle/.style={
    draw=red!55!black, very thick, rounded corners=5pt,
    fill=red!7, inner sep=9pt, align=center, text width=5.5cm
  },
  rcstyle/.style={
    draw=green!55!black, very thick, rounded corners=5pt,
    fill=green!5, inner sep=9pt, align=center, text width=5.5cm
  },
  appstyle/.style={
    draw=blue!55!black, thick, rounded corners=5pt,
    fill=blue!5, inner sep=9pt, align=center, text width=5.5cm
  },
  midstyle/.style={
    draw=yellow!35!black, thick, rounded corners=4pt,
    fill=yellow!5, inner sep=6pt, align=center, text width=2.0cm,
    font=\scriptsize
  },
]

\node[goalstyle] (goal) at (0, 0) {%
  \textbf{Objective:}\; Estimate $\langle\hat{O}\rangle = \mathrm{tr}(\hat{O}\,\rho)$
  for a target quantum state $\rho$ \\[2pt]
  \;{\footnotesize For example, a thermal state $\rho\!\!=\!\!e^{-\beta H}/Z$\; or\;
  time-evolved state $\rho(t)$}
};


\node[stdstyle] (std) at (-3.65, -4.05) {%
  \textbf{Standard QMC} (Sec.~\ref{sec:QMC})\\[5pt]
  Fixed approximation $\widetilde{\mathcal{P}}$\\[4pt]
  $\displaystyle\langle\hat{O}\rangle
    \approx \mathbb{E}_{X\sim\widetilde{\mathcal{P}}}\!\bigl[F(\hat{O},X)\bigr]$\\[7pt]
  Systematic error:\; $\mathcal{O}(\epsilon)$\\[2pt]
  Cost:\; $C$
};

\node[rcstyle] (rc) at (3.65, -4.05) {%
  \textbf{Randomly Compiled QMC (this work) } (Sec.~\ref{sec:RC_QMC})\\[5pt]
  Average over approxs. $\{\widetilde{\mathcal{P}}_j,\,p_j\}$\\[4pt]
  $\displaystyle\langle\hat{O}\rangle
    \approx \mathbb{E}_{j\sim p_j}\!\Bigl[
      \mathbb{E}_{X\sim\widetilde{\mathcal{P}}_j}\!\bigl[F(\hat{O},X)\bigr]
    \Bigr]$\\[7pt]
  Systematic error:\; $\mathcal{O}(\epsilon^2)$\\[2pt]
  Cost:\; $\approx C$
};

\node[midstyle] (improve) at (0, -4.05) {%
  \textit{Randomized}\\[1pt]
  \textit{Compiling}\\[3pt]
  $\mathcal{O}(\epsilon)$\\[1pt]
  $\Big\downarrow$\\[1pt]
  $\mathcal{O}(\epsilon^2)$\\[3pt]
  same cost
};

\draw[->, thick] (goal.south) -- ++(0,-0.55) -| (std.north);
\draw[->, thick] (goal.south) -- ++(0,-0.55) -| (rc.north);

\draw[<->, dashed, thick, black!75!] (std.east)    -- (improve.west);
\draw[<->, dashed, thick, black!75!] (improve.east) -- (rc.west);


\node[appstyle] (pimc) at (-3.65, -9.1) {%
  \textbf{Randomly Compiled Path Integral Monte Carlo}\; (Sec.~\ref{sec:Application_PIMC})\\[4pt]
  \textit{Thermal states}: $\rho = e^{-\beta H}/Z$\\[5pt]
  Random compilation method:\; \textit{QDrift}~\cite{Campbell_2019} \\[4pt] 
  {Removes dependence on number of terms in the Hamiltonian}
};

\node[appstyle] (qtm) at (3.65, -9.1) {%
  \textbf{Randomly Compiled Quantum Trajectories Method}\; (Sec.~\ref{sec:Application_QTM})\\[4pt]
  \textit{Open system dynamics}: $\frac{d}{dt} \rho(t) = \mathcal{L}[\rho]$\\[5pt]
  Random compilation method:\; \textit{Randomly-corrected Trotterization}~\cite{Cho_2024} \\[4pt]
  {Doubles the effective Trotter order}
};

\draw[->, thick] (rc.south) -- ++(0,-0.55) -| (pimc.north);
\draw[->, thick] (rc.south) -- ++(0,-0.55) -| (qtm.north);



\end{tikzpicture}
\caption{%
  \textbf{Overview of the randomly compiled QMC (RC-QMC) framework.}
  Standard QMC algorithms estimate an observable $\langle\hat{O}\rangle$ by averaging an estimator function $F(\hat{O}, X)$ over a random variable $X$ sampled from a fixed approximate distribution $\widetilde{\mathcal{P}}$ (see Eq.~\eqref{eq:QMC_approx}), incurring systematic error $\mathcal{O}(\epsilon)$.
  RC-QMC instead averages over a collection of approximate distributions $\{\widetilde{\mathcal{P}}_j\}$ weighted by probabilities $p_j$ (see Eq.~\eqref{eq:RC_QMC}), which suppresses the systematic error to $\mathcal{O}(\epsilon^2)$ at essentially the same computational cost. This improvement is demonstrated in two settings: \textbf{Randomly Compiled Path Integral Monte Carlo} for thermal state estimation, where QDrift replaces the standard Trotter decomposition and eliminates the explicit dependence on the number of Hamiltonian terms, and \textbf{Randomly Compiled Quantum Trajectories Method} for open system dynamics, where randomly corrected Trotterization doubles the effective Trotter order. 
}
\label{fig:rcqmc_overview}
\end{figure*}

%% file: figures/trotter_rand_fig.tex
\begin{figure}[htbp]
\centering
\begin{tikzpicture}[
    block/.style={
        rectangle,
        draw=black, line width=0.6pt,
        text centered,
        inner xsep=6pt, inner ysep=4pt
    },
    stdblock/.style={block, fill=blue!10},
    corrblock/.style={block, fill=orange!25},
    font=\small
]

\node[anchor=west, font=\small\bfseries] at (0, 1.30)
    {(a) Standard Trotterization};

\node[stdblock, minimum width=7.0cm, minimum height=0.82cm] (std)
    at (3.50, 0.40) {$S_n(\Delta t)$};

\node[anchor=north, font=\small] at (std.south) [yshift=-2pt]
    {error $= \mathcal{O}(\Delta t^{n+1})$};

\node[anchor=west, font=\small\bfseries] at (0, -1.05)
    {(b) Randomly Corrected Trotterization};

\node[stdblock, minimum width=2.1cm, minimum height=1.10cm, align=center]
    (lhalf) at (1.05, -2.05)
    {$S_n(\Delta t/2)$};

\node[corrblock, minimum width=2.1cm, minimum height=1.10cm]
    (cj) at (3.50, -2.05)
    {$C_j(\Delta t)$};

\node[stdblock, minimum width=2.1cm, minimum height=1.10cm, align=center]
    (rhalf) at (5.95, -2.05)
    {$S_n(\Delta t/2)$};

\node[anchor=north, font=\scriptsize, text=black!65]
    at (cj.south) [yshift=-2pt]
    {sample $j \sim p_j$};

\node[anchor=north, font=\small] at (3.50, -3.05)
    {avg.\ error $= \mathcal{O}(\Delta t^{2n+2})$};

\end{tikzpicture}
\caption{%
    Comparison of standard and randomly corrected Trotterization for one
    time step $\Delta t$.
    \textbf{(a)} Standard order-$n$ Trotterization approximates
    $e^{-i\hat{H}\Delta t}$ by the product formula $S_n(\Delta t)$,
    incurring error $\mathcal{O}(\Delta t^{n+1})$.
    \textbf{(b)} Randomly-corrected Trotterization~\cite{Cho_2024} sandwiches
    a randomly sampled correction operation $C_j(\Delta t)$ between two
    half-steps of $S_n$.
    Each $C_j$ is drawn with probability $p_j$ from short evolutions under
    nested commutators of the Hamiltonian terms, designed so that the
    leading Trotter errors cancel upon averaging.
    The resulting average error $\mathcal{O}(\Delta t^{2n+2})$ matches that
    of an order-$(2n+1)$ formula at the computational cost of order~$n$.
}
\label{fig:trotter_rand}
\end{figure}

%% file: figures/rand_2nd_order_fig.tex
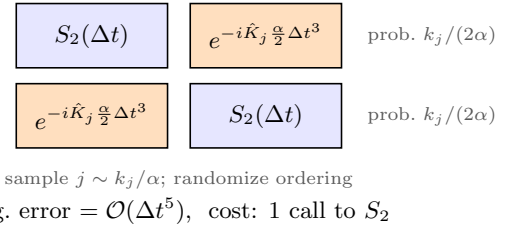
\begin{figure}[htbp]
\centering
\begin{tikzpicture}[
    block/.style={
        rectangle,
        draw=black, line width=0.6pt,
        text centered,
        inner xsep=6pt, inner ysep=4pt
    },
    stdblock/.style={block, fill=blue!10},
    corrblock/.style={block, fill=orange!25},
    font=\small
]

\node[anchor=west, font=\small\bfseries] at (0, 1.30)
    {(a) 2nd-Order Randomly-Corrected Trotter};

\node[stdblock, minimum width=1.80cm, minimum height=1.0cm, align=center]
    (lhalf) at (0.90, 0.30) {$S_2(\Delta t/2)$};

\node[corrblock, minimum width=2.40cm, minimum height=1.0cm, align=center]
    (cj) at (3.30, 0.30) {$C_j(\Delta t)$\\[2pt]
        {\scriptsize $= e^{-i\hat{K}_j \alpha (\Delta t/2)^3}$}};

\node[stdblock, minimum width=1.80cm, minimum height=1.0cm, align=center]
    (rhalf) at (5.70, 0.30) {$S_2(\Delta t/2)$};

\node[anchor=north, font=\scriptsize, text=black!65]
    at (cj.south) [yshift=-3pt] {sample $j \sim p_j$};

\node[anchor=north, font=\small] at (3.30, -0.70)
    {avg.\ error $= \mathcal{O}(\Delta t^{5})$, \ cost: 2 calls to $S_2$};

\node[anchor=west, font=\small\bfseries] at (0, -1.85)
    {(b) More Efficient 1-Call Variant};

\node[stdblock, minimum width=2.00cm, minimum height=0.85cm, align=center]
    (eS2a) at (2.15, -2.82) {$S_2(\Delta t)$};

\node[corrblock, minimum width=2.00cm, minimum height=0.85cm, align=center]
    (eCja) at (4.45, -2.82) {$e^{-i\hat{K}_j \frac{\alpha}{2} \Delta t^3}$};

\node[corrblock, minimum width=2.00cm, minimum height=0.85cm, align=center]
    (eCjb) at (2.15, -3.88) {$e^{-i\hat{K}_j \frac{\alpha}{2} \Delta t^3}$};

\node[stdblock, minimum width=2.00cm, minimum height=0.85cm, align=center]
    (eS2b) at (4.45, -3.88) {$S_2(\Delta t)$};

\node[anchor=west, font=\scriptsize, text=black!65]
    at (eCja.east) [xshift=6pt] {prob.\ $k_j/(2\alpha)$};
\node[anchor=west, font=\scriptsize, text=black!65]
    at (eS2b.east) [xshift=6pt] {prob.\ $k_j/(2\alpha)$};

\node[anchor=north, font=\scriptsize, text=black!65]
    at (3.30, -4.47) {sample $j \sim k_j/\alpha$;\ randomize ordering};

\node[anchor=north, font=\small] at (3.30, -4.84)
    {avg.\ error $= \mathcal{O}(\Delta t^{5})$, \ cost: 1 call to $S_2$};

\end{tikzpicture}
\caption{%
    Randomly-corrected 2nd-order Trotterization for one time step $\Delta t$,
    for $\hat{H} = \hat{A} + \hat{B}$.
    \textbf{(a)} The basic construction~\eqref{eq:RC_Trotter_order2}
    sandwiches a randomly sampled correction $C_j(\Delta t)$ between two
    half-steps of $S_2(\Delta t/2)$.
    Each correction $C_j(\Delta t) = e^{-i\hat{K}_j\alpha(\Delta t/2)^3}$
    implements evolution under one term from the nested-commutator operator
    $\hat{K} = -\tfrac{1}{12}[\hat{A}+2\hat{B},[\hat{A},\hat{B}]] = \sum_j k_j \hat{K}_j$,
    sampled with probability $p_j = k_j/\alpha$.
    On average the leading Trotter errors cancel, yielding
    $\mathcal{O}(\Delta t^5)$ error at the cost of two calls to $S_2$.
    \textbf{(b)} The efficient single-call variant~\eqref{eq:efficient_2nd_order_corrected_Trotter}
    reduces this to one call to $S_2(\Delta t)$ by additionally randomizing
    the \emph{ordering} of the Trotter step and the correction.
    Averaging over both orderings cancels the $\mathcal{O}(\Delta t^4)$
    errors by symmetry, again achieving $\mathcal{O}(\Delta t^5)$ error.
}
\label{fig:rand_2nd_order}
\end{figure}

%% file: figures/rand_pimc_error_fig.tex
\begin{figure}[htbp]
\centering
\begin{tikzpicture}[font=\small]


\fill[red!15] (0.00, 1.25) rectangle (0.35, 1.50);
\draw[black, line width=0.6pt] (0.00, 1.25) rectangle (0.35, 1.50);
\node[anchor=west, font=\scriptsize] at (0.42, 1.375)
    {QDrift step $\mathcal{A}_j$};

\fill[green!15] (0.00, 0.75) rectangle (0.35, 1.00);
\draw[black, line width=0.6pt] (0.00, 0.75) rectangle (0.35, 1.00);
\node[anchor=west, font=\scriptsize] at (0.42, 0.875)
    {exact imaginary-time step $e^{-\beta\hat{H}/r}$};

\fill[red!15] (0, -0.38) rectangle (3.8,  0.38);
\draw[black, line width=0.6pt] (0, -0.38) rectangle (3.8, 0.38);

\node[anchor=east, align=right] at (-0.15, 0.00)
    {$\sigma'_0 = \sigma'$\\[2pt]{\scriptsize\color{black!55}(proposed)}};

\draw[->, thick] (1.9, -0.48) -- (1.9, -1.52);
\node[fill=white, inner sep=1.5pt, font=\scriptsize] at (1.9, -1.00)
    {$\bigl\|\sigma'_n\!-\!\sigma'_{n+1}\bigr\|_1
      \leq \mathcal{O}\!\left(\tfrac{\lambda^2\beta^2}{r^2}\right)$};

\fill[red!15] (0,    -2.38) rectangle (2.85, -1.62);
\fill[green!15]   (2.85, -2.38) rectangle (3.8,  -1.62);
\draw[black, line width=0.6pt] (0, -2.38) rectangle (3.8, -1.62);
\draw[black, line width=0.6pt] (2.85, -2.38) -- (2.85, -1.62);

\node[anchor=east] at (-0.15, -2.00) {$\sigma'_1$};

\draw[->, thick] (1.9, -2.48) -- (1.9, -3.52);

\fill[red!15] (0,   -4.38) rectangle (1.90, -3.62);
\fill[green!15]   (1.90,-4.38) rectangle (3.80, -3.62);
\draw[black, line width=0.6pt] (0, -4.38) rectangle (3.8, -3.62);
\draw[black, line width=0.6pt] (1.90, -4.38) -- (1.90, -3.62);

\node[anchor=east] at (-0.15, -4.00) {$\sigma'_2$};

\draw[->, thick] (1.9, -4.48) -- (1.9, -5.10);
\node[font=\normalsize] at (1.9, -5.45) {$\vdots$};
\draw[->, thick] (1.9, -5.82) -- (1.9, -6.52);

\fill[green!15] (0, -7.38) rectangle (3.8, -6.62);
\draw[black, line width=0.6pt] (0, -7.38) rectangle (3.8, -6.62);

\node[anchor=east, align=right] at (-0.15, -7.00)
    {$\sigma'_{r/2}$\\[2pt]{\scriptsize\color{black!55}$= e^{-\beta\hat{H}}/Z$}};

\end{tikzpicture}
\caption{%
    Interpolating sequence $\{\sigma'_n\}_{n=0}^{r/2}$ used in the proof of
    Theorem~\ref{thm:ErrorBound_sigma_prime}.
    Each $\sigma'_n$ is obtained from $\sigma'_{n-1}$ by replacing one
    QDrift step with an exact imaginary-time step $e^{-\beta\hat{H}/r}$;
    the orange (left) portion of each bar shows remaining QDrift steps
    and the blue (right) portion shows exact steps already in place.
    Adjacent states differ by a single replacement, incurring a per-step
    1-norm error $\mathcal{O}(\lambda^2\beta^2/r^2)$.
    Summing over $r/2$ steps via the triangle inequality yields the
    overall bound
    $\|\sigma' - e^{-\beta\hat{H}}/Z\|_1 = \mathcal{O}(\lambda^2\beta^2/r)$.
}
\label{fig:rand_pimc_error}
\end{figure}
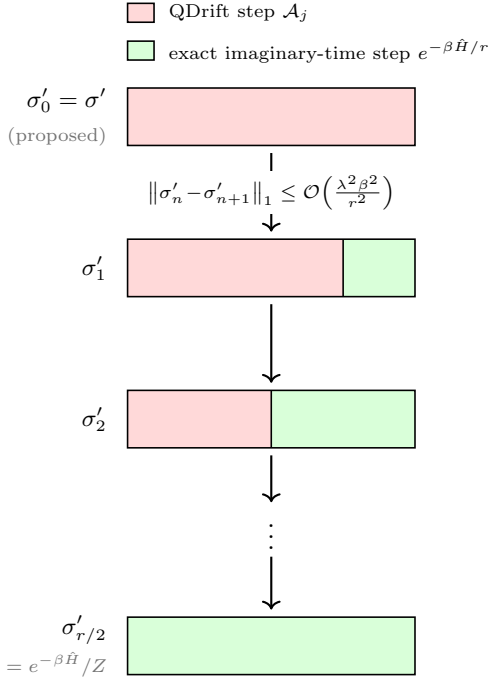